\documentclass[11pt]{article}
\usepackage[margin=1in]{geometry}

\usepackage{booktabs} % For formal tables

\usepackage[utf8]{inputenc}
\usepackage{authblk}
\usepackage{amsmath, amssymb, amsthm, thmtools, amsfonts, bm, bbm, thm-restate}
\usepackage{enumitem}
\usepackage{cite}
\usepackage{color}              % Need the color package
\usepackage{asymptote}
\usepackage{graphicx}
\usepackage{todonotes}
\usepackage{multirow}
\usepackage{enumitem}

\usepackage{comment}
\usepackage{dsfont}
\usepackage{bigstrut}
\usepackage{caption}
\usepackage{subcaption}
\usepackage{tcolorbox}
\usepackage{algorithm}
\usepackage[noend]{algpseudocode}
\usepackage{tikz}
\usepackage{cancel}
\usepackage{tikz-cd}
\usetikzlibrary{positioning,arrows.meta,decorations.pathreplacing,calc}
\usepackage{xcolor}

\definecolor{BrickRed}{rgb}{0.8,0.25,0.33}
\definecolor{darkgreen}{HTML}{06402B}
\definecolor{navyblue}{HTML}{000080}

\usepackage[pagebackref,colorlinks,citecolor=blue,linkcolor=BrickRed]{hyperref}
\usepackage[framemethod=tikz]{mdframed}
\usepackage{nicefrac}
\usepackage{bigints}
\usepackage{pifont}% http://ctan.org/pkg/pifont

\usepackage{cleveref}
\crefformat{paragraph}{#2Case~#1#3}

\newcommand{\cvxle}{\le_{\textup{\textsf{cv}}}}

\newcommand{\pr}[1]{{\rm Pr} \left[ #1 \right]}

\newcommand{\bs}[1]{\boldsymbol{#1}}

\newcommand{\ex}[1]{{\mathbb E} \left[ #1 \right]}
\newcommand{\expar}[1]{{\mathbb E} [ #1 ]}

\newenvironment{wrapper}[1]
{
	\smallskip
	\begin{center}
		\begin{minipage}{\linewidth}
			\begin{mdframed}[hidealllines=true, backgroundcolor=gray!20, leftmargin=0cm,innerleftmargin=0.35cm,innerrightmargin=0.35cm,innertopmargin=0.375cm,innerbottommargin=0.375cm,roundcorner=10pt]
				#1}
			{\end{mdframed}
		\end{minipage}
	\end{center}
	\smallskip
}

\newcommand{\prtwo}[2]{{\rm Pr}_{#1} \left[ #2 \right]}

\newcommand{\extwo}[2]{{\mathbb E}_{#1} \left[ #2 \right]}

\theoremstyle{plain}
\newtheorem{thm}{Theorem}[section]

\newtheorem{cor}[thm]{Corollary}
\newtheorem{prop}[thm]{Proposition}

\newtheorem{fact}[thm]{Fact}

\newtheorem{lemma}[thm]{Lemma}

\newtheorem{claim}[thm]{Claim}
\newtheorem{definition}[thm]{Definition}
\newtheorem{obs}[thm]{Observation}

\newtheorem{exm}[thm]{Example}

\crefname{thm}{Theorem}{theorems}
\crefname{cla}{Claim}{claims}
\crefname{lem}{Lemma}{lemmas}
\crefname{fact}{Fact}{facts}

\newcommand{\E}{\mathbb{E}}

\newcommand{\eps}{\epsilon}

\newcommand{\topp}{\mathrm{TOP}}
\newcommand{\bott}{\mathrm{BOT}}

\newcommand{\bb}{\mathbb}

\newcommand{\Var}{\textup{Var}}

\usepackage{color}              % Need the color package
\usepackage[suppress]{color-edits}
\addauthor{Ali}{red}
\addauthor{Amin}{brown}
\addauthor{Alireza}{blue}
\addauthor{Tristan}{green}

\newcommand{\optoff}{\mathrm{OPT}_\mathrm{off}}
\newcommand{\opton}{\mathrm{OPT}_\mathrm{on}}

\allowdisplaybreaks

\begin{document}

\title{Improved Approximations for Stationary Bipartite Matching: Beyond Probabilistic Independence}
% \author{\vspace{-1.5cm}}
\author[1]{Alireza AmaniHamedani\thanks{Email: {aamanihamedani@london.edu};}}
\affil[1]{London Business School}
\author[2]{Ali Aouad\thanks{Email: \href{mailto:maouad@mit.edu}{maouad@mit.edu}; This work was partially supported by the UKRI Engineering and Physical Sciences Research Council [EP/Y003721/1].}}
\affil[2]{Massachusetts Institute of Technology}
\author[3]{Tristan Pollner\thanks{Email: {tpollner@stanford.edu};}}
\affil[3]{Stanford University}
\author[3]{Amin Saberi\thanks{Email: {saberi@stanford.edu}.}}
% \affil[1]{Stanford University}
\date{}

\maketitle
\pagenumbering{arabic}
\begin{abstract}

We study stationary online bipartite matching, where both types of nodes---offline and online---arrive according to Poisson processes. Offline nodes wait to be matched for some random time, determined by an exponential distribution, while online nodes need to be matched immediately. This model captures scenarios such as  deceased organ donation and time-sensitive task assignments, where there is an inflow of patients and workers (offline nodes) with limited patience, while organs and tasks (online nodes) must be assigned upon arrival.

We present an efficient online algorithm that achieves a $(1-1/e+\delta)$-approximation to the optimal online policy's reward for a constant $\delta > 0$, simplifying and improving previous work by \cite{aouad2022dynamic}.  Our solution combines recent online matching techniques, particularly pivotal sampling, which enables correlated rounding of tighter linear programming approximations, and a greedy-like algorithm. A key technical component is the analysis of a  stochastic process that exploits subtle correlations between offline nodes, using renewal theory. A byproduct of our result is an improvement to the best-known \emph{competitive ratio}---that compares an algorithm's performance to the optimal offline policy---via a $(1-1/\sqrt{e} + \eta)$-competitive algorithm for a universal constant $\eta > 0$, advancing the results of \cite{patel2024combinatorial}.

% We study online bipartite matching under a stationary setting capturing uncertain agent waiting times. In particular, offline nodes arrive in the system each according to a Poisson process, and must be matched within some exponentially-distributed wait time (online nodes arrive in the system via a Poisson process and must be matched immediately). This model captures natural correlations between the neighborhood distributions of online nodes arriving in close succession. \Alirezacomment{I am not sure about emphasizing on the correlation. It could open a can of worms.}

% Our main result is an efficient online algorithm which achieves a $(1-1/e+\delta)$-approximation to the reward of the optimal online policy, improving the best-known ratio by $\delta > 0$. We also provide a significantly simpler proof of the $1-1/e$ bound provided in prior work \cite{aouad2022dynamic}.

% Breaking the $(1-1/e)$-approximation level requires a tighter LP relaxation and a correlated  LP-rounding using pivotal sampling. The main idea in the analysis is a novel stochastic process that captures weak correlations between offline nodes, going beyond the natural analysis with probabilistically independent nodes. As an additional byproduct of our approach for analyzing correlated queues, we also improve upon the best-known \emph{competitive ratio}---that compares an algorithm against the optimal offline policy---to obtain a $(1-1/\sqrt{e} + \eta)$ performance guarantee, for a constant $\eta > 0$ \cite{patel2024combinatorial}.

\end{abstract}

% \Alicomment{Let's agree on a paper outline}

%\tableofcontents 

\newpage
\section{Introduction}

Online bipartite matching has a storied history in computer science. The foundational work of Karp, Vazirani, and Vazirani introduced the problem along with the optimal RANKING algorithm, achieving a competitive ratio of $1-1/e$ \cite{karp1990optimal}.  Karp et al.'s work catalyzed extensive future research improving our understanding of online matching. Extensions considered have included guarantees for vertex-/edge-weighted graphs \cite{aggarwal2011online, fahrbach2020edge, gao2021improved, blanc2022multiway}, stochastic arrivals \cite{feldman2009online, manshadi2012online, jaillet2013online, huang2021online, tang2022fractional, huang2022power, chen2024stochastic}, post-allocation stochasticity \cite{mehta2012online, goyal2023online, huang2023onlinestochastic, huang2024online}, online matching mechanisms \cite{chawla2010multi, feldman2015combinatorial, kleinberg2019matroid}, and settings incorporating partial information or predictions \cite{jin2022online, antoniadis2023secretary}. When offline nodes are given upfront and online arrivals are sampled from time-dependent distributions, edge-weighted online matching was studied from the perspective of ``prophet inequalities''---namely, against the offline optimum. Numerous elegant $0.5$-{\em competitive} online algorithms are known \cite{feldman2015combinatorial,ezra2022prophets,aouad2023nonparametric}. Approximations to the online benchmark were introduced by \cite{papadimitriou2023online}, leading to ``philosopher inequalities'' \cite{saberi2021greedy, braverman2022max, naor2023dependentroundingarxiv, braverman2024new} improving the {\em approximation ratio} to $0.67$.

%These theoretical results, beautiful in their own right, are closely connected to diverse domains such as online advertising \cite{mehta2007adwords}, resource allocation, ride-sharing platforms \cite{dickerson2021reusable}, and organ allotment \cite{dickerson2012dynamic}. 

The theoretical results in matching have found practical applications in online advertising \cite{mehta2007adwords}, ride-sharing \cite{dickerson2021reusable,ashlagi2019edge}, and organ transplantation \cite{dickerson2012dynamic,ashlagi2021kidney}; see \cite{huang2024online} for a brief survey and \cite[Chap.~5]{echenique2023online} for a comprehensive treatment. These markets are often characterized by continuous agent turnover: 
%In real-world matching markets, there is often a continuous inflow and outflow of agents, which is not controlled by the matchmaker. %\Alidelete{For example, new passengers continually request rides with a risk of dropout if they are not matched to a driver in time \Alicomment{should it be flipped? drivers waiting}.  Similarly,} 
in the gig economy, a stream of arriving workers must be assigned to time-sensitive tasks within limited time windows. For example, drivers join and exit ride-hailing platforms at will. Similarly, organ transplant candidates arrive continuously and may exit the system unmatched. Many other resource allocation markets (like inventory management for a food bank) exhibit this dynamic structure, which affects the
“market thickness”—the availability of compatible market participants for matching.

%The theoretical results in matching have found practical applications in online advertising \cite{mehta2007adwords}, ride-sharing \cite{dickerson2021reusable, ashlagiwindowed}, and organ transplantation \cite{ashlagi-roth,dickerson2012dynamic}; see \cite{matching-market-book} for a comprehensive survey. These markets are characterized by continuous agent turnover: gig economy workers join and leave platforms at will, organ transplant candidates continuously and may exit the system unmatched, and similar dynamics appear in resource allocation systems like food banks. The key challenge lies in finding  within constrained and uncertain temporal windows.

This continuous flow of new offline agents and their limited patience are crucial considerations that are overlooked by the classic online matching models. A notable recent line of work considers the problem under adversarial arrivals and deadlines~\cite{huang2018match, huang2019tight, huang2020fully, eckl2021stronger, tang2022improved}, or assumes that agents are available for a fixed ``window'' of time~\cite{ashlagi2019edge}. 
In the stochastic setting, a recently introduced {\em stationary} formulation of the online matching problem models the arrivals and departures as a continuous-time process, accounting for heterogeneity across agents.
Aouad and Sarita{\c{c}} gave a $(1-1/e)$-approximation to the optimal online algorithm \cite{aouad2022dynamic}. Despite extensive efforts, and further simplification of their algorithm and analysis, improving upon the $(1-1/e)$ approximation ratio has been elusive. In the special case where there is a {\em single type} of offline node, an improvement was obtained by \cite{kessel2022stationary}, who gave a $0.656$-approximation. Against the offline benchmark, the best-known competitive ratio is $1-1/\sqrt{e}$, which is a recent improvement due to \cite{patel2024combinatorial}, relative to the original result by~\cite{collina2020dynamic}. 

The main contribution of this paper is an efficient online algorithm that achieves a $(1-1/e+\delta)$-approximation to the optimal online policy's reward, where $\delta>0$ is a  constant bounded away from zero.  
%Our solution combines recent online matching techniques, particularly pivotal sampling, which enables correlated rounding of tighter linear programming approximations with the analysis of a  stochastic process that exploits subtle correlations between offline nodes, using renewal theory.  
A byproduct of our result is an improvement to the best-known \emph{competitive ratio}---that compares an algorithm to the optimal offline policy---via a $(1-1/\sqrt{e} + \eta)$-competitive algorithm for a universal constant $\eta > 0$, advancing and simplifying the results of \cite{patel2024combinatorial}. 

% In these domains, the limited patience of users is a crucial consideration that is often overlooked by classical models. Offline nodes, whether drivers or patients, might depart from the system after a certain period if not matched. This temporal dimension of the problem has been studied mainly in the queueing and probability literature, often referred to as the ``stationary setting'' due to the memoryless nature of arrivals and departures in steady-state analysis.

\paragraph{Problem formulation: Online stationary bipartite matching.} We are given offline types $I$ and online types $J$. Offline nodes of type $i \in I$ arrive at rate $\lambda_i$ and each one departs after time $\text{Exp}(\mu_i)$, independently from others; online nodes of type $j$ arrive at rate $\gamma_j$. Upon arrival of a type-$j$ node, we must immediately and irrevocably decide how to match it (if at all). Matching to a present and unmatched offline node of type $i$ gains some specified reward $r_{i,j} \ge 0$ while choosing not to match gains no reward. Naturally, each arriving node can be matched at most once. Our goal is to design an online matching policy maximizing the expected long-term average reward, i.e., $$\textsf{Gain(ALG)} := \liminf_{t \rightarrow \infty} \frac{\textsf{ALG}[0,t]}{t} $$ where $\textsf{ALG}[0,t]$ denotes the reward $\textsf{ALG}$ accrues during time $[0,t]$. 

Our main performance measure is the \emph{approximation ratio}, which computes the ratio between the algorithm's performance $\textsf{Gain(ALG)}$ and  that of the \emph{optimal online} algorithm ($\opton$), which is the solution of a dynamic program solving Bellman's equations \cite{bertsekas2012dynamic, puterman2014markov}. We also consider the measure of \emph{competitive ratio}, which computes an analogous ratio with respect to the \emph{optimal offline} algorithm ($\optoff$), that has exact knowledge about all arrival and departure times a priori. % Denoting this optimal algorithm by $\optoff$, we define
%\[
%    \textsf{CR(ALG)} := \frac{\textsf{Gain(ALG)}}{\textsf{Gain(}\optoff\textsf{)}} \ .
%\] %Letting $\opton$ be this optimal online algorithm, AR of our algorithm is formally defined to be
%\[
%    \textsf{AR(ALG)} := \frac{\textsf{Gain(ALG)}}{\textsf{Gain(}\opton\textsf{)}} \ .
%\]

% \paragraph{Prior work.} 

\begin{comment}
Our main performance measure is the \emph{approximation ratio} (AR), which compares the performance of our algorithm versus that of the \emph{optimal online} algorithm, which is the solution of a dynamic program. Namely, this dynamic program keeps track of the set of available offline nodes and solves the ``Bellman's equations'' \cite{bertsekas2012dynamic, puterman2014markov}. Letting $\opton$ be this optimal online algorithm, AR of our algorithm is formally defined to be
\[
    \textsf{AR(ALG)} := \frac{\textsf{Gain(ALG)}}{\textsf{Gain(}\opton\textsf{)}} \ .
\]
We also consider another measure, \emph{competitive ratio} (CR), which compares an algorithm with the \emph{optimal offline} algorithm, that has exact knowledge about all arrival and abandonment times a priori. Denoting this optimal algorithm by $\optoff$, we define
\[
    \textsf{CR(ALG)} := \frac{\textsf{Gain(ALG)}}{\textsf{Gain(}\optoff\textsf{)}} \ .
\]
\end{comment}

\subsection{Our Results}

Our main result is an algorithm breaking the $1-1/e$ barrier for online stationary matching.

\begin{wrapper}
\begin{restatable}{theorem}{mainthm} \label{thm:main}
    There exists a polynomial-time algorithm for the online stationary matching problem that achieves expected average reward at least a $(1-1/e+\delta)$-factor of the optimal \underline{online} algorithm, for some universal constant $\delta > 0$. 
\end{restatable}
\end{wrapper}

As a byproduct, we obtain an algorithm and analysis that improves on the best-known competitive ratio, and additionally simplifies the proof of the existing bound. In particular, we provide a polynomial-time algorithm for the online stationary matching problem that achieves expected average reward at least a $(1-1/\sqrt{e}+\eta)$-factor of the optimal  {offline} algorithm, for some universal constant $\eta > 0$, improving the results of \cite{patel2024combinatorial}.  

As in \cite[Thm.~2]{aouad2022dynamic}, we can additionally extend our main result to get  improved approximation guarantees in settings where both sides of the graph have limited patience or where the graph is non-bipartite and all types are partially patient. Such extensions are deferred to Appendix~\ref{app:extensions}.

\begin{comment}
    \paragraph{Extension to other graph structures.} We note that every algorithm for the online stationary matching problem yields similar guarantees for two other settings: (i) both sides of the graph have limited patience, (ii) the graph is non-bipartite and all types are partially patient. For example, for the non-bipartite graph, consider an instance ${\cal I}$; we construct an alternative instance $\tilde{\cal I}$ by labeling every arriving node online or offline with the same probability. We claim that the optimal policy in $\tilde{\cal I}$ obtains a stationary reward at least $1/4$ of the optimal in ${\cal I}$. Indeed, the (potentially suboptimal) policy that mimics the optimal policy of ${\cal I}$ can successfully make each of its matches with probability $1/4$. 

In light of this observation, \Cref{thm:main} implies a $(1-1/e + \delta)/4$-approximation for the case of non-bipartite graphs, improving the $(1-1/e)/4$ bound of \cite{aouad2022dynamic}. The case for a bipartite graph with partial patience on both sides is similar; there, our approach gives a $(1-1/e + \delta)/2$-approximation, strengthening the $(1-1/e)/2$ result of \cite{aouad2022dynamic}.
\end{comment}

%\todo{Put a more formal discussion of the point about other graphs somewhere else (maybe when we retstate Theorem 1.1)?}
%\Alicomment{Good idea}
%\Alirezacomment{let's iterate on this to make it better. I feel like this is too much for here. maybe we should just mention the claim and discuss it in the appendix?
%{
%\bf Amin: agreed. Move to the end of intro?}}

\subsection{Our Techniques}
We introduce several new techniques to achieve an improved approximation ratio for online stationary matching.
% , with a careful design and analysis of correlations between offline nodes.

\paragraph{A tightened LP relaxation.}  At its core, our algorithm employs a polynomial-size LP relaxation of the optimal online policy. Extending ideas from prior literature \cite{huang2021online, huang2022power, kessel2022stationary}, we identify a new exponential family of constraints bounding the availability of offline types. %, which tighten existing LP relaxations. 

\paragraph{Pivotal sampling for stationary matching.} To round our LP online, upon the arrival of an online type $j \in J$, each available offline node of type $i\in I$ proposes to match with node $j$ with a ``proposal probability'' $p_{i,j}$, computed based an optimal solution to our LP relaxation. Drawing a connection with the discrete online matching literature \cite{braverman2024new}, we correlate these proposals using pivotal sampling, and then match to the proposing node with maximum reward.

We provide a simple and concise proof that this algorithm is $(1-1/e)$-approximate by coupling the random evolution of offline nodes with a simple Markov chain, where offline nodes are completely independent. In particular, we show that the true number of offline nodes stochastically dominates a stochastic process comprised of {\em independent Markov chains}, providing a new short proof of this baseline guarantee. This analysis holds even if the proposals are drawn independently. 

Independent proposals cannot beat $1-1/e$. In contrast, we show that pivotal sampling achieves a better approximation in several cases. In particular, when the proposal probabilities are bounded away from $1$, we leverage negative correlation of pivotal sampling for an improved guarantee. We also identify straightforward improvements when the LP solution does not saturate the capacity constraints of online nodes, or when rewards are not overwhelmingly {\em vertex-weighted}, i.e., the LP solution does not ``concentrate'' on edges of nearly equal weights.  

% \paragraph{The Balanced Greedy algorithm.} 
Via our tightened LP relaxation, we argue that the remaining cases are extremely structured instances, which we term {\em vertex-weighted highly-connected}. Here, breaching the $(1-1/e)$-approximation ratio requires a second algorithm, which we call {\em Balanced Greedy}. After removing edges and node types with negligible contributions to the LP solution, Balanced Greedy matches each arriving online node myopically, using a form of ``balanced'' tie-breaking between offline nodes (instead of prioritizing those with larger rewards). To this end, we split each offline type into ``top'' and ``bottom'' copies uniformly at random, and prioritize matches to top copies when possible. 

  \paragraph{ Key technical component.}
  %Analysis of weakly correlated Markov chains.}}While we draw algorithmic inspiration from discrete-time techniques, particularly in our use of pivotal sampling, the stationary nature of our problem demands new ideas. 
Unlike discrete settings, where offline nodes' availability progresses from empty to full over a known time horizon, the  stationary nature of this problem prevents us from beating the $1-1/e$ bound using independent availability (as in, e.g., \cite{naor2023dependentroundingarxiv, braverman2024new}). % This fundamental difference requires us to develop new tools for analyzing subtle correlations in the availability of the offline nodes.%.\footnote{We note in particular the results \cite{naor2023dependentroundingarxiv, braverman2024new} do not use strict negative correlation beyond independence of the offline nodes' availability to beat $1-1/e$, instead relying on an averaging argument or time-dependent scaling.} 

\begin{comment}
The core of our improvement stems from demonstrating that the deprioritization of  ``bottom nodes'' when running Balanced Greedy leads to a higher   probability of availability in steady-state. We prove this constructively via a new \emph{weakly correlated} Markov chain. Unlike in the independent case, we cannot solve for the stationary distribution of this Markov chain exactly, and require  probabilistic techniques to bound the dependence across nodes. 

In the ``vertex-weighted highly-connected'' instances, we would like to show that an arriving online node of type $j$ has no available neighbor with probability $(1/e - \Omega(1))$. We first bound the probability that $j$ sees no available neighboring top node via the simpler independent Markov chains. Conditioned on the top nodes being unavailable, we analyze the joint probability that $j$'s bottom neighbors are unavailable by analyzing 

Our task is then to establish an improved availability rate for each individual bottom type $i$ connected to $j$, conditional on $j$ having no available top nodes, when it arrives at time $t$. The difficult case is when $i$ is connected to many online nodes $k$ whose neighborhoods look very ``similar'' to that of $j$---if such $k$-s were blocked from matching with top nodes due to our conditioning, we may not get any boost of $i$'s availability. We handle this case via techniques from renewal theory, in particular proving that the rapid mixing of top nodes ``washes out'' this conditioning when looking backwards from $t$. 
\end{comment}

Instead, the core of our improvement stems from demonstrating that the deprioritization of  ``bottom nodes'' when running Balanced Greedy leads to a higher probability of availability in steady-state. Intuitively, the prioritization in Balanced Greedy means that bottom types are not depleted continuously over time---in the intervals where top types are present, the bottom nodes are not matched. 

Formalizing this idea via a new process with \emph{weakly correlated  Markov chains}, we show that an arriving online node of type $j$ has no available neighbor with probability $1/e - \delta'$ for some constant $\delta'>0$. We first bound the probability that $j$ sees no available neighboring top node via the simpler independent Markov chains. Conditioned on its top neighbors being unavailable, we analyze the joint probability that $j$'s bottom neighbors are unavailable by exploiting independence of arrivals/departures across types, and applying the Hardy-Littlewood inequality to bound the remaining correlation. The central challenge is to analyze the conditional probability that an individual bottom neighbor $i \sim j$ is unavailable; we do so by applying  techniques from renewal theory and an analysis of mixing times.

 \subsection{Further related work}

\paragraph{Online matching with announced departures.} In the unweighted case, a beautiful recent line of work \cite{huang2018match, huang2019tight, huang2020fully, eckl2021stronger, tang2022improved} considers the problem under adversarial arrivals and deadlines, giving algorithms significantly beating the $0.5$ baseline of Greedy and showing the $1-1/e$ guarantee of \cite{karp1990optimal} is unattainable. It is even possible to slightly beat greedy for general vertex arrivals when arriving nodes must be matched immediately or to later arrivals \cite{wang2015two, gamlath2019beating}. In settings with edge-weights, results on ``windowed matching'' study policies when agents stay in a marketplace for a fixed number of time periods, or for certain i.i.d. random durations~\cite{ashlagi2019edge}. 

\paragraph{Online stationary matching.}  
% A recently introduced ``stationary'' formulation of the online matching problem models departures without announcements and additionally accounts for heterogeneity across agents.
 The stationary (or dynamic) problem was studied by~\cite{collina2020dynamic,aouad2022dynamic}.\footnote{There is extensive work on the control of matching queues in the applied probability literature, with a focus on stability criteria~\cite{buvsic2013stability, mairesse2016stability}, implicit solutions for the steady-state distribution \cite{moyal2021product}, waiting times~\cite{cadas2022analysis}, and large-market approximations~\cite{ozkan2020dynamic} or heavy-traffic~\cite{varma2021transportation}. Notably, however, the combinatorial aspect of the problem, and the design of policies through the lens of approximation algorithms, has received less attention. } Here, the central planner is not notified before agents depart. More recently, \cite{patel2024combinatorial} develop a generalization of the stationary setting to combinatorial allocation, and leverage a reduction to offline contention resolution schemes to obtain constant-factor competitive algorithms. \cite{li2023fully} consider more general distribution of departure times. While it is unknown if the stationary matching is NP-hard  or APX-hard, \cite{amani2024adaptive} devise a fully polynomial-time approximation scheme in the special case where the number of offline nodes is a constant.

\section{Pivotal Sampling Applied to a Tightened LP Relaxation} \label{sec:algo}

\subsection{A Tightened LP Relaxation} \label{subsec:tlp}

We use the following LP relaxation that upper bounds the stationary gain of the optimal offline algorithm. It has variables $\{x_{i,j}\}$ corresponding to the match rate of $i \in I$ and $j \in J$, as well as variables $\{x_{i,a}\}$ corresponding to the rate at which each $i \in I$ abandons without being matched. In addition to a natural flow constraint, it uses a new \underline{tightened} constraint on the flow rate into any subset of offline types, which generalizes the constraint for subsets of size one observed by \cite{kessel2022stationary}.

\begin{align}
	\nonumber  \max \quad &  \sum_{i \in I} \sum_{j \in J} r_{i,j} \cdot x_{i,j} && \tag{TLP$_{\text{on}}$} \label{TLPon} \\
	\textrm{s.t.} \quad  & x_{i,a} + \sum_j x_{i,j} = \lambda_i \ , && \forall i \in I   \label{eqn:tightOfflineFlow}\\
	&  \sum_{i \in H} x_{i,j} \le \gamma_j \cdot \left(1 - \exp\left(-\sum_{i \in H}  \lambda_i / \mu_i \right)\right) \ , && \forall j \in J, \forall H  \subseteq I\label{eqn:tightOnlineFlow} \\
 	&  x_{i,j}/\gamma_j \le x_{i,a} / \mu_i \ , &&\forall i \in I, \forall j \in J \label{eqn:tightOnlineConstraint} \\
 & x_{i,j}, x_{i,a} \ge 0 \ . && \forall i \in I, \forall j \in J \label{eqn:nonnegativity}
\end{align}

Constraints \eqref{eqn:tightOfflineFlow} and \eqref{eqn:tightOnlineConstraint} are already known from prior work (c.f. \cite{aouad2022dynamic}). The rationale behind Constraint \eqref{eqn:tightOnlineFlow} is that in order to match an online type $j$ and an offline type in set $H$, it is necessary that at least one offline type $i\in H$ is {\em present}. This notion refers to an item which has arrived but has not yet departed (regardless of whether or not it has been matched). Then, since this event is independent for different offline types and the number of present items of type $i$ follows from a Poisson distribution, this probability corresponds to the right-hand side of Constraint \eqref{eqn:tightOnlineFlow}. The following claim formalizes this intuition (refer to \Cref{app:claimlprelaxation} for a proof). 
\begin{restatable}{claim}{claimlprelaxation} \label{claim:lprelaxation} Let $\textup{OPT}\eqref{TLPon}$ be the optimal value of \eqref{TLPon}. Then, $\textup{OPT}\eqref{TLPon} \geq \opton$.
\end{restatable}

We note that without the new tightened family of constraints, the value of the optimum online may be upper bounded by a $(1-1/e+o(1))$-fraction of the LP optimum; see \Cref{app:needTighteningExample}.

\paragraph{Computational tractability.}
Although \eqref{TLPon} has exponentially many constraints, it can be solved in polynomial time. The proof is nearly identical to Theorem 5 of \cite{huang2021online}; see  \Cref{app:proofTLPpolysolvable}. 

\begin{restatable}{claim}{TLPpolysolvable} \label{TLPpolysolvable}
    \eqref{TLPon} is polynomial-time solvable. 
\end{restatable}

\subsection{Algorithm with Pivotal Sampling}

Pivotal sampling \cite{srinivasan2001distributions, gandhi2006dependent} is a classic algorithm for sampling according to a provided list of marginal probabilities with negative correlation. For convenience, we restate some of its properties used in our algorithm and analysis.

\begin{thm}[as in \cite{srinivasan2001distributions}] \label{def:pivotalsampling}
	The pivotal sampling algorithm with input $(x_i)_{i=1}^n$ produces (in polynomial time) a random subset of $[n]$, denoted $\textup{\textsf{PS}}(\vec{x})$, with the following properties:
		\begin{enumerate}[label=(P{{\arabic*}})]%uses package enumitem
    \item \textbf{Marginals:} $\Pr[ i \in \textup{\textsf{PS}}(\vec{x})] = x_i$ for all $i \in [n]$.\label{level-set:marginals}
    \item \textbf{Prefix property:} $\Pr[ |\textup{\textsf{PS}}(\vec{x}) \cap [k] |\geq 1] = \min(1,\,\sum_{i\leq k} x_i)$ for all $k \in [1,n]$. \label{level-set:prefix}
\end{enumerate}
\end{thm}

We next present our algorithm. It first solves \eqref{TLPon}, and using the optimal solution computes a {\em proposal probability} $p_{i,j}$ for each $i \in I$ and $j \in J$. Upon the arrival of each online type, to match it we use pivotal sampling with the relevant proposal probabilities for each {\em available} offline node, which is both present and unmatched.

% For both benchmarks, the algorithm is identical other than the relevant LP relaxation that is solved (in \Cref{line:solveLP}) and the formula for $p_{i,j}$ (in \Cref{line:setPij}). For convenience we hence use \Cref{alg:corrprop} to describe the algorithm for the optimum online benchmark, marking in the first two lines the small changes needed for the optimum offline benchmark. 

\begin{algorithm}[H]
	\caption{Correlated Proposals  for \eqref{TLPon}}
	\label{alg:corrprop}
	\begin{algorithmic}[1]
 \State Solve \eqref{TLPon} for $\{x_{i,j}\}$  \label{line:solveLP}
 \State $p_{i,j} \gets \frac{x_{i,j}/\gamma_j}{x_{i,a} / \mu_i }$ \label{line:setPij}
 \medskip
    
		\ForAll{timesteps where an online $j$ arrives} \label{line:loop-start}
        \State $\vec{p} \gets $ vector of $p_{i,j}$'s for each available offline $i$, sorted by decreasing $r_{i,j}$ \label{lin:inside-loop-start}
        \State $\textsf{Proposers} \gets \textsf{PS}(\vec{p} )$ \label{lin:PS} \Comment{$\textsf{PS}(\cdot)$ denotes \emph{pivotal sampling}, as in \Cref{def:pivotalsampling}}
		\State If \textsf{Proposers} is non-empty, match $j$ to a proposing node of maximum reward  \label{lin:match}
        %\EndIf        
        \EndFor
	\end{algorithmic}
\end{algorithm}	

We briefly emphasize that when multiple nodes of some type $i$ are available, they are \emph{all} included as multiple entries with the same $p_{i,j}$ value in the vector $\vec{p}$ on which we call the pivotal sampling subroutine. The algorithm is well-defined as $p_{i,j} \in [0,1]$ by Constraint \eqref{eqn:tightOnlineConstraint} of \eqref{TLPon}.

\subsection{$\mathbf{(1-1/e)}$-Approximation via Independent Markov Chains} \label{sec:oneminusoneovere}

As a warm-up, we will first show that \Cref{alg:corrprop} achieves at least a $(1-1/e)$-approximation to $\opton$. This provides a significantly simpler algorithm and analysis for this known bound. In fact, the $(1-1/e)$-approximation, derived below, holds even when the call to pivotal sampling in \Cref{lin:PS} of \Cref{alg:corrprop} is replaced with independent sampling. 

Let $Q(t) := (Q_i(t))_i$ denote the continuous-time Markov chain for the number of nodes of each type $i$ that are \underline{q}ueued in the system at time $t$, when running \Cref{alg:corrprop}. As each offline node has a positive arrival and departure rate, it is immediate that this Markov chain is irreducible and positive recurrent, and hence has a unique stationary distribution $\pi$. The arrivals of online types are determined by Poisson processes independent of the state $Q(t)$; via the classic ``PASTA'' property it suffices to understand the stationary distribution $\pi$ to compute the expected long-term average gain from matching $j$ (refer to \Cref{app:prelims} for a brief refresher on continuous-time Markov chains and the PASTA property).  

\begin{lemma}[PASTA property, see \Cref{lem:pasta}] \label{lem:pastasection2}
    For $Q \in \mathbb{Z}_{\ge 0}^I$, let $\textup{\textsf{ALG}}(j, Q)$ denote the expected instantaneuous gain of an online algorithm $\textup{\textsf{ALG}}$ if online type $j$ just arrived and there are $Q_i$ nodes of type $i$ available for every $i \in I$. Then, the average reward gain by \textup{\textsf{ALG}} from matching arrivals of type $j$ equals $\gamma_j \cdot  \mathbb{E}_{Q \sim \pi} [\textup{\textsf{ALG}}(j, Q)].$
\end{lemma}

Unfortunately, the stationary distribution $\pi$ induced by \Cref{alg:corrprop} may be extremely complex. A technique introduced by \cite{aouad2022dynamic} is hence to prove stochastic dominance between $\pi$ and the stationary distribution of a simpler Markov chain, where the number of available nodes for each type evolves according to some Markov chain, independent of other types.

\begin{definition}\label{def:imc}
    In the \emph{independent Markov chains}, denoted by $Q^{\textup{\textsf{ind}}}(t) := (Q^{\textup{\textsf{ind}}}_i(t))_{i \in I}$, each $Q^{\textup{\textsf{ind}}}_i(t)$ evolves as an independent birth-death process. When $Q^{\textup{\textsf{ind}}}_i(t)$ is in state $k \in \mathbb{Z}_{\ge 0}$, it increases by 1 at rate $\lambda_i$ and decreases by 1 at rate $k \cdot  \left( \mu_i + \sum_{j \in J} \gamma_j \cdot p_{i,j} \right)$. We denote the stationary distribution of $Q^{\textup{\textsf{ind}}}$ by $\pi^{\textup{\textsf{ind}}}$.
\end{definition}

% \Alirezacomment{It may be worth highlighting that the queue above is not what \cite{aouad2020dynamic} used. For them, the downward pressure was $k\mu_i + \sum_j \gamma_j p_{i,j}$, which is more difficult to analyze.}

To see why $Q^{\textsf{ind}}(t)$ is dominated by the true evolution of $Q(t)$, consider a time $t$ where $Q_i(t) = k$ (i.e., there are $k$ nodes of type $i$ waiting who have not yet been matched or departed). The ``birth rate''\footnote{By birth rate, we mean the rate at which $Q_i$ transitions from state $k$ to state $k+1$. By death rate, we mean the transition rate from $k$ to $k-1$.} with which a new node of type $i$ joins is precisely the arrival rate $\lambda_i$, as in the independent Markov chains. Each of the waiting nodes of type $i$ will depart after time $\text{Exp}(\mu_i)$, hence the ``death rate'' is at least $k \cdot \mu_i$ just due to departures. The other contribution to the death rate is the rate at which $i$ is matched. If an online type $j$ arrives, each copy of type $i$ proposes (in a correlated fashion) with probability $p_{i,j}$; only one proposer will be matched according to \Cref{alg:corrprop}. To compute the exact match rate, we hence need to know the number of available nodes of all other types, i.e., $(Q_{i'}(t))_{i' \neq i}$. However, it is easy to \emph{upper bound} that no matter the state of these other types, each copy of type $i$ is matched to $j$ at rate at most $\gamma_j \cdot p_{i,j}$. For this reason, we have that $Q(t)$ stochastically dominates $Q^{\textup{\textsf{ind}}}(t)$, and analyzing the reward of \Cref{alg:corrprop} under the independent Markov chains is a lower bound on its actual performance. This is formalized in the following claim. 

\begin{restatable}{claim}{claimstochasticdominance}\label{claim:stochasticdominance}
    The evolution of $Q(t)$ stochastically dominates that of the independent Markov chains $Q^{\textup{\textsf{ind}}}(t)$. In particular, letting $\textup{\textsf{ALG}}(j, Q)$ be the expected instantaneous gain of \Cref{alg:corrprop} when online type $j$ just arrived and there are $Q_i$ nodes of type $i$ available for every $i \in I$, then $$\gamma_j \cdot  \mathbb{E}_{Q \sim \pi} [\textup{\textsf{ALG}}(j, Q)] \ge \gamma_j \cdot\mathbb{E}_{Q \sim \pi^{\textup{\textsf{ind}}}} [\textup{\textsf{ALG}}(j, Q)]. $$
\end{restatable}

For a formal treatment of stochastic dominance and a full proof of this claim, we refer the reader to \Cref{appendix:proofofclaimstochasticdominance}. The advantage of studying the evolution of $Q^{\textup{\textsf{ind}}}(t)$ via independent Markov chains is that we can solve for the stationary distribution $\pi^{\textup{\textsf{ind}}}$ exactly, as each $Q^{\textup{\textsf{ind}}}_i(t)$ evolves as a classic Markov chain known as a birth-death process (see \Cref{app:prelims} for a reminder). We note that while \cite{aouad2022dynamic} used stochastic dominance with an independent Markov chains, it took a different, harder-to-analyze form than ours.\footnote{In \cite{aouad2022dynamic}, an offline node's neighbors are drawn once upon its arrival, whereas we redraw the proposals upon each new online node's arrival. In our algorithm, the propensity of matching each offline type increases with queue lengths, which is a fundamental difference in the algorithm design.}

\begin{cor} \label{claim:stationarydistiMcPoisson}
    The stationary distribution $\pi^{\textup{\textsf{ind}}}$ of the Markov chain $Q^{\textup{\textsf{ind}}}(t)$ has the number of available nodes of type $i$ distributed independently as $\textup{Pois}(x_{i,a} / \mu_i)$.
\end{cor} 

\begin{proof}
Using that $p_{i,j} := \frac{x_{i,j}/\gamma_j}{x_{i,a}/\mu_i}$ and Constraint~\eqref{eqn:tightOfflineFlow} of \eqref{TLPon}, we have $$\mu_i + \sum_j \gamma_j \cdot p_{i,j} = \mu_i + \sum_j \frac{x_{i,j}}{x_{i,a}/\mu_i} \overset{\eqref{eqn:tightOfflineFlow}}{=} \mu_i + \frac{\lambda_i - x_{i,a}}{x_{i,a}/\mu_i} = \mu_i \cdot \frac{\lambda_i}{x_{i,a}} \ . $$ The claim follows (refer to \Cref{claim:stationarydistbirthdeath} for a refresher on the calculation). 
\end{proof}

For every online type $j$ and $w > 0$, define $$R_{j}(w) := \sum_{i : r_{i,j} \ge w} p_{i,j} \cdot \text{Pois}(x_{i,a}/\mu_i) \,$$ where $\{\text{Pois}(x_{i,a}/\mu_i)\}_i$ denote independent Poisson random variables with rates $\{x_{i,a}/\mu_i\}_i$. When available offline types are distributed according to $\pi^{\textup{\textsf{ind}}}$, the probability that \Cref{alg:corrprop} generates reward at least $w$ from matching an arrival of type $j$ is directly related to the distribution of $R_j(w)$, as formalized in the following claim. 

\begin{claim}\label{claim:gainequalexpectedmin}
   For every $j \in J$, we have $ \prtwo{Q \sim \pi^{\textup{\textsf{ind}}}}{\textup{\textsf{ALG}}(j, Q )  \ge w] = \mathbb{E}[\min(1, R_{j}(w))}  .$
\end{claim}

\begin{proof}
    Under a certain realization of $Q$, let $k$ denote the number of available nodes $i$ with $r_{i,j} \ge w$, i.e., $k := \sum_{i : r_{i,j} \ge w} Q_i$.  For this realization, note that the first $k$ entries of $\vec{p}$ correspond to exactly those nodes with reward at least $w$ because the nodes are sorted according to non-increasing $r_{i,j}$ values in \Cref{lin:inside-loop-start}. Then, applying Property \ref{level-set:prefix} of pivotal sampling with this value of $k$, the probability that an arrival $j$ is matched to a node with reward at least $w$ is equal to $\min \left(1, \sum_{i: r_{i,j} \ge w} Q_i \cdot p_{i,j} \right).$
    Using \Cref{claim:stationarydistiMcPoisson}, the claim follows when sampling $Q$ from the stationary distribution of the independent Markov chains. 
\end{proof}

To bound expectations of the form $\mathbb{E}[\min(1, R_{j}(w))]$, we recall some basic notions about the convex order on random variables.

\begin{definition}
For random variables $X$ and $Y$, we say that $X$ is dominated by $Y$ in the convex order, and write $X \cvxle Y$, if for any convex $f$  we have $\E[f(X)] \le \E[f(Y)].$ 
\end{definition}

We will use the following two facts (refer to \Cref{app:proofofweightedPoisCvx} for a proof of the latter).
\begin{fact}[c.f., \cite{shaked2007stochastic}]
If $X_i \cvxle Y_i$ for $i \in [n]$, and $(X_i)_{i=1}^n, (Y_i)_{i=1}^n$ are jointly independent, then $\sum_{i=1}^n X_i \cvxle \sum_{i=1}^n Y_i.$
\end{fact}
\begin{restatable}{fact}{weightedPoisCvx} \label{lem:weightedPoisCvx}
    For any $0 \le a \le b$, we have $a \cdot \textup{Pois}(b) \cvxle b \cdot \textup{Pois}(a).$
\end{restatable}

\paragraph{The $\mathbf{(1-1/e)}$ bound.} With these tools in place, a lower bound of $(1-1/e)$ on the approximation ratio is direct. Indeed, note that by \Cref{lem:weightedPoisCvx}, we have that 
\begin{equation}
R_{j}(w)  
\cvxle \sum_{i : r_{i,j} \ge w} 1 \cdot \text{Pois}(p_{i,j} \cdot x_{i,a}/\mu_i) \notag 
= \sum_{i : r_{i,j} \ge w} \text{Pois} \left( \frac{x_{i,j} }{ \gamma_j } \right) \notag 
= \text{Pois} \left(\sum_{i : r_{i,j} \ge w} \frac{x_{i,j} }{ \gamma_j } \right).  \label{eq:matchpois}
\end{equation}
As $\min(1,x)$ is concave, we infer that $\mathbb{E}[\min(1, R_{j}(w))]$ is lower bounded by  
\begin{align}
 \mathbb{E} \left[ \min \left( 1 , \text{Pois} \left(\sum_{i : r_{i,j} \ge w} \frac{x_{i,j} }{ \gamma_j } \right) \right) \right] 
= 1 - \exp \left( - \sum_{i : r_{i,j} \ge w} \frac{x_{i,j}}{\gamma_j } \right)  
\ge (1-1/e) \cdot  \sum_{i : r_{i,j} \ge w} \frac{x_{i,j}}{ \gamma_j }.  \label{Emin1lower}
\end{align}

We can calculate the stationary reward from matching to $j$ by integrating the complementary CDF, as follows:
\begin{align*}
    \gamma_j \cdot  \mathbb{E}_{Q \sim \pi^{\textsf{ind}}} [\textup{\textsf{ALG}}(j, Q)] &\ge \gamma_j \cdot \int_0^\infty \E[\min(1,R(w))] \, dw && \hspace{-7em} \text{\Cref{claim:stochasticdominance}, \Cref{claim:gainequalexpectedmin}}\\
    &\ge \gamma_j \cdot \int_0^\infty (1 - 1/e) \cdot \sum_{i: r_{i,j} \ge w} x_{i,j} / \gamma_j \, dw && \text{\eqref{Emin1lower}}\\
    &\ge (1-1/e) \cdot \int_0^\infty \sum_i x_{i,j} \cdot \mathbbm{1}[r_{i,j} \ge w]  \, dw   
    = (1-1/e) \cdot \sum_i x_{i,j} \cdot r_{i,j}.
\end{align*}
This demonstrates \Cref{alg:corrprop} achieves a $(1-1/e)$-approximation to $\textup{OPT}\eqref{TLPon}$. \\

\noindent \emph{Remark.} It is not hard to see that the same argument works for an algorithm which makes proposals independently, rather than via pivotal sampling. For such an algorithm, the probability that no copy of $i$ proposes to $j$ in the stationary distribution is given by $\mathbb{E}[(1-p_{i,j})^{\text{Pois}(x_{i,a}/\mu_i)}] = \exp \left(- x_{i,a} / \mu_i \cdot p_{i,j} \right) = \exp \left( x_{i,j} / \gamma_j \right). $ Thus $\Pr_{Q \sim \pi^{\textsf{ind}}}[\textsf{ALG}(j,Q)] \ge w] = 1 - \prod_{i : r_{i,j} \ge w} \exp(-x_{i,j} / \gamma_j)$ which results in the same bound. In \Cref{sec:sharpened}, to beat $(1-1/e)$ we will crucially use the correlation of pivotal sampling to get a better lower bound on $\mathbb{E}[\min(1, R_j(w))]$, which does \emph{not} hold for independent proposals.

\section{Improved Approximation Ratio}\label{sec:approximation}

In this section we prove our main result, restated below. 

\mainthm*

%This directly results in improvements for other graph structures, as in \Cref{app:othergraphstructures}. 
Our first step is noting that for many inputs, the analysis of \Cref{sec:oneminusoneovere} may already be loose, in which case \Cref{alg:corrprop}'s approximation ratio is strictly better than $1-1/e$. We concretely identify three such cases. When none of these cases apply, it is possible that the analysis in the previous section cannot go beyond $1-1/e$; the main technical hurdle is the independence of offline types in the independent Markov chains. We also observe that, on certain instances, \Cref{alg:corrprop} may make many matches on edges with low reward, at the expense of high-reward edges, leading to an approximation ratio no better than $1-1/e$.

We overcome these obstacles by noting that inputs which do not satisfy any of the three cases are extremely structured. In particular, they (roughly) are online-vertex-weighted and ``highly-connected''---proposals are with probability $\approx 1$ and each online node j has a neighbor present with high probability. For such inputs, we propose a new algorithm (\Cref{alg:second}). As a pre-processing step, this algorithm removes ``unimportant'' edges and nodes. On the resulting subgraph, we run a greedy algorithm that breaks ties in a balanced way rather than using rewards. In particular we split each offline node into a top and bottom type with equal probability; our {\em Balanced Greedy} algorithm prioritizes the nodes in the top section. The core technical contribution is using the structure of vertex-weighted highly-connected inputs to analyze Balanced Greedy, while accounting for correlations among offline types using a new notion of {\em weakly correlated Markov chains}.

In \Cref{sec:sharpened} we present a sharpened analysis of \Cref{alg:corrprop} for the easy cases. In \Cref{ssec:correlated_analysis}, we formally define ``vertex-weighted highly-connected'' instances and present \Cref{alg:second} along with its analysis. Throughout, we will use the following notation.

% These bounds are also useful for analyzing the \emph{competitive ratio} against the optimum offline benchmark, as we show in \Cref{sec:competitiveRatio}.

\paragraph{Notation.} For convenience we define the shorthand $\textsf{ALG}_j := \gamma_j \cdot \mathbb{E}_{Q \sim \pi^{\textup{\textsf{ind}}}} [\textsf{ALG}(j,Q)]$ for the expected stationary reward of \Cref{alg:corrprop} from matches to $j$ in the stationary distribution of the independent Markov chains. Let $(x_{i,j})_{i,j}, (x_{i,a})_i$ be an optimal solution of \eqref{TLPon}. For any $S \subseteq J$ we let $\textsf{LP-Gain}_S := \sum_{j \in S} \sum_{i \in I} r_{i,j} x_{i,j}$ denote the total gain from matching to online types in $S$ in the optimal LP solution. For convenience of notation, we will let $\textsf{LP-Gain}_j := \textsf{LP-Gain}_{\{j\}}$. For any offline type $i$, we let $N_i := \{ j : x_{i,j} > 0 \}$ denote the neighborhood of $i$ (and we define $N_j$ analogously). For every offline type $i$, we let $\Gamma_i := \sum_{k \in N_i} \gamma_k$.

\subsection{Three (Relatively) Easy Cases} \label{sec:sharpened}

Below, we fix an impatient type $j$ and identify three different cases, each of which implies that \Cref{alg:corrprop} already achieves a boosted $(1-1/e + b(\eps))$ approximation on $j$ for a continuous, increasing function $b(\cdot)$ with $b(0) = 0$. Recall from \Cref{sec:oneminusoneovere} that $\textsf{ALG}_j$ can be written as
\begin{align}
    \textsf{ALG}_j &= \int_0^\infty \ex{\min(1,R_j(w))} \ge \int_0^\infty \left(1-\exp\left(-\sum_{i: r_{i,j} \geq w} {x_{i,j}}/{\gamma_j}\right)\right) \, dw  \ .\label{ineq:wintegrationsec3}
\end{align} 

Previously, our $(1-1/e)$-approximation followed from the \emph{concavity step} lower bounding $1-\exp(-z) \ge (1-1/e)z$ for $z \in [0,1]$. The first two cases identified below simply quantify when we can improve on this simple bound. The third case requires more care, as we instead give a tighter bound for $\mathbb{E}[\min(1, R_j(w))]$ than in \Cref{sec:oneminusoneovere}; our improvement here crucially relies on the correlation present in pivotal sampling and does \emph{not} hold for independent proposals.  

For each of the three cases, we label it with an informal title followed by the corresponding formal observation. Full proofs can be found in \Cref{app:threecases}.

% When these cases do not apply, if a case does not hold, we remove its unimportant edges, incurring a small loss of LP's gain. This edge exclusion will help subsequent cases.\footnote{Note that this edge exclusion merely simplifies the exposition and it is not required for our algorithm.} In the following, let $0 < \eps < 0.1$ be a fixed constant (later, $\epsilon$ will be chosen sufficiently small relative to the constant factor gain in the approximation ratio that we establish). % Note that while we should use different values of $\eps$ at each case to obtain a $(1-1/e + \Omega(1))$ approximation, we abuse the notation by reusing $\eps$ to keep the notation convenient. Finally, we show that if a significant number of impatient types satisfy our case analysis, our algorithm obtains an overall $(1-1/e + \Omega(1))$ approximation. Otherwise, we must go beyond independent queues, as discussed in \Cref{ssec:correlated_analysis}.

% Throughout the rest of this section, we use the algebraic fact that if $z \leq 1-\eps$, we have
% \[
%     1-\exp(-z) = 1-\exp\left(-(1-\eps) \cdot \frac{z}{1-\eps}\right) \geq \frac{1-e^{-(1-\eps)}}{1-\eps} \cdot z \ ,
% \]
% where the inequality follows from the convexity of $z \to 1-\exp(-bz)$. Note that the boost function $b(\eps) := \frac{1-\exp(-(1-\eps))}{1-\eps}$ is strictly increasing and thus, it is bounded away from $1-1/e$ for any fixed $\eps > 0$. 

\paragraph{Case 1: {\bf\eqref{TLPon}} does not saturate $\boldsymbol{j}$.\label{case1}} This case is the easiest---if the LP mass incident to $j$ is less than $\gamma_j$, the concavity step is loose. Formally, we have the following observation.

\begin{restatable}{obs}{obscaseone} \label{obscase1}
For $\eps \in (0,0.1)$, if $\sum_{i} x_{i,j} \le (1-\epsilon ) \cdot \gamma_j$, then $\textup{\textsf{ALG}}_j \ge (1 - 1/e + b (\epsilon)) \cdot \textsf{\textup{LP-Gain}}_j$ for some $b (\epsilon) > 0$. 
\end{restatable}

\paragraph{Case 2: $\textsf{LP-Gain}_j$ is not concentrated on a single reward weight.\label{case2}} This case follows from carefully extending the logic of \hyperref[case1]{Case 1}. In particular, if there is a threshold $\bar{w} \in \mathbb{R}_{\ge 0}$ such that the LP mass with reward at least $\bar{w}$ is non-zero but does not saturate $j$, the concavity step in \eqref{ineq:wintegrationsec3} is loose for rewards above $\bar{w}$. We formalize this intuition in the following way. 

 \begin{restatable}{obs}{obscasetwo} \label{obs:case2}
    For real numbers $a \le b$ define $$\textup{\textsf{LP-Gain}}_j^{[a,b]} := \sum_{i : r_{i,j} \in [a,b]} x_{i,j} \cdot r_{i,j}.$$ For $\eps \in (0,0.1)$, suppose that for every $w \in \mathbb{R}_{\geq 0}$ we have $\textup{\textsf{LP-Gain}}_j^{[w, w(1+\epsilon)]} < (1-\epsilon) \cdot \textup{\textsf{LP-Gain}}_j$. Then, we have $\textup{\textsf{ALG}}_j \geq (1-1/e+b(\epsilon)) \cdot \textup{\textsf{LP-Gain}}_j$, for some $b(\epsilon) > 0$. 
\end{restatable}

 \paragraph{Case 3: Some of $\bs{j}$'s proposal probabilities $\bs{\{p_{i,j}\}}$ are bounded away from 1.\label{case3}} Our final case is the most important and relies crucially on the correlated proposals of \Cref{alg:corrprop}. Roughly speaking, there is no gain to be had when correlating proposals if all of the marginal probabilities are equal to 1---proposing independently would be equivalent. If however a non-negligible number of $j$'s neighbors have a marginal proposal probability bounded away from 1, we can take advantage of the correlation introduced by pivotal sampling. In particular, we can tighten our bounding of $R_j(w)$ in the convex order, resulting in the following observation.

\begin{restatable}{obs}{obscasethree} \label{obscase3}
For $\eps \in (0,0.1)$, suppose that there exists $S \subseteq I$ such that $\sum_{i \in S} x_{i,j} \ge \epsilon \cdot \gamma_j$ and $\max_{i \in S} p_{i,j} \leq 1-\epsilon$. Then, we have $\textup{\textsf{ALG}}_j \geq (1-1/e+b (\epsilon)) \cdot \textup{\textsf{LP-Gain}}_{j}$ for some $b (\epsilon) > 0$.
\end{restatable}

We can of course assume Observations~\ref{obscase1}--\ref{obscase3} hold with the same continuous boost function $b(\cdot)$. We next provide our definition of a \emph{hard} online type. It would be natural to say that any online type $j$ is hard if none of these three cases apply. We, however, use a slightly stronger definition. In particular, for some online type $j$,  note that if \hyperref[case2]{Case 2} applies, we are already better-than-$(1-1/e)$ approximate on $j$. If \hyperref[case2]{Case 2} does not apply, $j$'s rewards are concentrated around some $r_j \in \mathbb{R}_{\ge 0}$. We can then delete edges $(i,j)$ with $r_{i,j} \notin [r_j, r_j(1+\eps)]$ and attempt to apply the first or third case on the resulting neighborhood, with some larger parameter $\eps'$. If this is insufficient as well, we call an online type \emph{hard}, as formalized below. 

\begin{definition} \label{def:hardtype}
    We say online type $j$ is $\epsilon$-\emph{hard} if
    \begin{enumerate}[label=(\roman*)]
        \item There exists $r_j \ge 0$ such that $\textup{\textsf{LP-Gain}}_j^{[r_j, r_j(1+\epsilon)]} \ge (1 - \epsilon ) \cdot \textup{\textsf{LP-Gain}}_j$;
        \item For $S_j := \{i : r_{i,j} \in [r_j, r_j(1+\epsilon )]\}$ and $\epsilon' :=   b^{-1}(2\epsilon)$ we have $\sum_{i\in S_j \, : \, p_{i,j} \geq 1-\eps'} x_{i,j}/\gamma_j \geq 1-2\eps'.$
    \end{enumerate}
\end{definition}

If a constant fraction of the LP reward is supported on online types $j$ that are not hard, then \Cref{alg:corrprop} beats $1-1/e$, as formalized in the following observation.  

\begin{restatable}{lem}{casesbeatoneminusoneovere} \label{lem:casesBeat1-1/e}
    If $\mathcal{H}_{\eps}$ denotes the set of $\eps$-hard online types, and $\textup{\textsf{LP-Gain}}_{J \setminus \mathcal{H}_{\eps}} \geq \epsilon \cdot \textup{\textsf{LP-Gain}}_J$, then \Cref{alg:corrprop} achieves stationary reward at least $(1 - 1/e + c(\epsilon)) \cdot \textup{\textsf{LP-Gain}}_J$ for some constant $c(\epsilon) > 0$. 
\end{restatable}

The proof is straightforward, arguing that on hard types \Cref{alg:corrprop} achieves a baseline $(1-1/e)$-approximation, while for easy types it achieves a boosted $(1-1/e+b(\eps))$-approximation, via either \hyperref[case2]{Case 2} or one of \hyperref[case1]{Case 1} or  \hyperref[case3]{Case 3} applied to the neighborhood $S_j$ with parameter $\epsilon'$. A full proof is in \Cref{app:casesbeat}.

% \begin{prop}
%     If instance ${\cal I}$ is non-expansive, \Cref{alg:corrprop} is $(1-1/e+b(\eps))$-approximate. 
% \end{prop}

\subsection{Difficult Case: Vertex-weighted Highly-connected Instances}\label{ssec:correlated_analysis}

\Cref{lem:casesBeat1-1/e} shows that if an instance derives a constant fraction of its LP gain from easy types, \Cref{alg:corrprop} suffices to obtain a $(1-1/e+c(\eps))$-approximation for some constant $c(\epsilon) > 0$. Therefore, the next section focuses on the core difficult scenario, where instances that  do not satisfy the conditions of \Cref{lem:casesBeat1-1/e} because they are primarily composed of hard types. We refer to this class of instances as {\em vertex-weighted highly-connected}, to make their structure explicit. Indeed, such instances are approximately {\em online vertex-weighted}.\footnote{Interestingly, this a different notion from the vertex-weighted assumption frequently considered in the online matching literature, which generally requires that \emph{offline} nodes are weighted.} Moreover, the {\em highly-connected} property signals the fact that most proposal probabilities are approximately equal to $1$, and upon each hard type $j$'s arrival, there is a high probability that one of $j$'s neighbors is present. Indeed, the capacity constraints of hard online types in \eqref{TLPon} are nearly saturated. In light of the new tightening constraints~\eqref{eqn:tightOnlineFlow},  this condition implies that $\sum_{i \in N_j} \lambda_i / \mu_i $ is large, which translates into a high probability that one of $j$'s neighbors is present in steady-state.

% By our new tightened constraint (i.e., Constraint~\eqref{eqn:tightOnlineFlow}),  $j$ being nearly saturated implies that $\sum_{i \in N_j} \lambda_i / \mu_i $ is large. This means that upon each $j$'s arrival, there is a high probability that one of $j$'s neighbors is present. For this reason, we refer to this class of instances as ``vertex-weighted highly-connected.''

\begin{definition}
    A \emph{$\eps$-vertex-weighted highly-connected ($\eps$-VWHC)} instance is one where $$\textup{\textsf{LP-Gain}}_{J \setminus \mathcal{H}_{\eps}} < \epsilon \cdot \textup{\textsf{LP-Gain}}_J.$$ 
\end{definition}

Despite the improvements identified in the previous section, it is possible that \Cref{alg:corrprop} obtains an approximation ratio arbitrarily close to $1-1/e$ (see \Cref{app:alg_bad_example}). For this reason, in this section we introduce a related algorithm (\Cref{alg:second}) and show that on any instance, at least one of our two algorithms beats $1-1/e$ by a universal constant. 

% In this section, we focus on vertex-weighted highly-connected (VWHC) instances and obtain an improved approximation ratio for them. On such instances,  Therefore, we present a second algorithm and show that it is $(1-1/e + 
% \Omega(1))$-approximate. Our main result of this section is stated in the next proposition. 

\begin{prop}\label{prop:main}
    There exists a constant $\zeta > 0$, independent of $\eps$, and a polynomial-time algorithm with approximation ratio at least $(1- f(\eps)) \cdot (1 - 1/e + \zeta)$ on $\epsilon$-VWHC instances, for some non-negative continuous $f(\cdot)$ with $f(0) = 0$. 
\end{prop}

With this proposition in place, the proof of our main theorem follows by suitable choice of $\eps$. 

\begin{proof}[Proof of \Cref{thm:main}]
Taking $\eps = f^{-1}(\zeta / 2)$ immediately gives at least a $(1-1/e + \zeta / 2)$-approximation on $\eps$-VWHC instances. For any instance that is not $\eps$-VWHC, by \Cref{lem:casesBeat1-1/e} the approximation ratio of \Cref{alg:corrprop} is at least $1-1/e+c(f^{-1}(\zeta/2)).$
\end{proof}

To prove this proposition, in the remainder of the section for ease of notation we assume $\epsilon$ and $\eps' = b^{-1}(2\eps)$ are fixed, and slightly abuse the notation accordingly (for example, we refer to $\eps$-VWHC instances simply as VWHC instances). 

\paragraph{Algorithm outline.} As our analysis of \Cref{alg:corrprop} might be tight due to overmatching offline nodes to online nodes of low rewards, we resolve this issue by explicitly deleting all easy online types, as well as unimportant edges incident to hard nodes. For VWHC instances, such deletions lose very little LP weight. For the remaining subgraph, \Cref{alg:second} obtains a constant-factor boost using a form of balanced tie-breaking that we achieve by splitting each offline type into two copies, referred to as ``top'' and ``bottom'' nodes. The \textsf{Instance Transformation} subroutine formalizes these ideas (and \Cref{fig:transformation} visualizes them). 
\begin{tcolorbox}[colback=white, colframe=black, width=\textwidth, arc=5mm, auto outer arc]
\noindent \textsf{Instance Transformation}  \label{transf}
\vspace{-0.3cm}
\paragraph{}{\it Step 1: Removal of easy nodes and unimportant edges.} We remove all easy nodes from the graph. Moreover, for every hard online type $j$, we remove every edge with $r_{i,j} \not \in [r_j, (1+\eps) r_j]$ or $p_{i,j} < 1-\eps'$, where we recall that $\eps' = b^{-1}(2\eps)$ as defined in the previous section. Algorithmically, these removals mean that we ignore such edges or nodes.
\vspace{-0.3cm}
\paragraph{}{\it Step 2: Randomization through offline-node splits. \label{step2}}  %Our greedy matching policy requires a specific randomization that can be captured by splitting each offline vertex into two identical nodes. 
We use the Poisson thinning property to split type-$i$ online nodes into types $(i,1)$ and $(i,2)$ with $\lambda_{(i,1)} = \lambda_{(i,2)}$, by flipping a fair coin upon the arrival of any offline type. Let $$\text{TOP} := \{(i,1)\}_{i \in I} \quad \text{and} \quad \text{BOT} := \{(i,2)\}_{i \in I}. $$ 
\end{tcolorbox}

The \hyperref[transf]{\textsf{Instance Transformation}} subroutine leaves us with very structured instances---they are nearly (online) vertex weighted, all proposal probabilities are almost 1, and all online types are almost saturated. As previously mentioned, the tightening constraints \eqref{eqn:tightOnlineFlow} in \eqref{TLPon} imply that for every online type $j$, $\sum_{i \in N_j} \frac{\lambda_i}{\mu_i}$ is large. Intuitively, this means there is a high chance $j$ has an adjacent offline neighbor present upon its arrival, which is the reason for our highly-connected nomenclature.  With this subroutine in hand, we present our second algorithm. 

At a high level, \Cref{alg:second} matches each online node greedily to one of its available neighbors on the offline side. Nonetheless, it prioritizes top nodes over bottom nodes whenever possible; this will ensure sufficiently balanced match rates across different types. It is instructive to contrast how our two algorithms operate on the transformed instances. While  \Cref{alg:corrprop} would specify proposal probabilities close to 1, it instead prioritizes offline nodes based on rewards, as implied by the pivotal sampling prefix ordering.

\begin{figure}[ht]
    \centering
    % First minipage for Figure 1
    \begin{minipage}[b]{0.25\textwidth}%
        \centering
        \begin{tikzpicture}[scale=0.8, every node/.style={scale=0.8}]
            % Left vertices
            \foreach \i in {1,...,3} {
                \node[draw, circle, fill=black, inner sep=3pt] (L\i) at (0, -1.5*\i+1.5) {};
            }
            
            % Right vertices
            \node[draw, circle, fill=red, inner sep=3pt] (R1) at (2, 0) {};
            \node[draw, circle, fill=green, inner sep=3pt] (R2) at (2, -1.5) {};
            \node[draw, circle, fill=red, inner sep=3pt] (R3) at (2, -3) {};
            
            % Draw solid edges
            \draw[line width=1.2pt] (L1) -- (R1);
            \draw[dashed] (L2) -- (R1);
            \draw[line width=1.2pt] (L2) -- (R2);
            \draw[line width=1.2pt] (L3) -- (R1);
            \draw[line width=1.2pt] (L3) -- (R2);
            \draw[line width=1.2pt] (L3) -- (R3);
            
            % Draw dashed edges
            \draw[dashed] (L1) -- (R3);
            \draw[line width=1.2pt] (L2) -- (R3);
        \end{tikzpicture}
        \caption*{Original instance}
        \label{fig:figure1}
    \end{minipage}%
    \hspace{0.02\textwidth}%
    % Arrow between Figure 1 and Figure 2
    \raisebox{2cm}{\color{black} \Large$\Rightarrow$}%
    \hspace{0.02\textwidth}%
    % Second minipage for Figure 2
    \begin{minipage}[b]{0.25\textwidth}%
        \centering
        \begin{tikzpicture}[scale=0.8, every node/.style={scale=0.8}]
            % Left vertices
            \foreach \i in {1,...,3} {
                \node[draw, circle, fill=black, inner sep=3pt] (L\i) at (0, -1.5*\i+1.5) {};
            }
            
            % Right vertices
            \node[draw, circle, fill=red, inner sep=3pt] (R1) at (2, 0) {};
            % \node[draw, circle, fill=red, inner sep=3pt] (R2) at (2, -1.5) {};
            \node[draw, circle, fill=red, inner sep=3pt] (R3) at (2, -3) {};
            
            % Draw solid edges
            \draw[line width=1.2pt] (L1) -- (R1);
            % \draw[dashed] (L2) -- (R1);
            % \draw[line width=1.2pt] (L2) -- (R2);
            \draw[line width=1.2pt] (L3) -- (R1);
            % \draw[line width=1.2pt] (L3) -- (R2);
            \draw[line width=1.2pt] (L3) -- (R3);
            
            % Draw dashed edges
            % \draw[dashed] (L1) -- (R3);
            \draw[line width=1.2pt] (L2) -- (R3);
        \end{tikzpicture}
        \caption*{After step 1}
        \label{fig:figure2}
    \end{minipage}%
    \hspace{0.02\textwidth}%
    % Arrow between Figure 2 and Figure 3
    \raisebox{2cm}{\color{black} \Large$\Rightarrow$}%
    \hspace{0.02\textwidth}%
    % Third minipage for Figure 3
    \begin{minipage}[b]{0.25\textwidth}%
        \centering
        \begin{tikzpicture}[scale=0.8, every node/.style={scale=0.8}]
            % Left vertices
            \foreach \i in {1,...,6} {
                \node[draw, circle, fill=black, inner sep=3pt] (L\i) at (0, -0.75*\i+1.25) {};
            }
            
            % Right vertices
            \node[draw, circle, fill=red, inner sep=3pt] (R1) at (2, 0) {};
            % \node[draw, circle, fill=red, inner sep=3pt] (R2) at (2, -1.5) {};
            \node[draw, circle, fill=red, inner sep=3pt] (R3) at (2, -3) {};
            
            % Draw solid edges
            \draw[line width=1.2pt] (L1) -- (R1);
            \draw[line width=1.2pt] (L4) -- (R1);

            % \draw[dashed] (L2) -- (R1);
            % \draw[line width=1.2pt] (L2) -- (R2);
            \draw[line width=1.2pt] (L3) -- (R1);
            \draw[line width=1.2pt] (L6) -- (R1);
            % \draw[line width=1.2pt] (L3) -- (R2);
            \draw[line width=1.2pt] (L3) -- (R3);
            \draw[line width=1.2pt] (L6) -- (R3);
            
            % Draw dashed edges
            % \draw[dashed] (L1) -- (R3);
            \draw[line width=1.2pt] (L2) -- (R3);
            \draw[line width=1.2pt] (L5) -- (R3);

            \draw[decorate, decoration={brace, amplitude=10pt, mirror, raise=2pt}, thick]
    (-0.2,0.8) -- (-0.2,-1.25) node[midway, left=12pt] {TOP};
        \draw[decorate, decoration={brace, amplitude=10pt, mirror, raise=2pt}, thick]
        (-0.2,-1.4) -- (-0.2,-3.5) node[midway, left=12pt] {BOT};
        \end{tikzpicture}
        \caption*{Transformed instance}
        \label{fig:figure3}
    \end{minipage}
    
    \caption{Schematic visualization of \textsf{``Instance Transformation''}}
    \label{fig:transformation}
    \begin{minipage}{1\textwidth}
        \raggedright
        \vspace{0.2cm}
        \footnotesize
In the original instance, red nodes are hard online types and green node is easy. Dashed edges have either a reward not in $[r_j, (1+\eps)r_j]$ or a proposal probability $p_{i,j} < 1-\eps'$. In the second step, each offline node (on the left) is splitted into two identical nodes. 
    \end{minipage}
    
%     \footnotesize
% In the original instance, red vertices are hard online types and green vertex is easy. Dashed edges have either a reward not in $[r_j, (1+\eps)r_j]$ or a proposal probability $p_{i,j} < 1-\eps'$. In the second step, each offline vertex (on the left) is splitted into two identical vertices. 

    \label{fig:sidebyside}
\end{figure}
% These operations ensure that every remaining online vertex satisfies \eqref{asmp1}-\eqref{asmp3}. 

% This operation does not change the LP value and retains feasibility and optimality, and will be useful in the analysis. 

% We now present an analysis that goes beyond the independent queues in the independent Markov chains. 

\begin{algorithm}[H]
	\caption{Balanced Greedy}
	\label{alg:second}
	\begin{algorithmic}[1]
 \State Solve \eqref{TLPon} for $\{x_{i,j}\}$ 
 \State Invoke the \hyperref[transf]{\textsf{Instance Transformation}} subroutine, modify the graph and LP solution accordingly \label{line:instancetransformation}
 
 \ForAll{timesteps where an online node $u$ of type $j$ arrived}
            \If{$u$ has an available neighbor in TOP}
                \State Match $u$ to one such neighbor arbitrarily
            \Else
                \State Match $u$ to an available neighbor in BOT arbitrarily (if any)
            \EndIf
        \EndFor
	\end{algorithmic}
\end{algorithm}

Recall that for VWHC instances, easy types make up at most an $\eps$-fraction of the (original) LP gain. Moreover, \Cref{def:hardtype} immediately implies that removing unimportant edges in Step 1 of \hyperref[transf]{\textsf{Instance Transformation}} incurs at most $O(\eps')$-fraction loss of LP gain. Therefore, performing the subroutine yields a modified instance with an LP gain that is at least $(1-\eps)(1-O(\eps')) = 1 - \tilde{f}(\eps)$ fraction of the original LP gain, as summarized in the following observation.\footnote{We emphasize that \Cref{alg:second} does not ``re-solve'' \eqref{TLPon} after running the subroutine; it simply modifies the original LP solution.} 

\begin{obs}\label{obs:smallLPloss}
    After performing \Cref{line:instancetransformation} of \Cref{alg:second}, for VWHC instances the resulting LP solution is feasible with value at least $(1-\tilde{f}(\eps)) \cdot \textup{OPT}\text{\eqref{TLPon}}$, for some continuous $\tilde{f}$ with $\tilde{f}(0)=0$. 
\end{obs}

\paragraph{Properties of the modified instance.} We  summarize below the properties we obtained by performing \hyperref[transf]{\textsf{Instance Transformation}}. For the analysis, we also assume that we ``split'' the solution of \eqref{TLPon} when performing \hyperref[step2]{Step 2} of the transformation, obtaining $x_{(i,1),j} = x_{(i,2), j}$ and $x_{(i,1),a} = x_{(i,2), a}$ for every $i \in I, j \in J$. It is straightforward to see that this new LP solution is feasible with the same objective value, using that $\lambda_{(i,1)} = \lambda_{(i,2)} = \lambda_i /2$; in particular, the non-linear Constraint~\eqref{eqn:tightOnlineFlow} continues to hold by concavity of $z \to 1-\exp(-z)$. For Property (v), which we call {\em binary queues}, we rely on the Poisson thinning property to assume that, for each offline type $i$, $\lambda_i$ is very small compared to $\mu_i$. We emphasize that this  property is exclusively for convenience of the analysis, and is not required algorithmically, in contrast to the Poisson thinning used in \hyperref[transf]{\textsf{Instance Transformation}}.

We henceforth abuse the notation by using our original parameters to refer to the modified instance, e.g., $J$ refers to the set of all online types remaining after running \hyperref[transf]{\textsf{Instance Transformation}}. 

    \begin{enumerate}[label=(\roman*)]

        \item {\it Vertex-weighted edges}: For every $j \in J$, there exists a non-negative $r_j$ such that every edge incident to $j$ has a reward within the interval $[r_j, (1+\eps)r_j]$. To simplify the notation, we assume in the sequel that all edges incident to $j$ have a reward exactly equal to $r_j$. Clearly, this assumption introduces an error factor of at most \(\epsilon\) into our analysis.
        \item {\it High proposal probability}: For every $i \in I, j \in J$, we have $1-\eps' \leq p_{i,j} \leq 1$. 
        \item {\it Saturation of online types}: For every $j \in J$, we have $1-\eps' \leq \sum_{i \in N_i} x_{i,j}/\gamma_j \leq 1$. Combining this property with Property (ii) gives for every $j$:
        \begin{align} \label{ineq:xiamubound}
            1 - \epsilon' \le \sum_{i \in N_j} \frac{x_{i,a}}{\mu_i} \le \frac{1}{1 - \epsilon'}  \ . 
        \end{align}
        \item {\it Balancedness}: Consider some $j \in J$. Let $N_j^\uparrow, N_j^\downarrow$ be, respectively, $j$'s neighbors in TOP and BOT. The splitting in our instance transformation guarantees that for every $j\in J$, the instance is {\em $j$-balanced}, meaning that $ \sum_{i\in N_j^\uparrow} \frac{\lambda_i}{\mu_i} = \frac{1}{2} \cdot \sum_{i\in N_j} \frac{\lambda_i}{\mu_i}$. Thus,    \begin{align}
\frac{1-\epsilon'}{2} \le \sum_{i \in N_j^\uparrow} \frac{x_{i,a}}{\mu_i} = \sum_{i \in N_j^\downarrow} \frac{x_{i,a}}{\mu_i}  \le \frac{0.5}{1-\epsilon'}  .                     \label{eqn:jBalancedness}
    \end{align}
    
        \item {\it Binary queues} \label{propertyv}: To simplify our exposition, we assume that the queue of available offline nodes of a fixed type waiting is \emph{binary}, i.e., has length 0 or 1. This property can be achieved by further splitting each offline type $i$ sufficiently many times---we refer the reader to \Cref{app:binary_queues} where we formally construct a reduction to binary queues incurring negligible loss, furthermore guaranteeing $\lambda_i / \mu_i \le \eps^2$ for all $i$. We preserve our definitions of TOP and BOT, and all other LP properties. Importantly, we make no change to our algorithm, and only use this property for purpose of analysis. Hereafter, we abuse our notation and use $I$ to denote the set of all offline types, after this splitting procedure.

         %We preserve the membership in TOP and BOT i.e. if $i \in \text{TOP}$, we split it into some vertices in TOP. Hence, the splits in this stage do not change the behavior of the algorithm and are only in the analysis. See \Cref{prf:instance_property} for a formal discussion. Going forward, we abuse the notation and the parameters refer to the instance after the splits in this stage, e.g., $n$ is increased $2K$ times compared to its value in the original instance. 

    \end{enumerate}

%    \Alicomment{To mention: The tightening constraints are non-linear in input parameters, so check that splitting $i$-s is fine (it is from concavity).}
    
        As a consequence of these properties, recalling the notation $\Gamma_i = \sum_{k \in N_i} \gamma_k$, Constraint~\eqref{eqn:tightOfflineFlow} of \eqref{TLPon} implies 
        \begin{align}
            (1-\eps') \cdot \frac{\gamma_j \lambda_i}{\mu_i + \Gamma_i} \leq x_{i,j} \leq \frac{1}{1-\eps'} \cdot \frac{\gamma_j \lambda_i}{\mu_i + \Gamma_i} && \forall i \in I, j \in N_i \label{ineq:match_bound}
        \end{align}
    and \begin{align}
            \frac{\mu_i\lambda_i}{\mu_i + \Gamma_i} \leq x_{i,a} \leq \frac{1}{1-\eps'} \cdot \frac{\mu_i\lambda_i}{\mu_i + \Gamma_i} \ . && \forall i \in I \label{ineq:abandonment_bound}
        \end{align}
        Refer to \Cref{app:approximate_lp} for a quick proof. 

\paragraph{Approximations up to $\mathbf{\epsilon}$.} To avoid carrying around cumbersome functions of $\epsilon$ and $\epsilon'$ in the proof, for convenience we write $\alpha \simeq_{\epsilon}  \beta$ to denote that there exists some absolute constant $C$, independent of $\epsilon$ (and hence $\epsilon'$), such that $(1 - C \cdot \max(\epsilon, \epsilon')) \cdot \beta \le \alpha \le (1 + C \cdot \max(\epsilon, \epsilon')) \cdot \beta$. Similarly, $\alpha \lesssim_{\eps} \beta$ means there exists some $C$ independent of $\epsilon$ with $\alpha \le (1 + C \cdot \max(\eps, \eps')) \cdot \beta.$

\subsection{Analysis via Weakly Correlated Markov Chains}

In the remainder of this section, we develop a tighter analysis of the expected average match rate for each offline type $j\in J$. To do so, we need to account for dependence between the queues of offline nodes; as previously mentioned,  the independent Markov chains cannot go beyond a $(1-1/e)$-approximation. Our crucial new tool is a Markov chain that better approximates the process induced by \Cref{alg:second} and is sufficiently tractable for analysis.

\begin{definition}[The weakly correlated Markov chains] The weakly correlated Markov chains, denoted by $Q^{\textup{\textsf{weak}}}(t) = (Q^{\textup{\textsf{weak}}}_i(t))_{i \in I}$ are defined as follows: 
\begin{enumerate}[label=(\roman*)]
\item For each $i \in \textup{TOP}$, the queue $Q_i^{\textup{\textsf{weak}}}(t)$ evolves according to an independent birth-death process, which at time $t$ increases by 1 at rate $\lambda_i$ and decreases by 1 at rate $Q_i^{\textup{\textsf{weak}}}(t) \cdot \left(\mu_i + \Gamma_i \right)$.\footnote{This is roughly the same as the evolution of the independent Markov chains, in which the death rate is $Q_i^{\textup{\textsf{ind}}}(t) \cdot (\mu_i + \sum_{j \in N_i} \gamma_j \cdot p_{i,j})$, because in our instance $p_{i,j} \ge 1-\epsilon$ for every $i \in N_j$.}

\item Define $\textup{\textsf{TE}}_j(t)\in \{0,1\}$ as $\textup{\textsf{TE}}_j(t) := 1 - \mathbbm{1}[\sum_{i\in  N_j^{\uparrow}} Q^{\textup{\textsf{weak}}}_i(t) > 0 ]$. For each $i \in \textup{BOT}$ we let ${Q}^{\textup{\textsf{weak}}}_i(t)$ denote the Markov-modulated queueing process where (i) offline nodes of type $i$ arrive with rate $\lambda_i$ and abandon with rate $\mu_i$, and (ii) the queue is depleted (when $Q^{\textup{\textsf{weak}}}_i(t) > 0$)  at rate $\sum_{j\in N_i} \textup{\textsf{TE}}_j(t) \cdot \gamma_j$.
\end{enumerate}
 We let $\pi^{\textup{\textsf{weak}}}$ denote the stationary distribution of $Q^{\textup{\textsf{weak}}}$.
\end{definition}

% \paragraph{Notation.} For convenience, we let $\Pr_{\textup{\textsf{weak}}}[\cdot ]$ and $\mathbb{E}_{\textup{\textsf{weak}}}[\cdot ]$ denote $\Pr_{Q^{\textup{\textsf{weak}}} \sim \pi^{\textup{\textsf{weak}}}}[ \cdot ]$ and $\mathbb{E}_{Q^{\textup{\textsf{weak}}} \sim \pi^{\textup{\textsf{weak}}}}[ \cdot ]$, respectively. \\

Note that when $\textup{\textsf{TE}}_j(t)=1$, an arriving type-$j$ will have its neighboring  \textup{\underline{T}OP} queues  \textup{\underline{e}mpty} (leading to our notation ``$\textsf{TE}$''). In such a scenario, under \Cref{alg:second} the arrival will search for a match in BOT. In contrast, when $\textsf{TE}_j(t) = 0$, i.e., $j$'s neighbors in TOP have non-empty queues, $j$ is guaranteed to match to someone in TOP. In our definition, the depletion rate for some offline node $i \in \text{BOT}$ under the weakly correlated Markov chains considers that arrivals of type $j \in N_i$ only contribute to matches with $i$ when $\textsf{TE}_j(t) = 1$.  It is thus straightforward to see that the weakly correlated Markov chains satisfy the following stochastic dominance relationship.

% Contrast this to the independent Markov chains, where depletion rate of $i \in \text{BOT}$ would be $\Gamma_i = \sum_{j \in N_i} \gamma_j$. 

% Neither of these processes describes the exact evolution of the system. The the independent Markov chains corresponds to a system where each offline type has a disjoint stream of online nodes of type $j$, arriving with rate $\gamma_j$ for every $j \in J$. In other words, we have $n$ identical copies of each online type $j$. The weakly correlated Markov chain, however, is closer to the actual system. It ``imagines'' duplicate online types for the depletion of TOP queues; for the depletion of a queue $i \in \textup{BOT}$, however, it only considers the real type $j$, assuming that $i$ is prioritized over every other queue in BOT.
\begin{restatable}{claim}{claimstochasticdominancesecond} \label{claim:stochasticdominancesecond}
If $Q(t) = (Q_i(t))_{i \in I}$ denotes the evolution of queues of offline types under \Cref{alg:second}, we have $Q\succeq_{st} Q^{\textup{\textsf{weak}}}$ where $\succeq_{st}$ denotes stochastic dominance. 
\end{restatable}

A formal treatment of the stochastic dominance and proof is deferred to \Cref{proofofstochasticdominancesecond}. We note it is also easy to see that $Q^{\textsf{weak}}$ stochastically dominates the independent Markov chains for \Cref{alg:second}; hence the weakly correlated Markov chains give a middle ground between independent queues and the full (potentially complex) dynamics introduced by \Cref{alg:second}. As an immediate result of this stochastic dominance, it suffices to show that under the weakly correlated Markov chains, the likelihood that an arriving node of type $j$ sees only empty neighboring queues is strictly smaller than $1/e$. 

\begin{lemma} \label{prop:proof main}
 There exists a universal constant $\zeta > 0$, independent of $\eps$, such that for every $j\in [m]$ we have $$\prtwo{Q \sim \pi^\textup{\textsf{weak}}}{\sum_{i\in N_j} Q_i=0 }\lesssim_{\eps} 1/e - \zeta \  .$$
\end{lemma}

Recall our previously defined notation $\lesssim_{\eps}$, denoting that the inequality is true up to a $1+C \cdot \eps'$ factor for some universal constant $C$. With Lemma~\ref{prop:proof main} in place, the proof of \Cref{prop:main} immediate. Indeed, stochastic dominance implies that the expected gain of \Cref{alg:second} is no lower when replacing the true stationary distribution with that of the weakly correlated Markov chains (as in \Cref{app:factstochasticdom}). Then, by the PASTA property of Poisson arrivals (\Cref{lem:pasta}) and the greedy nature of \Cref{alg:second}, we have that
    the stationary gain of \Cref{alg:second} is lower bounded by
    \begin{align*}
    \sum_j \gamma_j \cdot \prtwo{Q \sim \pi^\textup{\textsf{weak}}} { \sum_{i\in N_j} Q_i \ge 1 } &\ge (1 - (1/e - \zeta) \cdot (1 + C \cdot \epsilon')) \cdot \sum_j \gamma_j \\
    &\ge (1 - 1/e + \zeta) \cdot (1 - C \cdot \epsilon') \cdot \sum_j \gamma_j && \\
    &\ge (1 - 1/e + \zeta) \cdot (1 - C \cdot \epsilon') \cdot (1 - \tilde{f}(\eps)) \cdot \text{OPT}\text{\eqref{TLPon}} \ ,
    \end{align*}
where the final line used \Cref{obs:smallLPloss}. Hence the remainder of this section is devoted to the proof of \Cref{prop:proof main}. Accordingly, fix an online type $j$. It is easy to bound the stationary probability that $\textup{\textsf{TE}}_j  = 1$ (i.e., that $j$ has no neighbor available in TOP) because in the weakly correlated Markov chains, the queues in TOP evolve as independent birth-death processes, as in our analysis in \Cref{sec:oneminusoneovere}. For convenience of notation we let $t^\infty$ denote some time when $Q^{\textsf{weak}}$ is in steady-state; whenever a time is not specified we will assume it to be at  $t^\infty$. For example $Q^{\textsf{weak}}_i, \textsf{TB}_j$ refer to $Q^{\textsf{weak}}_i(t^{\infty}), \textsf{TB}_j(t^{\infty})$, respectively. As $\textsf{TE}_j$ is a binary variable, we will also slightly abuse notation throughout by letting it also refer to the event that $\textsf{TE}_j = 1$, as below.

\begin{restatable}{claim}{claimPrTEj} \label{claim:PrTEj}
    We have $ \pr{ \textup{\textsf{TE}}_j  } \simeq_{\epsilon} e^{-0.5}.$
\end{restatable}
\begin{proof}
    Indeed, using the independence between queues in TOP in the weakly correlated Markov chains, and \text{\Cref{claim:stationarydistbirthdeath}}, we can compute 
    $$ 
        \pr{ \textup{\textsf{TE}}_j  } = \prod_{i \in N_j^\uparrow} \pr{ Q_i = 0  } 
        = \prod_{i \in N_j^\uparrow} \exp \left( - \frac{\lambda_i}{\mu_i + \Gamma_i} \right). $$
    From inequality \text{\eqref{ineq:abandonment_bound}} we know $\frac{\lambda_i}{\mu_i + \Gamma_i} \simeq_{\eps} \frac{x_{i,a}}{\mu_i}$ and from inequality \text{\eqref{eqn:jBalancedness}} we know $\sum_{i \in N_j^{\uparrow}} \frac{x_{i,a}}{\mu_i} \simeq_{\eps} 0.5$; these together imply the claim. 
\end{proof}

 The remaining challenge is to bound the probability that all of $j$'s neighboring queues in BOT are empty, conditioned on those in TOP being empty. To this end, we define the following random variables. 

\begin{enumerate}[label=(\roman*)]

\item {\em Presence indicator:} For each $i \in I$, indicate by $A_i(t)\in \{0,1\}$ whether an offline node of type $i$ is present, i.e., has arrived and not departed yet due to its wait time expiring. Note that $A_i(t) \geq Q_i^{\textsf{weak}}(t)$, with strict inequality if a type $i$ node is present but it has been matched.
\end{enumerate}

As $A_i$ is binary, we will again abuse notation and let it additionally refer to the event that $A_i = 1$ when writing ``$\Pr[A_i]$''. In the weakly correlated Markov chains, the arrivals and wait times of each  type $i\in \bott$ are independent of $\{Q_i^{\textup{\textsf{weak}}}: i\in \topp\}$ and the arrivals of online nodes. By \hyperref[propertyv]{Property (v)} of our instance (binary queues) the probability that $A_i$ is 1 is equal to
\begin{align}
\Pr \left[ A_i   \mid \textup{\textsf{TE}}_j   \right] = \Pr[  A_i ]= \frac{\lambda_i}{\mu_i+\lambda_i}. \label{probAi1}
\end{align} We also observe that these indicators are independent across queues. Hence, for every $I' \subseteq I$, 
\begin{align}
    \Pr \left[ \bigwedge_{i \in I'} A_i \; \Big| \; \textup{\textsf{TE}}_j  \right] = \prod_{i \in I'} \frac{\lambda_i}{\mu_i+\lambda_i}. \label{probAi2}
\end{align}
Here, we slightly abuse notation as $A_i$ and $\textsf{TE}_j$ stand for the events $\{A_i =1\}$ and $\{\textsf{TE}_j = 1\}$. We use a similar convention for binary random variables throughout the proof.

Next we define notation for the time spent since the last type $i$ arrival.
\begin{enumerate}[label=(\roman*)]
\setcounter{enumi}{1}
\item {\em Time since arrival:} Let $t_i$ be the random time that has elapsed since the latest arrival of type $i$, going backwards in time from $t^\infty$.
\end{enumerate}

 Since the process $\{A_i(t)\}_t$ is time-reversible we have $t_i \sim \text{Exp}(\lambda_i)$. Note also that when conditioning on the event $A_i$ (i.e., the event $A_i(t^{\infty}) = 1$) the distribution of $t_i$ is given by $\text{Exp}(\lambda_i + \mu_i)$. Our final piece of notation captures extends the notion of time since arrival to incorporate when a neighboring online type sees empty neighboring queues in TOP. 

 \begin{enumerate}[label=(\roman*)]
\setcounter{enumi}{2}

\item {\em Time since $i$'s arrival with $k \in N_i$ seeing empty queues in TOP:} For each $i \in I$ and $k \in N_i$, define the random variable $\theta_{i, k} = \int_{t^\infty-t_i}^{t^\infty} \textup{\textsf{TE}}_k(t) \, dt$.  
\end{enumerate} 

With this notation in place, we write the following expression for the conditional probability that all of the bottom queues neighboring $j$ are empty. 

\begin{restatable}{lemma}{jointproballqueuesempty} \label{lem:jointproballqueuesempty}
    We have
    $$ \Pr \left[ \sum_{i \in N_j^{\downarrow}} Q_i^{\textup{\textsf{weak}}} = 0  \, \middle | \,\textup{\textsf{TE}}_j \right] = \mathbb{E}  \left[ \prod_{i \in N_j^{\downarrow}} \left( 1 - \frac{\lambda_i}{\mu_i + \lambda_i} \cdot \exp \left( - \sum_{k \in N_i} \gamma_k \cdot \theta_{i,k} \right)\right) \middle | \bigwedge_{i \in N_j^{\downarrow}} A_i, \textup{\textsf{TE}}_j \right].$$
\end{restatable}

The full proof is deferred to \Cref{app:approximation}. Intuitively, the lemma follows from the observation that $Q_i^{\textup{\textsf{weak}}}$ will be non-empty at $t^{\infty}$ only if $A_i = 1$ (i.e., a node of type $i$ is present) and if since the latest arrival of $i$, queue $i$ was not depleted in the weakly correlated Markov chains. As mentioned above, the $(A_i)_i$ indicators are independent of $\textup{\textsf{TE}}_j$ and of each other, with $\Pr[A_i] = \lambda_i \cdot (\mu_i + \lambda_i)^{-1}$. The total depletion rate of $Q_i^{\textup{\textsf{weak}}}$ in the interval $[t^{\infty} - t_i, t^{\infty}]$ is $\int_{t^{\infty}- t_i}^{t^{\infty}} \sum_{k \in N_i}\gamma_k\cdot \textup{\textsf{TE}}_k(t) dt = \sum_{k \in N_i} \gamma_k \cdot \theta_{i,k}$, so the conditional probability that $Q_i^{\textup{\textsf{weak}}}$ avoids being depleted if present at $t^{\infty}$ is $\mathbb{E}   \left[ \exp \left( - \sum_{k \in N_i} \gamma_k \cdot \theta_{i,k} \right) \mid \textup{\textsf{TE}}_j, A_i \right]$. The precise form of \Cref{lem:jointproballqueuesempty} takes into account that $\theta_{i,k}$'s may be correlated across different $i$'s.

To get some intuition about the expression in the right-hand side of \Cref{lem:jointproballqueuesempty}, we note that a simple upper bound is obtained by using the fact that $\theta_{i,k} \le t_i$. Indeed, when we replace $\theta_{i,j}$ by $t_i$, the right-hand side in Lemma 3.11  is straightforward to calculate, because $t_i$ is independent of $\textup{\textsf{TE}}_j$ and $(A_{i'})_{i' \neq i}$. Additionally, when conditioning on $A_i = 1$, we have $t_i \sim \text{Exp}(\lambda_i + \mu_i)$. Therefore, the moment generating function of the exponential distribution gives
\begin{align}
    \mathop{\mathbb{E}}_{t_i \sim \text{Exp}(\lambda_i + \mu_i)} \left[ \frac{\lambda_i}{\mu_i + \lambda_i} \cdot   \exp\left(- t_i\cdot \left(\sum_k \gamma_k\right)   \right) \middle | \textup{\textsf{TE}}_j, A_i \right]
    &= \frac{\lambda_i}{\mu_i + \lambda_i} \cdot \frac{\lambda_i + \mu_i}{\lambda_i + \mu_i + \Gamma_i} \label{MGFequality}  \\
    &\ge \frac{\lambda_i}{\epsilon ^2 \mu_i + \mu_i+\Gamma_i} \nonumber && \text{by \hyperref[propertyv]{Property (v)}} \\
    &\gtrsim_{\eps}  \frac{x_{i,a}}{\mu_i} \ ,\label{MGFcalculation}
\end{align}
where in the final line we used inequality \eqref{ineq:abandonment_bound}. 
By \Cref{lem:jointproballqueuesempty} $$\pr{   \bigwedge_{i \in N_j^{\downarrow}} Q_i^{\textup{\textsf{weak}}}  = 0 \, \middle | \, \textup{\textsf{TE}}_j } \le \prod_{i \in N_j^{\downarrow}} \left( 1 - \frac{x_{i,a}}{\mu_i} \cdot (1 - O(\eps)) \right) \lesssim_{\eps} \exp \left( - \sum_{i \in N_j^{\downarrow}} \frac{x_{i,a}}{\mu_i} \right) \lesssim_{\eps}  \exp(-0.5) \ ,$$ where the final bound follows from inequality \eqref{eqn:jBalancedness}. Combining this inequality with \Cref{claim:PrTEj}, we recover our $1-1/e$ bound. To obtain a constant-factor boost, we argue that $\theta_{i,k}$ has a constant probability of being significantly smaller than $t_i$. Formally, we show this is the case when $t_i$ takes values in the tail its distribution, in particular values at least $\Omega(\Gamma_i^{-1})$.
\begin{lemma}\label{lem:main}
 There exist universal constants $c, \delta \in (0,1)$, independent of $\epsilon$, such that for each $i \in N_j^\downarrow$, we have
\[ \pr{\left.\sum_{k\in  N_i}\gamma_k\theta_{i,k}\leq (1-c) \cdot t_i \cdot \Gamma_i  \right|  t_i, \textup{\textsf{TE}}_j, \bigwedge_{i' \in N_j^{\downarrow}} A_{i'}} \geq {c} \cdot \mathbbm{1}[t_i \geq \tau_i^*] \ ,
\]
where $\tau_i^* =  ( \delta \cdot \Gamma_i )^{-1}$.\footnote{Recall that $\textup{\textsf{TE}}_j$ and $A_{i'}$ denote events, while $t_i$ is a random variable.}
\end{lemma}

%To unpack \Cref{lem:main}, suppose we were interested in lower bounding the left-hand side without  conditioning on $\textsf{TE}_j$, i.e. $\pr{\left.\sum_{k\in  N_i}\gamma_k\theta_{i,k}\leq (1-c) \cdot t_i \cdot \Gamma_i  \right|  t_i,(A_{i'})_{i' \in N_j^\downarrow}}$ . Recalling that we have $\gamma_k \theta_{i,k} = \int_{t^\infty-t_i}^{t^\infty} \gamma_k \cdot \textsf{TE}_k(u) \, du$, it suffices to argue that $\ex{\textsf{TE}_k(u) | t_i, A_{i'}} \lesssim_{\eps} \exp(-0.5)$ analogously to \Cref{claim:PrTEj}. This shows that the top offline nodes block the depletion of $i$ by each online node $k \in N_i$ a constant fraction of the time. The time spent in such ``blocking'' events then translates into a boost of availability for the focal node $i$.

The proof of \Cref{lem:main} is deferred to Section~\ref{sec:lem:main}. Due to conditioning on the event $\textsf{TE}_j$,the  proof requires a fine-grained analysis of the correlation structure between $\theta_{i,k}$'s. Informally speaking, we use concentration bounds to show that the mixing time of the queues is typically smaller than $t_i$, so the initial conditioning on $\textup{\textsf{TE}}_j$ washes away relatively quickly when we go backwards in time from $t^\infty$.  The crucial ingredients are bounds on the first-order and second-order moments of the cumulative time where $k$'s neighbors in TOP have empty queues, for each $k\in N_i$.

To close this section, we prove that \Cref{lem:main} allows us to establish \Cref{prop:proof main}. Note that the random durations $\theta_{i,k}$ could be positively correlated, including with the event $\textup{\textsf{TE}}_j$. To handle such correlations, we rely on the Hardy-Littlewood inequality. In our setting, it intuitively says that the product of non-negative random variables with fixed marginal distributions is maximized when they are maximally positively associated.

\begin{fact}[c.f. \cite{burchard2006rearrangement}]\label{fact:rearrangement}
Let $X_1, X_2, \ldots, X_n \geq 0$ be random variables with fixed marginal CDFs $\{ F_{X_i}(\cdot) \}$. Then, we have $$\mathbb{E}   \left[ \prod_i X_i \right] \le \int_0^1 \left( \prod_i F_{X_i}^{-1}(x) \right)\, dx.$$ 
\end{fact}

\paragraph{Proof of \Cref{prop:proof main}.} By \Cref{claim:PrTEj} and \Cref{lem:jointproballqueuesempty}, we have 
\begin{align}
    \pr {\sum_{i \in N_j} Q_i^{\textup{\textsf{weak}}}  = 0} &= \pr{\textup{\textsf{TE}}_j } \cdot \pr {\left. \sum_{i \in N_j^\downarrow} Q_i^{\textup{\textsf{weak}}} = 0 \right. \, \middle | \, \textup{\textsf{TE}}_j  } \nonumber \\
    &\lesssim_{\eps} e^{-0.5} \cdot \mathbb{E}   \left[ \prod_{i \in N_j^{\downarrow}} \left( 1 - \frac{\lambda_i}{\mu_i + \lambda_i} \cdot \exp \left( - \sum_{k \in N_i} \gamma_k \cdot \theta_{i,k} \right)\right) \, \middle | \, \bigwedge_{i \in N_j^{\downarrow}} A_i, \textup{\textsf{TE}}_j \right]. \label{ineq:semi_q_empty} 
\end{align}
To bound the right-hand side we will consider drawing the (conditional) $t_i$'s first. Recall that conditioned on $\wedge_{i \in N_j^{\downarrow}} A_i$ and $\textup{\textsf{TE}}_j$, the vector $(t_i)$ is formed by mutually independent random variables with distribution $\text{Exp}(\mu_i+\lambda_i)$ at each coordinate $i$. Let $F_i(\cdot, \tau_i)$ denote the conditional CDF of $\sum_{k \in N_i} \gamma_k \cdot \theta_{i,k}$ given $\textup{\textsf{TE}}_j$ and $t_i = \tau_i$. We will apply \Cref{fact:rearrangement} for each draw of $(\tau_i) \sim (\text{Exp}(\mu_i + \lambda_i))$, with $F_{X_i}(\cdot)$ taken as $F_i(\cdot, \tau_i)$, to get the following bound.
\begin{align}
   \pr{\sum_{i \in N_j} Q_i^{\textup{\textsf{weak}}} = 0} \lesssim_{\eps}  e^{-0.5  } \cdot \hspace{-0.5em} \mathop{\mathbb{E}}_{(\tau_i \sim \text{Exp}(\mu_i + \lambda_i))_{i}} \left[  \int_{0}^1\prod_{i\in N_j^\downarrow} \left(1- \frac{\lambda_i}{\mu_i+\lambda_i}\cdot \exp\left(- F_{i}^{-1}(u,\tau_i)\right)   \right) du   \right]. \label{eqn:proballempty}
\end{align}
Consequently, we use \Cref{lem:main} to incorporate the constant-factor gap between $\sum_{k\in  N_i}\gamma_k\theta_{i,k}$ and $\sum_{k\in  N_i}\gamma_k\tau_i$. This lemma can be restated as saying that the conditional distribution of $\sum_{k \in N_i} {\gamma_k} \theta_{i,k}$ given $\textup{\textsf{TE}}_j$ is stochastically dominated by that of
\begin{align}
t_i &\cdot \Gamma_i \left( 1 - \text{Ber}({c}) \cdot c \cdot \mathbbm{1}[ t_i \ge \tau_i^*] \right) \, \label{berstochdomdistribution}
\end{align}
where  $\text{Ber}({c}) $ denotes an independent Bernoulli with success probability ${c}$. We will let    $\tilde{F}_i(\cdot, \tau_i)$ denote the conditional CDF of the distribution \eqref{berstochdomdistribution} given $t_i = \tau_i$, which corresponds to a step-function with two pieces. Clearly, the preceding discussion implies that, for every fixed $u$ and $\tau_i$, we have $F_i(u, \tau_i) \leq \tilde{F}_i(u, \tau_i)$. Therefore, 
\begin{align}
     &\extwo{ (\tau_i \sim \text{Exp}(\mu_i + \lambda_i))_{i}}{\left. \int_{0}^1\prod_{i\in N_j^\downarrow} \left(1- \frac{\lambda_i}{\mu_i+\lambda_i}\cdot \exp\left(- F_{i}^{-1}(u,\tau_i)\right)   \right) du \right.  } \nonumber \\
     & \quad \le \extwo{(\tau_i \sim \text{Exp}(\mu_i + \lambda_i))_{i}}{\left. \int_{0}^1\prod_{i\in N_j^\downarrow} \left(1- \frac{\lambda_i}{\mu_i+\lambda_i}\cdot \exp\left(- \tilde{F}_{i}^{-1}(u,\tau_i)\right)   \right) du \right.  } \nonumber \\
     & \quad = \underbrace{(1-{c})\cdot \mathbb{E}_{(\tau_i \sim \text{Exp}(\mu_i + \lambda_i))_{i}} \left[  \prod_{i\in  N_j^\downarrow} \left( 1- \frac{\lambda_i}{\mu_i+\lambda_i}\cdot \exp\left(-\tau_i \Gamma_i \right)\right)   \right] }_{(\dagger)} \nonumber \\
      & \quad \quad \quad \quad \quad + \underbrace{{c} \cdot \mathbb{E}_{(\tau_i \sim \text{Exp}(\mu_i + \lambda_i))_{i}} \left[   \prod_{i\in  N_j^\downarrow} \left( 1- \frac{\lambda_i}{\mu_i+\lambda_i}\cdot \exp\left(-\left(1-  c \cdot \mathbbm{1}\left[\tau_i \geq \tau_i^*\right]\right) \cdot \tau_i \Gamma_i \right)\right)  \right]  }_{(\star)}. \nonumber
\end{align}
To obtain the equality, we observe that the worst-case distribution after applying the Hardy-Littlewood inequality corresponds to ``perfectly correlated'' two-outcome random variables with probabilities $c$ and $1-c$. Analogously to equation \eqref{MGFcalculation}, we can calculate
\begin{align}
\mathbb{E}_{(\tau_i \sim \text{Exp}(\mu_i + \lambda_i))_{i}} \left[ 1 - \frac{\lambda_i}{\mu_i + \lambda_i} \cdot \left. \exp\left(- \tau_i \Gamma_i \right)  \right. \right] 
    = 1  - \frac{\lambda_i}{ \mu_i + \lambda_i +  \Gamma_i} \le \exp \left( - \frac{\lambda_i}{\mu_i + \lambda_i + \Gamma_i} \right)\ . \label{easybound}
\end{align}

To argue that $(\star)$ gets a better bound than this calculation, we will focus on offline nodes $i$ where $\Gamma_i$ is not dominated by $\mu_i$. Intuitively, if $\mu_i \gg \Gamma_i$, then $\tau_i \sim \text{Exp}(\mu_i + \lambda_i)$ may have only a very small probability of exceeding $\tau_i^*$, which we recall is on the order of $\Gamma_i^{-1}$. Consequently, we focus on a specific class of offline types $i$, whose incident arrival rate $\Gamma_i$ is not dominated by $\mu_i$. 
\begin{definition}
    For some small universal constant $\eta > 0$ which we fix later, independent of $\epsilon$, define $I_j \subseteq N_j^\downarrow$ as the subset of offline types $i$ such that $\Gamma_i \geq \eta \cdot \mu_i$,  
\end{definition}
It is not hard to see that under our parameter assumptions, for sufficiently small $\eta$ some constant fraction of $j$'s LP mass is incident to types in $I_j$. This is formalized in the following claim, whose proof is deferred to \Cref{app:proofofclaimimportanti}.
\begin{restatable}{claim}{claimimportanti}\label{clm:important_i}
    For any $\eps', \eta < 10^{-3}$ we have  $\sum_{i \in I_j} \frac{x_{i,a}}{\mu_i} \geq  0.1$. 
\end{restatable}
We can now bound $(\star)$ using the following inequality for $i \in I_j$. 
\begin{restatable}{claim}{clmIjbound} \label{clm:Ij_bound}
    There exists a universal constant $\tilde{c} > 0$, independent of $\epsilon$, such that for each $i \in I_j$ and $\tau_i \sim \textup{Exp}(\lambda_i + \mu_i)$, 
    \[ 
        \ex{\exp\left(-\left(1- c \cdot \mathbbm{1}\left[\tau_i \geq \tau_i^*\right]\right) \cdot \tau_i \Gamma_i \right)} \geq (1 + \tilde{c}) \cdot \frac{\mu_i + \lambda_i}{\mu_i + \lambda_i + \Gamma_i} \ .
    \]
\end{restatable}

The proof of \Cref{clm:Ij_bound} follows from a direct calculation and hence is deferred to \Cref{app:proofclaimIjbound}. We are now ready to complete the proof of \Cref{lem:main}. By inequality \eqref{easybound}, we have 
\begin{align*}
    (\dagger) &\le (1 - {c}) \cdot \exp \left( - \sum_{i \in N_j^{\downarrow}  } \frac{\lambda_i}{\mu_i + \lambda_i + \Gamma_i} \right).
\end{align*}
On the other hand, we upper bound $(\star)$ using \Cref{clm:Ij_bound}:
\begin{align*}
    (\star) &\leq {c} \cdot \prod_{i\in  N_j^\downarrow \setminus I_j} \ex{1- \frac{\lambda_i}{\mu_i+\lambda_i}\cdot \exp\left(-\tau_i \Gamma_k \right)}  \prod_{i\in I_j} \ex{1- \frac{\lambda_i}{\mu_i+\lambda_i}\cdot \exp\left(- (1-c \cdot \mathbbm{1}[\tau_i \geq \tau_i^*]) \cdot \tau_i \Gamma_k \right)} \\
    &\le {c} \cdot \exp\left(-\sum_{i \in N_j^\downarrow \setminus I_j} \frac{\lambda_i}{\mu_i + \lambda_i + \Gamma_i} \right) \cdot \exp\left(- (1 + \tilde{c}) \cdot \sum_{i \in I_j} \frac{\lambda_i}{\mu_i + \lambda_i + \Gamma_i} \right)  \hspace{4em} \text{by \Cref{clm:Ij_bound}} \\
    &= {c} \cdot \exp\left(-\sum_{i \in N_j} \frac{\lambda_i}{\mu_i + \lambda_i + \Gamma_i}\right)
    \cdot \exp\left(- \tilde{c}  \cdot \sum_{i \in I_j} \frac{\lambda_i}{\mu_i + \lambda_i + \Gamma_i}\right).
\end{align*}
By inequality \eqref{MGFcalculation}, we know that   $\frac{\lambda_i}{\mu_i + \lambda_i + \Gamma_i} \gtrsim_{\eps}  \frac{x_{i,a}}{\mu_i}$. Furthermore, $\sum_{i \in I_j} \frac{x_{i,a}}{\mu_i} \ge 0.1$ by \Cref{clm:important_i}.
By combining the preceding inequalities, we obtain
\begin{align}
     (\star) &\lesssim_{\eps}  {c} \cdot \exp\left(-\sum_{i \in N_j} \frac{\lambda_i}{\mu_i + \lambda_i + \Gamma_i}\right) \cdot \exp \left( -  \tilde{c}    \cdot 0.1 \right). \label{ineq:starbound}
\end{align}

To conclude, we define  $f(\tilde{c}, {c}) :=c \cdot (1 -     \exp \left(- \tilde{c}  \cdot  0.1 \right) $, and simplify using inequality~\eqref{eqn:proballempty}:
\begin{align*}
    \pr{\sum_{i \in N_j} Q_i^{\textup{\textsf{weak}}} = 0} &\lesssim_{\eps} e^{-0.5} \cdot ((\dagger) + (\star)) \\
    &\simeq_{\eps} e^{-0.5 } \cdot  \exp\left(-\sum_{i \in N_j^\downarrow} \frac{\lambda_i}{\mu_i + \lambda_i +  \Gamma_i} \right) \cdot  f(\tilde{c}, {c})  \\
    &\simeq_{\eps} e^{-0.5} \cdot \exp\left(-\sum_{i \in N_j^\uparrow} \frac{x_{i,a}}{\mu_i}\right) \cdot  f(\tilde{c}, {c})  \\
    &\simeq_{\eps} e^{-1} \cdot  f(\tilde{c}, {c})  &&\text{by \eqref{eqn:jBalancedness}.}
\end{align*}
Recalling that $c \in (0,1)$ and $\tilde{c} >0$, we observe that $f(\tilde{c}, c)$ is some universal constant smaller than 1, which completes the proof of \Cref{prop:proof main}.

% % \Alirezaedit{For $T \subseteq I$ and $S \subseteq J$, let $\textsf{LP}_{S,T} = \sum_{i \in T} \sum_{j \in S} \frac{x_{i,j}}{\gamma_j}$. 
  
\subsection{Proof Outline of Lemma~\ref{lem:main}}
 \label{sec:lem:main}
% Note that the proportion of time spent in a blocked cycle $\textup{\textsf{TE}}_j(t)=1$ is equal to $\prpartwo{s}{\textup{\textsf{TE}}_j(t)=1} \leq e^{-\frac{1}{2}+\eps}$ from inequality~\eqref{ineq:balanced}. From renewal theory, this is also equal to $\frac{\expar{\eta^{\rm e}_{j}}}{\expar{\eta^{\rm e}_{j}+\eta^{\rm n}_{j}}}$ where $\eta^{\rm e}_{j},\eta^{\rm n}_{j}$ are the random durations for one busy cycle and one non-busy cycle, respectively. Since $\expar{\eta^{\rm e}_{j}} = \frac{1}{\sum_{i\in N_j^\uparrow}\lambda_i}$, we obtain 
% \begin{eqnarray}
% \expar{\eta^{-}_j} \geq \frac{1 -  e^{-\frac{1}{2}+\eps}}{ e^{-\frac{1}{2}+\eps}} \cdot \frac{1}{\sum_{i\in N_j^\uparrow}\lambda_i} \ .
% \end{eqnarray}
This subsection outlines the main elements of the proof of Lemma~\ref{lem:main}; some details are deferred to \Cref{app:approximation}. Recall that \Cref{lem:main} considers fixed types $j \in J, i \in N_j^\downarrow$, and bounds the aggregate amount of time in $[t^{\infty} - t_i, t^{\infty}]$ that online types $k \in N_i$ see all their neighboring queues in TOP empty, conditioned on the queues in $N_j^\uparrow$ being empty at time $t^\infty$. Hereafter, saying that a subset $S \subseteq I$ is empty means that all the queues of a type in $S$ are empty. Recall that the random variable $\theta_{i, k}$ tracks for each $k \in N_i$, the amount of time, since $i$'s most recent arrival, that $N_k^\uparrow$ was empty. We would like to show that conditioned on $N_j^\uparrow$ being empty and $\wedge_{i' \in N_j^\downarrow} A_{i'}$, with a constant probability, there is a constant-factor gap between $\sum_{k \in N_i} \gamma_k \theta_{i,k}$ and $t_i \Gamma_i$. We consider two qualitatively distinct cases depending on the degree of overlap between $j$'s neighborhood and the neighborhoods of $k \in N_i$, under our LP solution. 

In the easier case, we consider online types $k$ that have a somewhat different neighborhood to that of $j$. Let the set of these online nodes be $J^{\rm indep}$. In this case, we might intuitively expect that conditioning on $N_j^\uparrow$ being empty does not affect much the distribution of $\theta_{i,k}$, as there is little overlap between $N_j^\uparrow$ and $N_k^\uparrow$. Hence, we consider each $k \in J^{\rm indep} \cap N_i$ independently by focusing on set $I_k = N_k^\uparrow \setminus N_j^\uparrow$, since $I_k$ is independent of $N_j^\uparrow$. We then analyze the expected duration in which $I_k$ is empty. By combining this result for different $k \in J^{\rm indep} \cap N_i$, we obtain our result. 

In the harder case, however, node $k$'s neighborhood could have a large degree of overlap with that of $j$, in which case there is a significant correlation between $\theta_{i,k}$ and $\theta_{i,j}$; this correlation makes the conditioning on $\textsf{TE}_j$ highly influential. Let the set of such nodes $k$ be $J^{\rm correl}$. We jointly analyze nodes in $J^{\rm correl}$ by focusing on a specific set $I^{\rm core} \subseteq N_j^\uparrow$. We show that a set $I^{\rm core}$ exists such that (i) it has a large arrival rate from its neighbors in $J^{\rm correl} \cap N_i$, and (ii) $(\sum_{i' \in I^{\rm core}} \lambda_{i'}) / (\sum_{k \in J^{\rm correl} \cap N_i} \gamma_j)$ is lower bounded by a constant. Given these properties of $I^{\rm core}$ and by employing tools from renewal theory, we analyze the duration in which $I^{\rm core}$ is empty and use it to bound $\sum_{k \in N_i} \gamma_k \theta_{i,k}$.

% \todo{Write more about the different techniques used.}

First, we formalize the notions of $J^{\rm indep}$ and $J^{\rm correl}$, schematically visualized in \Cref{fig:J_correl_indep}. 
\begin{definition}
    Let $0 < \kappa < 0.1$ denote a fixed constant that we will choose later. We then define $$J^{\rm indep} := \left\{k\in J: \sum_{i\in N_k^\uparrow\setminus N_j^\uparrow} \frac{x_{i,a}}{\mu_i} \geq \kappa \right\} \quad \text{and} \quad J^{\rm correl} := J \setminus J^{\rm indep} \ . $$  
\end{definition} 

\begin{figure}
    \centering
\begin{center}
\scalebox{0.7}{\begin{tikzpicture}[
    >=Stealth,
    node distance=1cm and 2cm,
    scale = 1.2,
    every node/.style={font=\large},
    dot/.style={circle, fill, inner sep=1.5pt}
]
% Titles
\node[anchor=north west, font=\LARGE] (offline) at (0.4,7) {Offline};
\node[anchor=north west, font=\LARGE] (online) at (7,7) {Online};
% Boxes for Offline and Online sections (unfilled, thicker)
\draw[lightgray, line width=2pt] (-1.4,-2.2) rectangle ++(5.5,8.5);
\draw[lightgray, line width=2pt] (5,-2.2) rectangle  ++(5.5,8.5);
% TOP and BOT labels with left-facing brackets (larger font)
\draw[decorate, decoration={brace, amplitude=10pt, mirror, raise=2pt}, thick]
    (0.35,6) -- (0.35,1.1) node[midway, left=12pt, font=\LARGE] {TOP};
\draw[decorate, decoration={brace, amplitude=10pt, mirror, raise=2pt}, thick]
    (0.35,0.9) -- (0.35,-2) node[midway, left=12pt, font=\LARGE] {BOT};
% Dotted line between TOP and BOT (thicker)
\draw[dotted, line width=1pt] (0.4,1) -- (4.08,1);
% Offline labels
\node[anchor=east, text=orange, font=\Large] (Njd) at (0.9,-0.5) {$N_j^{\downarrow}$};
\node[anchor=east, text=darkgreen, font=\Large] (Nku) at (0.9,5) {$N_k^{\uparrow}$};
\node[anchor=east, text=BrickRed, font=\Large] (Nju) at (0.9,3) {$N_j^{\uparrow}$};
\node[anchor=east, text=blue, font=\Large] (Nkprimeu) at (0.9,2) {$N_{k'}^{\uparrow}$};
% Online column
\node[dot] (k) at (6,5) {};
\node[anchor=west, font=\LARGE] (k_label) at ($(k)+(0.2,0)$) {$k \in J^{\text{indep}}$};
\node[anchor=west, text width=5cm, font=\small] (k_text) at (6.2,4.2) {($N_k^{\uparrow}$ and $N_j^{\uparrow}$ have significant non-overlapping sections)};
\node[dot] (k_prime) at (6,3) {};
\node[anchor=west, font=\LARGE] (k_prime_label) at ($(k_prime)+(0.2,0)$) {$k' \in J^{\text{correl}}$};
\node[anchor=west, text width=5cm, font=\small] (k_prime_text) at (6.2,2.2) {($N_{k'}^{\uparrow}$ and $N_j^{\uparrow}$ have lots of overlap)};
\node[dot] (j) at (6,1) {};
\node[anchor=west, font=\LARGE] (j_label) at ($(j)+(0.2,0)$) {$j$};
\node[dot] (i) at (1.15,-0.5) {};
\node[dot] (a) at (1.15,3.1) {};
\node[dot] (b) at (1.15,2.8) {};
\node[dot] (c) at (1.15,2.5) {};
\node[dot] (d) at (1.15, 1.7) {};
\node[dot] (e) at (1.15,5) {};
\node[dot] (e) at (1.15,4.3) {};

\node[anchor=west, font=\LARGE] (i_label) at ($(i)+(0.1,-0.2)$) {$i$};
% Shaded triangles
\fill[opacity=0.1, color=orange] ($(Njd)+(0.35,0.7)$) -- ($(j)$) -- ($(Njd)+(0.35,-0.7)$) -- cycle;
\fill[opacity=0.1, color=BrickRed] ($(Nju)+(0.35,1)$) -- ($(j)$) -- ($(Nju)+(0.35,-1)$) -- cycle;
\fill[opacity=0.1, color=darkgreen] ($(Nku)+(0.35,1)$) -- ($(k)$) -- ($(Nku)+(0.35,-1.2)$) -- cycle;
\fill[opacity=0.1, color=blue] ($(Nkprimeu)+(0.35,1)$) -- ($(k_prime)$) -- ($(Nkprimeu)+(0.35,-1)$) -- cycle;
% Connecting lines
\draw[opacity=1] ($(i)$) -- ($(j)$);
\draw[opacity=1] ($(i)$) -- ($(k)$);
\draw[opacity=1] ($(i)$) -- ($(k_prime)$);
\end{tikzpicture}}\end{center}
\caption{Schematic visualization of ${J^{\rm correl}}$ and ${J^{\rm indep}}$}
    \label{fig:J_correl_indep}
\end{figure}
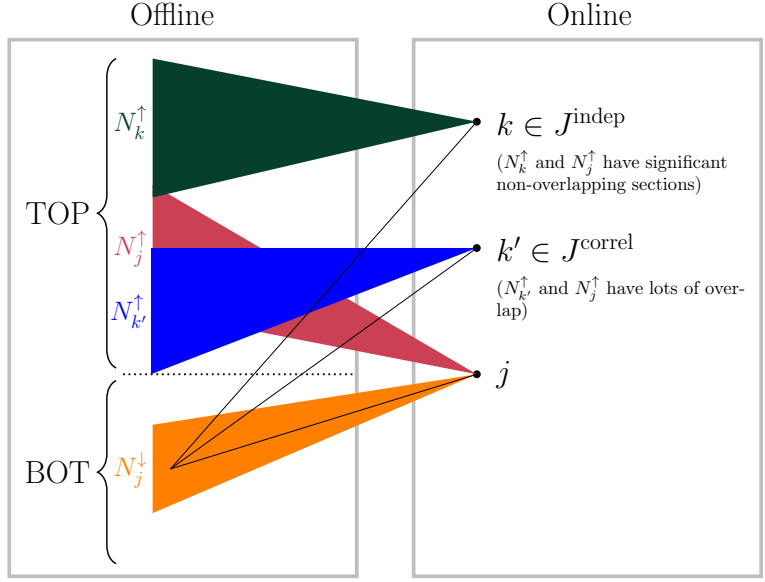

\paragraph{Proof of \Cref{lem:main}.} The proof proceeds by analyzing $J^{\rm correl}$ and $J^{\rm indep}$ separately, showing a constant probability of having a constant-factor gap between $\sum_{k\in {\cal J} \cap N_i}\gamma_k\theta_{i,k}$ and $t_i \cdot \sum_{k\in {\cal J} \cap N_i}\gamma_k$ for ${\cal J} \in \{J^{\rm correl}, J^{\rm indep}\}$. We begin  with proving this result for ${\cal J} = J^{\rm correl}$---the case of online types whose neighborhood has a high level of overlap with that of $j$. Next, we state the result for $J^{\rm indep}$. Since the proof is less challenging, we defer it to the appendix. Lastly, we combine the two results to conclude the proof of \Cref{lem:main}.
\begin{lemma} \label{lem:correl}
There exist constants $c_0, \bar{\delta} \in (0,1)$ independent of $\eps, \kappa$ such that if $\delta \leq \bar{\delta}$, for each $i \in N_j^\downarrow$, we have 
\[ \pr{\left.\sum_{k\in J^{\rm correl}\cap N_i}\gamma_k\theta_{i,k}\leq (1-c_0) \cdot t_i \sum_{k\in J^{\rm correl}\cap N_i}\gamma_k  \right| t_i, (A_{i'})_{i' \in N_j^\downarrow}, \textup{\textsf{TE}}_j } \geq c_0 \cdot \mathbbm{1}\left[t_i \geq \tilde{\tau}_i\right]  \ ,
\] where $\tilde{\tau}_i = ({\delta \cdot \sum_{k \in J^{\rm correl} \cap N_i}{\gamma_k}})^{-1}$. 
\end{lemma}
{\it Proof sketch.} Below, we sketch the main ideas of the proof with detailed derivations deferred to Appendix \ref{prf:correl}. To simplify the notation, we make the conditioning on the event $\{\textup{\textsf{TE}}_j, (A_{i'})_{i' \in N_j^\downarrow}\}$ implicit throughout this proof and assume that $t_i \geq \tilde{\tau}_i$ since the lemma is trivially true otherwise. % It turns out that conditioning on the event $(A_{i'})_{i' \in N_j^\downarrow}$ has no effect on our argument. 
For convenience, define $\tilde{\Gamma} := \sum_{k\in J^{\rm correl}\cap N_i}\gamma_k$. The proof of \Cref{lem:correl} proceeds in three steps. First, we establish the existence of a ``core set'' of offline nodes $I^{\rm core} \subseteq N_j^\uparrow$ such that (roughly) (i) its neighbors in $J^{\rm correl} \cap N_i$ have a large aggregate arrival rate, and (ii) $(\sum_{i' \in I^{\rm core}} \lambda_{i'}) / \tilde{\Gamma}$ is 
bounded below by a constant. In the second step, we leverage property (ii) of $I^{\rm core}$ to consider a sufficiently long horizon of empty and nonempty cycles, deriving concentration bounds on the aggregate empty and nonempty durations within that horizon. Lastly, we use property (i) of $I^{\rm core}$ to bound $\sum_{k \in J^{\rm correl} \cap N_i} \gamma_k \theta_{i,k}$ via the bounds on the total duration of empty and nonempty cycles (from step 2).

\vspace{0.2cm}

\noindent {\it Step 1: A core set of offline nodes.} We prove that the subgraph induced by $N^\uparrow_j$ and $J^{\rm correl} \cap N_i$ is (roughly speaking) highly connected. In particular, a counting argument in \Cref{app:core_set} shows that there exists a subset $I^{\rm core}\subseteq N_j^{\uparrow}$ such that 
\begin{eqnarray}\label{ineq:core_sum}
 \sum_{{i'} \in I^{\rm core}} \frac{x_{{i'}, a}}{\mu_{i'}} \geq u(\eps', \kappa) \ ,   
\end{eqnarray} for $u(\eps', \kappa):= {1-\eps'} - 2\kappa - \frac{0.5}{1-\eps'}$
and, for each ${i'} \in I^{\rm core}$,  
\begin{eqnarray} \label{ineq:core_gamma_lb}
\sum_{k\in J^{\rm correl} \cap N_i} \gamma_k \cdot \mathbbm{1}\left[{i'} \in N^\uparrow_k\right] \geq \frac{\tilde{\Gamma}}{2} \ .
\end{eqnarray}
It is straightforward to see that $u(\eps', \kappa) \geq 0.1$ for any $0 < \eps', \kappa < 0.1$. Letting $\tilde{\Lambda} = \sum_{i' \in I^{\rm core}} \lambda_{i'}$, \Cref{app:core_set} combines inequalities \eqref{ineq:core_sum} and \eqref{ineq:core_gamma_lb} to show that 
\begin{align}\label{ineq:u_lb}
\frac{\tilde{\Lambda}}{\tilde{\Gamma}} \geq \frac{(1-\eps')u(\eps', \kappa)}{2} \geq 0.045 \ .
\end{align}

\noindent  {\it Step 2: Analyzing the evolution of queues $I^{\rm core}$.} For any subset of offline nodes $S \subseteq \textup{TOP}$, we define an ``empty cycle'' as a maximal time interval under which every queue in the weakly correlated Markov chains is empty (i.e., $\sum_{i \in S} Q_i^{\textup{\textsf{weak}}} = 0$ during this interval); the remaining intervals are ``nonempty" cycles. By definition, an empty cycle ends with a new arrival of nodes in $S$ which indicates the start of a new nonempty cycle. Moreover, the durations of empty or nonempty cycle are IID which implies that the evolution of TOP queues, across the succession of cycles, forms a renewal process.

Let $\textsf{NonEmpty}(S)$ denote the distribution of the amount of time $S$ spends in a nonempty cycle. While characterizing $\textsf{NonEmpty}(S)$ is complicated, we provide a simple distribution that $\textsf{NonEmpty}(S)$ stochastically dominates.
\begin{claim} \label{clm:nonbusy}
$\textup{\textsf{NonEmpty}}(S)$ stochastically dominates the hyperexponential distribution $$ Z_S \sim 
    \Big( \textup{Exp}(\mu_{i'} + \Gamma_{i'}) \; \; \textup{with probability } \frac{\lambda_{i'}}{\sum_{i''\in S} \lambda_{i''}} \Big)_{i' \in S} \ .$$  
\end{claim}
\begin{proof}
    Each $i' \in S$ has its queue transition out of the empty state at rate $\lambda_{i'}$ in the weakly correlated Markov chains. Thus, if $S$ is currently empty, a nonempty cycle will start after time distributed as $\text{Exp} \left( \sum_{i' \in S} \lambda_{i'} \right)$; the type that causes this transition from empty to nonempty is $i'$ with probability $\lambda_{i'} \cdot \left( \sum_{i'' \in S} \lambda_{i''} \right)^{-1}$. Note that the nonempty cycle cannot end until the queue for $i'$ empties. As $i' \in \textup{TOP}$ and we assume binary queues, we know that the depletion rate is precisely $\mu_{i'} + \Gamma_{i'}$. 
\end{proof}

We now focus on the empty and nonempty cycles of $I^{\rm core}$. Since the Markov chain is time-reversible, we consider the backward process, starting from $t^\infty$ assuming that we start from an empty cycle (due to the conditioning on $\textup{\textsf{TE}}_j$). 
In particular, we consider a horizon of $T = \lceil \frac{\tilde{\Lambda}}{\delta' \tilde{\Gamma}} \rceil$ inter-renewal times, where $\delta' \leq 4\delta$ is a constant specified later. Concretely, we analyze $T$ empty and nonempty cycles, denoted by $\eta^{\rm e}_1,\eta^{\rm n}_1,\ldots,\eta^{\rm e}_T,\eta^{\rm n}_T$ respectively. The properties of $I^{\rm core}$ described in inequality \eqref{ineq:u_lb} imply $T \geq \frac{(1-\eps')u(\eps',\kappa)}{2\delta'} \geq \frac{0.045}{\delta'}$; taking $\delta'$ sufficiently small ensures a long-enough horizon. Using concentration bounds presented in \Cref{app:concentration_bounds}, and \Cref{clm:nonbusy}, we obtain 
\begin{align} \label{ineq:eta_minus_bound}
    \pr{\sum_{t=1}^T \eta^{\rm n}_t \geq \frac{1}{4\delta' \tilde{\Gamma}}} \geq \frac{1}{4} \cdot \frac{1}{\frac{8}{(1-\eps')u(\eps', \kappa)} \delta' + 1} \geq \frac{1}{4} \cdot \frac{1}{90\delta' + 1} \ ,
\end{align}
and
\begin{eqnarray} \label{ineq:eta_plus_bound}
    \pr{\sum_{t=1}^T \eta^{\rm e}_t\geq \frac{2}{\delta' \tilde{\Gamma}}} \leq \frac{1}{2} e^{-\frac{(1-\eps')u(\eps', \kappa)(1 -\ln 2)}{2\delta'}} \leq \frac{1}{2} e^{-\frac{0.01}{\delta'}}  \ .
\end{eqnarray}

\noindent {\it Step 3: Bounding $\sum_{k \in J^{\rm correl} \cap N_i} \gamma_k\theta_{i,k}$ via $\sum_{t=1}^T \eta^{\rm e}_t$ and $\sum_{t=1}^T \eta^{\rm n}_t$.} Finally, we argue that the bounds in inequalities \eqref{ineq:eta_minus_bound} and \eqref{ineq:eta_plus_bound} lead to the desired constant-factor gap between $\sum_{k \in J^{\rm correl} \cap N_i} \gamma_k\theta_{i,k}$ and $t_i \tilde{\Gamma}$. Recall that we assume $t_i \geq {1}/(\delta \tilde{\Gamma})$ from the lemma's hypothesis. We thus choose $\delta'$ such that $\frac{1}{\delta'} =  \frac{t_i \tilde{\Gamma}}{4} \geq \frac{1}{4\delta}$. Let $\bar{\delta} > 0$ be the constant that satisfies 
\begin{align}
    \frac{1}{4} \cdot \frac{1}{360\bar{\delta} + 1} - \frac{1}{2}\exp\left({-\frac{0.0025}{\bar{\delta}}}\right) = \frac{1}{32} \ . \label{eq:delta_bar}
\end{align}
Therefore, if $\delta \leq \bar{\delta}$, with probability at least 
\begin{align}
    \frac{1}{4} \cdot \frac{1}{\frac{16}{(1-\eps')u(\eps', \kappa)} \delta + 1} - \frac{1}{2} e^{-\frac{(1-\eps')u(\eps', \kappa)(1 -\ln 2)}{\delta}} \geq \frac{1}{4} \cdot \frac{1}{360\delta + 1} - \frac{1}{2}e^{-\frac{0.0025}{\delta}} \geq \frac{1}{32}  \label{ineq:good_prob}
\end{align}
we have 
\begin{align}\label{cond:eta_bounds}
    \sum_{t = 1}^T \eta_t^{\rm e} \leq \frac{2}{\delta' \tilde{\Gamma}} = \frac{t_i}{2} \quad \text{and} \quad \sum_{t = 1}^T \eta_t^{\rm n} \geq \frac{1}{4\delta'\tilde{\Gamma}} = \frac{t_i}{16} \ .
\end{align}
 We henceforth assume that conditions \eqref{cond:eta_bounds} hold. If at any time $t \in [t^\infty - t_i, t^\infty]$, $I^{\rm core}$ is nonempty, inequality \eqref{ineq:core_gamma_lb} implies 
 \begin{align}\label{ineq:unblocked_bound}
     \sum_{k \in J^{\rm correl} \cap N_i} \gamma_k \cdot \textup{\textsf{TE}}_k(t) \leq \frac{\tilde{\Gamma}}{2} \ ,
     \end{align}
regardless of which types in $I^{\rm core}$ are available. Therefore, it follows that
 \begin{align}
    \sum_{k\in J^{\rm correl}\cap N_i} \gamma_k\theta_{i,k} &\leq \max\left\{ t_i \tilde{\Gamma} - \left(t_i - \sum_{t=1}^T \eta_t^{\rm e}\right) \cdot \frac{\tilde{\Gamma}}{2}, t_i \tilde{\Gamma} - \left(\sum_{t=1}^T \eta_t^{\rm n}\right) \cdot \frac{\tilde{\Gamma}}{2}\right\} \label{ineq:first_theta} \\ & \leq \tilde{\Gamma} \cdot \max\left\{ \frac{3t_i}{4}, \frac{31t_i}{32}\right\} = \frac{31}{32} t_i \tilde{\Gamma}
     \ . \label{ineq:final_delta}
\end{align}
To unpack the first inequality, note that the first argument in the max expression follows from assuming that $t_i$ is smaller than $\sum_{t = 1}^T (\eta_t^{\rm e} + \eta_t^{\rm n})$. In this case, the nonempty duration is at least $t_i - \sum_{t=1}^T \eta_t^{\rm e} \geq \frac{t_i}{2}$ by inequality \eqref{cond:eta_bounds}. Consequently, inequality \eqref{ineq:unblocked_bound} yields
\[ 
    \sum_{k \in J^{\rm correl} \cap N_i} \gamma_k \theta_{i,k} = \int_{t^\infty - t_i}^{t^\infty} \sum_{k \in J^{\rm correl} \cap N_i}  \gamma_k \cdot \textup{\textsf{TE}}_k(t) dt \leq t_i \tilde{\Gamma} - \left(t_i - \sum_{t=1}^T \eta_t^{\rm e}\right) \cdot \frac{\tilde{\Gamma}}{2} \ .
\]
The second argument of the max expression in inequality \eqref{ineq:first_theta} corresponds to the reverse case that $t_i$ exceeds $\sum_{t = 1}^T (\eta_t^{\rm e} + \eta_t^{\rm n})$, in which case the nonempty time duration is at least $\sum_{t=1}^T \eta_t^{\rm n} \geq \frac{t_i}{16}$, from inequality \eqref{cond:eta_bounds}. Using a similar line of reasoning as in the previous case, we have
\[ 
    \sum_{k \in J^{\rm correl} \cap N_i} \gamma_k \theta_{i,k} = \int_{t^\infty - t_i}^{t^\infty} \sum_{k \in J^{\rm correl} \cap N_i}  \gamma_k \cdot \textup{\textsf{TE}}_k(t) dt \leq t_i \tilde{\Gamma} - \left(\sum_{t=1}^T \eta_t^{\rm n}\right) \cdot \frac{\tilde{\Gamma}}{2} \ .
\]
In light of inequalities \eqref{ineq:good_prob} and \eqref{ineq:final_delta} we conclude the proof, with $c_0 = \nicefrac{1}{32}$ and $\bar{\delta}$ specified in equality \eqref{eq:delta_bar}, that
\[
    \pr{\left.\sum_{k\in J^{\rm correl}\cap N_i}\gamma_k\theta_{i,k}\leq (1-c_0) \cdot t_i \sum_{k\in J^{\rm correl}\cap N_i}\gamma_k  \right| t_i, (A_{i'})_{i' \in N_j^\downarrow}, \textup{\textsf{TE}}_j } \geq c_0 \cdot \mathbbm{1}\left[t_i \geq \tilde{\tau}_i\right] \ ,
\]
for every $\delta \leq \bar{\delta}$ and $\tilde{\tau}_i = (\delta \cdot \sum_{k \in J^{\rm correl} \cap N_i} \gamma_k )^{-1}$. \hfill \qedsymbol

\paragraph{}Turning attention to the easier case of $J^{\rm indep}$, the following lemma demonstrates a constant-factor gap between $\theta_{i,k}$ and $t_i$ with a constant probability. The proof appears in Appendix \ref{prf:rand} where for each $k \in J^{\rm indep} \cap N_i$, we analyze the expected empty duration of $I^k = N_k^\uparrow \setminus N_j^\uparrow$, since the queues in $I^k$ are independent of $\textsf{TE}_j$ and evolve independently---identically to their evolution in $Q^{\textsf{ind}}$.
\begin{lemma} \label{lem:rand}
There exists a constant $c_1 \in (0,1)$, independent of $\eps$, such that for each $i \in N_j^\downarrow$, we have 
% $\prpar{\theta_{i,k} \geq (1-c_1) \tau_i  | t_i = \tau_i,\textup{\textsf{TE}}_j(t^\infty) =1} \leq 1-c_1$, implying 
% $\expar{\sum_{k\in J^{\rm indep}\cap N_i}\gamma_k\theta_{i,k} | t_i = \tau_i,\textup{\textsf{TE}}_j(t^\infty) =1} \leq (1-c_1) \sum_{k\in J^{\rm indep}\cap N_i}\gamma_k\tau_i $.
\[ 
\pr{\left.\sum_{k\in J^{\rm indep}\cap N_i}\gamma_k\theta_{i,k}\leq (1-c_1)\cdot t_i \cdot \sum_{k\in J^{\rm indep}\cap N_i} \gamma_k   \, \right| \, t_i, (A_{i'})_{i' \in N_j^\downarrow}, \textup{\textsf{TE}}_j } \geq c_1 \ .
\]
\end{lemma}

\paragraph{Completing the proof of \Cref{lem:main}.} Combining \Cref{lem:correl} and \Cref{lem:rand} completes our proof. Indeed, if $\sum_{k \in J^{\rm correl} \cap N_i} \gamma_k \geq {\Gamma_i}/{2}$, by having $\delta \leq \bar{\delta}$ (e.g., $\delta = \bar{\delta}/2$), with probability $c_0$, we have
\begin{align*}
    \sum_{k \in N_i}  \gamma_k \theta_{i,k} = \sum_{k \in J^{\rm indep} \cap N_i}  \gamma_k \theta_{i,k} + \sum_{k \in J^{\rm correl} \cap N_i} \gamma_k \theta_{i,k} &\leq t_i \cdot \sum_{k \in J^{\rm indep} \cap N_i} \gamma_i + (1-c_0) t_i \sum_{k \in J^{\rm correl} \cap N_i} \gamma_k 
    \\ &\leq \left(1-\frac{c_0}{2}\right) t_i \cdot \sum_{k \in N_i} \gamma_k \ ,
\end{align*}
where the inequality follows from \Cref{lem:correl}. Therefore, \Cref{lem:main} holds with $c = c_0/2$ and $\delta = \bar{\delta}/2$. Conversely, if $\sum_{k \in J^{\rm indep} \cap N_i} \gamma_k = \Gamma_i - \sum_{k \in J^{\rm correl} \cap N_i} \gamma_k \geq \Gamma_i / 2$, we can similarly show that \Cref{lem:main} holds with $c = c_1/2$ and $\delta = \bar{\delta}/2$. In conclusion, with $c = \min\{c_0, c_1\} / 2$ and $\delta = \bar{\delta}/2$, we have
\[ 
    \pr{\left.\sum_{k\in  N_i}\gamma_k\theta_{i,k}\leq (1-c) \cdot t_i \cdot \Gamma_i  \right|  t_i, \textup{\textsf{TE}}_j, \bigwedge_{i' \in N_j^{\downarrow}} A_{i'}} \geq {c} \cdot \mathbbm{1}[t_i \geq \tau_i^*] \ ,
\]
where $\tau_i^* =  ( \delta \cdot \Gamma_i )^{-1}$. The proof of \Cref{lem:main} is now complete. \hfill \qedsymbol

\color{black}

\section{Improved Competitive Ratio}

Although our main technical contribution is to obtain an improved approximation ratio, our approach can also be adapted to yield an improved competitive ratio against the offline optimum, improving upon the wresultork of \cite{patel2024combinatorial}, as formalized in the next theorem. 

\begin{restatable}{thm}{compratiothm} \label{thm:compratio}
    There exists a polynomial-time algorithm for the online stationary matching problem that achieves expected average reward at least a $(1-1/\sqrt{e}+\delta)$-factor of the optimal \underline{offline} algorithm, for some universal constant $\delta > 0$. 
\end{restatable}

The algorithm for \Cref{thm:compratio} uses an offline LP relaxation that is nearly identical to \eqref{TLPon}, however, constraint \eqref{eqn:tightOnlineConstraint} is removed so that the LP reflects the optimal offline benchmark. Consequently, our algorithm (\Cref{alg:corrpropoff}) is similar to \Cref{alg:corrprop}, albeit with using the offline LP solution and corresponding proposal probability. Contrary to \Cref{thm:main} which  required a second algorithm to handle VWHC instances, \Cref{alg:corrpropoff} is sufficient to prove \Cref{thm:compratio}. Below, we highlight the main differences between our approximation versus competitive analysis. A full treatment and proof is deferred to \Cref{app:competitive} to keep the paper concise. 

The first difference between this analysis and the one in \Cref{sec:algo} is that the independent Markov chains, induced by \Cref{alg:corrpropoff}, stochastically dominate a collection of independent $\textsf{Pois}(u_i \cdot \lambda_i/\mu_i \cdot (u_i + \lambda_i / \mu_i)^{-1})$ distributions, where we use $u_i = 1-\exp(-\lambda_i  / \mu_i)$. This should be contrasted with the independent Markov chains induced by \Cref{alg:corrprop}, where each is distributed exactly according to a $\textsf{Pois}(x_{i,a}/\mu_i)$ distribution (\Cref{claim:stationarydistiMcPoisson}). Consequently, for the competitive analysis, in the convex order, we have
\begin{align}
R_{j}(w) &:=  \sum_{i : r_{i,j} \ge w} p_{i,j} \cdot \text{Pois}\left(\frac{\lambda_i/\mu_i}{u_i+\lambda_i/\mu_i} \cdot u_i\right) \notag \\ 
& \cvxle \textup{Pois} \left( \sum_{i:r_{i,j} \geq w}\frac{\lambda_i/\mu_i}{u_i+\lambda_i/\mu_i} \cdot \frac{x_{i,j}}{\gamma_j} \right) \ . \label{ineq:pois_main_body}
\end{align} Now, observing 
\[
    f(\lambda_i / \mu_i) := \frac{\lambda_i/\mu_i}{u_i + \lambda_i / \mu_i} \geq \frac{1}{2}
\]
for every value of $\lambda_i / \mu_i$, entails the $(1-1/\sqrt{e})$-competitive ratio. 

To go beyond this bound, we use the properties of correlated proposals and analyze three different cases for every online type $j$. The first case is exactly similar to \hyperref[case3]{Case 3} in \Cref{sec:sharpened}---i.e., the majority of proposal probabilities to $j$ is nearly 1. In the second case, we classify offline types into ``abundant'' and ``scarce'' based on their $\lambda_i / \mu_i$. Note that as $\lambda_i/ \mu_i$ increases, $f(\lambda_i / \mu_i)$ increases too. Therefore, if a non-negligible fraction of $j$'s matches come from abundant offline types, we can give a sharpened bound for inequality \eqref{ineq:pois_main_body}. If neither of the above two cases hold for $j$, it means that scarce queues constitute most of $j$'s matches with high proposal probability. In this case, we show that the tightening constraints of the LP imply that $j$ cannot be saturated. Consequently, we can argue, similarly to \Cref{obscase1}, that when the LP mass incident to $j$ is bounded away from $\gamma_j$, the concavity step can be tightened to obtain a better bound. Combining these three observations, we show that our algorithm is $(1-1/\sqrt{e} + \delta)$-competitive on every $j \in J$ for some constant $\delta>0$, and \Cref{thm:compratio} immediately ensues.

\section{Conclusion}

In this work, we aimed to develop a better understanding of the stationary bipartite matching problem that reflects the dynamic evolution of many matching markets. By leveraging the toolkit of approximation algorithms and competitive analysis, while accounting for the stochastic and dynamic nature of arrivals and departures inherent in queueing systems, we  obtained improved performance guarantee and simplified existing analytical results. Notably, we breached the $(1-1/e)$-approximation ratio through a combination of correlated LP rounding, greedy-like algorithm, and fine-grained analysis of the correlation structure among the offline nodes. Our work suggests that combining algorithmic insights with analytical tools from stochastics and probability theory can provide valuable results and insights for tackling online decision making problems.

% \bibliography{a.bib}

\newpage
\appendix 

\section{Preliminaries} \label{app:prelims}
Throughout this paper we use some basic facts about continuous-time Markov chains and birth-death processes, which are recalled here for completeness.

\begin{definition}
A \emph{continuous-time Markov chain (CTMC)} $\{X(t), t \geq 0\}$ on a countable state space $S$ is characterized by transition rates $q_{ij} \geq 0$, for all $i, j \in S$ with $i \neq j$. For every $t \ge 0$ we have $\Pr[X(t+h) = j \mid X(t) = i] = q_{ij} \cdot h + o(h)$ as $h \to 0$.
\end{definition}

\begin{definition}
A CTMC is \emph{irreducible} if for any two states $i, j \in S$, there is a positive probability of reaching $j$ from $i$ in finite time.
\end{definition}

%\begin{definition}
%A probability distribution $\pi$ on $S$ is a \emph{stationary distribution} for a CTMC if $\pi Q = 0$, where $Q$ denotes the \emph{rate matrix} with entries $q_{ij}$ for $i \neq j$ and $q_{ii} = -\sum_{j \neq i} q_{ij}$.
%\end{definition}

\begin{definition}
An irreducible CTMC is called \emph{positive recurrent} if for all states $i \in S$, the expected return time to state $i$, starting from $i$, is finite. Formally, let $\tau_i = \inf\{t > 0 : X(t) = i \mid X(0) = i\}$. The chain is positive recurrent if $\mathbb{E}_i[\tau_i] < \infty$ for all $i \in S$.
\end{definition}

\begin{claim}
If $\{X(t), t \geq 0\}$ is an irreducible, positive recurrent CTMC on a countable state space $S$ it has a unique \emph{stationary distribution} $\pi$, where $\lim_{t \to \infty} P(X(t) = j | X(0) = i) = \pi_j$ for all $i, j \in S$.
\end{claim}

A particularly simple CTMC where the state is a queue that either grows by one or shrinks by one in each step is known as a birth-death process. 

\begin{definition}
A \emph{birth-death process} is a special case of a continuous-time Markov chain on state space $\mathbb{Z}_{\ge 0}$ where transitions occur only between adjacent states. The transition rates are given by $q_{i,i+1} = \lambda_i$, $q_{i,i-1} = \mu_i$, and $q_{ij} = 0$ for $|i-j| > 1$.
\end{definition}

When the birth rate is constant across states, and the death rate scales linearly with the state, an easy calculation shows the stationary distribution is Poisson. 

\begin{claim}\label{claim:stationarydistbirthdeath}
For a birth-death process with constant birth rate $\lambda$ and linear death rate $\mu_i = i \cdot \mu$, the stationary distribution  $\pi$ is Poisson with parameter $\lambda/\mu$.
\end{claim}

\begin{proof}
Note that $\pi_0 \cdot \lambda = \mu \cdot \pi_1$ because in the stationary distribution, the rate of transition out of state 0 equals the rate of transition into state 0. Similarly, for $i \ge 1$ we have $$\pi_{i} \cdot \left( \lambda + i \cdot \mu \right) = \pi_{i-1} \cdot \lambda + \pi_{i+1} \cdot (i+1) \cdot \mu.$$
Solving recursively yields $\pi_i = \pi_0 \frac{\lambda^i}{i! \cdot \mu^i}$. Normalizing so that $\sum_{i=0}^\infty \pi_i = 1$ gives $\pi_0 = e^{-\lambda/\mu}$, and thus $\pi_i = e^{-\lambda/\mu} \cdot \frac{\lambda^i}{i! \cdot \mu^i}$ for $i \ge 0$. 
\end{proof}

The ``PASTA'' property is a classic result in probability theory which implies that if a Poisson arrival process is independent of an irreducible, positive recurrent CTMC, as $t \rightarrow \infty$ each arrival observes the system as if in the stationary distribution. The formal claim we will use in this paper is as follows. 

\begin{lemma}[implied by \cite{wolff1982poisson}] \label{lem:pasta}
    Let $\{X(t), t \ge 0\}$ denote an irreducible, positive recurrent CTMC on a countable state space with stationary distribution $\pi$. Consider an independent stream of arrivals according to a Poisson process with rate $\gamma$; upon an arrival at time $t$, a possibly (random) reward $f(X(t))$ is generated depending only on the state $X(t)$. If $\textup{\textsf{Reward}}[0,t]$ denotes the cumulative reward generated in the time interval $[0,t]$ then
    $$\lim_{t \rightarrow \infty} \mathbb{E} \left[ \frac{\textup{\textsf{Reward}}[0,t]}{t} \right] = \gamma_j \cdot  \mathbb{E}_{Q \sim \pi} [f(Q)].$$
\end{lemma}

\section{Informative Examples}

\begin{exm} \label{exampleneedtightening}
    Say we have $n$ offline types, each with $\lambda_i = n^{-1}$ and $\mu_i = 1$, and one online type with arrival rate $\gamma_j = n^{-2}$. The rewards along each of the $n$ edges are all equal to 1. 
    %As $L \rightarrow \infty$, 
\end{exm}

\begin{claim} \label{app:needTighteningExample}
Consider the instance defined by \Cref{exampleneedtightening}. If Constraint~\eqref{eqn:tightOnlineFlow} of \eqref{TLPon} is replaced with the simple ``flow balance'' constraint of $\sum_{i \in I} x_{i,j} \le \gamma_j$ and/or the ``singleton constraints" $x_{i,j} \le \gamma_j \cdot \left( 1 - \exp(-\lambda_i/\mu_i) \right)$ the stationary gain of optimum online is at most a $(1-1/e+o(1))$-fraction of the optimal value of the LP. 
\end{claim}
\begin{proof}
    Consider the solution which sets $x_{i,j} = \frac{1}{n^3+n^2}$ and $x_{i,a} = \frac{1}{n} - \frac{1}{n^3+n^2}$ for every offline type $i$; this satisfies Constraints~\eqref{eqn:tightOfflineFlow}, \eqref{eqn:tightOnlineConstraint}, and \eqref{eqn:nonnegativity} of \eqref{TLPon}, has $\sum_{i} x_{i,j} \le \gamma_j$, and additionally has $x_{i,j} \le \gamma_j \cdot (1 - \exp(-\lambda_i / \mu_i))$ for sufficiently large $n$. The objective value it achieves is given by $\frac{n}{n^3+n^2} = \frac{1}{n^2} \cdot (1 + o(1))$. 

    However, by the fact that \eqref{TLPon} is a valid relaxation (which we prove as \Cref{claim:lprelaxation}), we have that if $x^*_{i,j}$ denotes the average match rate of $i$ and $j$ by the optimum online, then $$\sum_{i \in I} x_{i,j}^* \le \gamma_j \cdot \left( 1 - \exp \left( \sum_{i \in I} -\lambda_i / \mu_i \right) \right) = (1-1/e) \cdot \frac{1}{n^2}.$$ As $\sum_{i \in I} x_{i,j}^*$ is precisely the stationary reward of the optimum online policy, the claim follows.
\end{proof}

\begin{exm} \label{exampleneeddependent}
    Say we have $n$ offline types, each with $\lambda_i = \frac{1}{\sqrt{n}}$ and $\mu_i = 1$, and one online type with arrival rate $\gamma_j = \sqrt{n} - 1$. The rewards along each of the $n$ edges are all equal to 1. 
    %As $L \rightarrow \infty$, 
\end{exm}

\begin{claim}
     If $\pi^{\textup{\textsf{ind}}}$ denotes the stationary distribution of the independent Markov chains when running \Cref{alg:corrprop} on the instance defined by \Cref{exampleneeddependent}, we have $$ \sum_j \gamma_j \cdot  \mathbb{E}_{Q \sim \pi^{\textup{\textsf{ind}}}} [\textup{\textsf{ALG}}(j, Q)] \le (1 - 1/e + o(1)) \cdot \textup{\text{OPT}\eqref{TLPon}} \ .$$ In other words, analyzing \Cref{alg:corrprop} via a stochastic dominance argument with independent offline queues cannot achieve better than a $(1-1/e)$-approximation to \textup{\eqref{TLPon}}. 
\end{claim}
%We consider two examples, each with the property that the state of each queue is in $\{0,1\}$ with a high probability, and ``proposes'' to the arriving online node with probability 1. 

\begin{proof}
    It is straightforward to see that for this instance, $\text{OPT}\eqref{TLPon} \ge (\sqrt{n} - 1) \cdot (1 - \exp(-\sqrt{n}))$, as it is feasible to set $x_{i,j} = \frac{\sqrt{n}-1}{n} \cdot (1 - \exp(-\sqrt{n}))$ and $x_{i,a} = \frac{1}{\sqrt{n}} - \frac{\sqrt{n}-1}{n} \cdot (1 - \exp(-\sqrt{n}))$ for every offline type $i$. However, in the stationary distribution of the independent Markov chains with this (optimal) LP solution, by \Cref{claim:stationarydistiMcPoisson} the number of available nodes of type $i$ is distributed as $\text{Pois} \left( \frac{1}{\sqrt{n}} - \frac{\sqrt{n}-1}{n} \cdot (1 - \exp(-\sqrt{n})) \right)$. As availability is independent acorss types in the stationary distribution, the probability no type is available is given by $$\exp \left( - \sqrt{n} + (\sqrt{n}-1) \cdot (1 - \exp(-\sqrt{n}))) \right) = 1/e + o(1)$$ where $o(1)$ tends to zero as $n \rightarrow \infty$. By \Cref{lem:pastasection2} (i.e., the PASTA property), the claim follows. 
\end{proof}

\begin{exm}\label{app:alg_bad_example}
    Consider an instance (visualized in \Cref{fig:bad_example}) where $I = \{1,\cdots, n\}$ with $\lambda_i = \mu_i = 1$ for every $i \in I$. Furthermore, let $J = \{1,\cdots, n+1\}$ and $\gamma_j = n$, $r_{j,j} = 0$ for every $j \leq n$; we further have $\gamma_{n+1} = 1$ and $r_{i, n+1} = 1$ for every $i \in I$. An optimal solution of \eqref{TLPon} is $x_{i,n+1} = x_{i,a} =  1/n$ and $x_{i,j} = (n-2)/n$ for every $1\leq i = j \leq$. Under \Cref{alg:corrprop}, the stationary distribution is almost the same as the independent Markov chains. Hence, the online vertex $(n+1)$ has a matching probability of $\approx \textup{Bin}(n,1/n)$, whereas $\opton$ has a reward of $\approx 1$. Letting $n \to \infty$ shows that the approximation ratio of \Cref{alg:corrprop} on these instances can be $(1-1/e + o(1))$. 
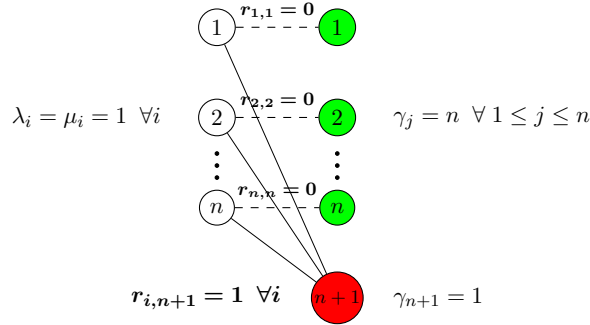
\begin{figure}[H]
    \centering
    \begin{tikzpicture}[scale=0.8, every node/.style={scale=0.8}]
            % Left vertices
            \node[draw, circle, inner sep=3pt] (L1) at (0, 0) {$1$};
            \node[draw, circle, inner sep=3pt] (L2) at (0, -1.5) {$2$};
            \node[draw, circle, inner sep=3pt] (L3) at (0, -3) {$n$};
            
            \node[draw, circle, fill=black, inner sep=.6pt] at (0, -2.3) {};
            \node[draw, circle, fill=black, inner sep=.6pt] at (0, -2.5) {};
            \node[draw, circle, fill=black, inner sep=.6pt] at (0, -2.1) {};
            \node[draw, circle, fill=black, inner sep=.6pt] at (2, -2.3) {};
            \node[draw, circle, fill=black, inner sep=.6pt] at (2, -2.5) {};
            \node[draw, circle, fill=black, inner sep=.6pt] at (2, -2.1) {};

            % Right vertices
            \node[draw, circle, fill=green, inner sep=3pt] (R1) at (2, 0) {1};
            \node[draw, circle, fill=green, inner sep=3pt] (R2) at (2, -1.5) {2};
            \node[draw, circle, fill=green, inner sep=3pt] (R3) at (2, -3) {$n$};
            \node[draw, circle, fill=red, inner sep=0.8pt] (R4) at (2, -4.5) {\footnotesize $n+1$};
            
            % Draw solid edges
            \draw[dashed] (L1) -- (R1) node[midway, above] {\footnotesize $\boldsymbol{r_{1,1} = 0}$};
            \draw[dashed] (L2) -- (R2) node[midway, above] {\footnotesize $\boldsymbol{r_{2,2} = 0}$};
            \draw[dashed] (L3) -- (R3) node[midway, above] {\footnotesize $\boldsymbol{r_{n,n} = 0}$};
            \draw[line width=0.1pt, bend right] (L1) -- (R4);
            \draw[bend left, line width=0.1pt] (L2) -- (R4);
            \draw[bend left, line width=0.1pt] (L3) -- (R4);

            \node[left, xshift = -0.8cm] at (R4) {$\boldsymbol{r_{i,n+1} = 1 \; \; \forall i}$};
            \node[right, xshift = +0.8cm] at (R4) {$\gamma_{n+1} = 1$};

            \node[left, xshift = -0.8cm] at (L2) {$\lambda_i = \mu_i = 1 \; \; \forall i$};
            \node[right, xshift = +0.8cm] at (R2) {$\gamma_j = n \; \; \forall \; 1 \leq j \leq n$};

            % \draw[line width=1.2pt] (L1) -- (R1);
            % \draw[dashed] (L2) -- (R1);
            % \draw[line width=1.2pt] (L2) -- (R2);
            % \draw[line width=1.2pt] (L3) -- (R1);
            % \draw[line width=1.2pt] (L3) -- (R2);
            % \draw[line width=1.2pt] (L3) -- (R3);
            
            % % Draw dashed edges
            % \draw[dashed] (L1) -- (R3);
            % \draw[line width=1.2pt] (L2) -- (R3);
            
        \end{tikzpicture}
     \caption{Family of instances where \Cref{alg:corrprop} is $(1-1/e+o(1))$-approximate}
        \label{fig:bad_example}

\end{figure}
    
\end{exm}

\begin{comment}
OPT can combine all the offline types, obtaining a super-type with $\lambda_1' = L, \mu_1' = 1$. By the flow balance equation, we have \begin{align}
     L &= \text{abandonment rate } + \text{match rate} \notag \\ &= \mu_1' \cdot \ex{Q(L)} + \text{match rate} \notag \\ &= \ex{Q(L)} + (L-1) \cdot \pr{Q(L) > 0} \ . \label{eq:flow-balance} \end{align} Therefore, to prove our claim, it suffices to show that $\frac{\expar{Q(L)}}{L} \to 0$. A straightforward drift analysis proves the following claim, which completes the argument:
\begin{claim}
    $\ex{Q(L)} \leq \sqrt{L} + 1$
\end{claim}
\begin{proof}[Proof sketch.]
    With the Lyapunov function $V(q) = q^2 - q$ and generator matrix $G$, we have
    \[ GV(q) = L(2q + 1) + (L-1+q)(-2q+1) - (L-1) \cdot \mathbb{I}[q=0] \ . \] By Foster-Lyapunov criterion and rearrangement, we obtain \[ \ex{2Q(L)^2 + 2Q(L) \cdot (L-1) + (L-1) \cdot \mathbb{I}[Q(L) > 0]} \leq \ex{2LQ(L) + Q(L) + 2L-1}  \ . \] Then, applying Jensen's inequality completes the proof. 
\end{proof}
\end{comment}

\label{app_algorithm}

\section{Deferred Proofs of Section~\ref{sec:algo}} 

\subsection{Proof of \Cref{claim:lprelaxation}} \label{app:claimlprelaxation}
\claimlprelaxation* 
\begin{proof}
    For the first inequality, for every $t \ge 0$ let $\text{OPT}_{i,j}[0,t]$ denote the number of matches between types $i$ and $j$ made by the optimum online algorithm in the time range $[0, t]$. Similarly let $\text{OPT}_{i,a}[0,t]$ denote the number of times a node of type $i$ abandons without being matched. Define $$x_{i,j}^* := \lim_{t \rightarrow \infty} \frac{\mathbb{E}[\text{OPT}_{i,j}[0,t]]}{t} \quad \text{and} \quad x_{i,a}^* :=  \lim_{t \rightarrow \infty} \frac{\mathbb{E}[\text{OPT}_{i,a}[0,t]]}{t}.$$ Prior work (e.g. \cite{collina2020dynamic}, \cite{aouad2022dynamic}) shows that these limits exist, and furthermore that the variables $\{x_{i,j}^*\}_{i \in I, j \in J}, \{x_{i,a}^*\}_{i \in I}$ satisfy Constraint~\eqref{eqn:tightOfflineFlow} and  Constraint~\eqref{eqn:tightOnlineConstraint}. It is immediate that the expected stationary gain of the optimum offline policy is given by $$\lim_{t \rightarrow \infty} \frac{\sum_{i \in I,j \in J} \text{OPT}_{i,j}[0,t] \cdot r_{i,j}}{t} = \sum_{i \in I} \sum_{j \in J} x_{i,j}^* \cdot r_{i,j}.$$ Constraint~\eqref{eqn:tightOnlineFlow} is new, although a special case was already observed in \cite{kessel2022stationary}. To see why it holds, for each offline type $i$ we let $A_i(t)$ denote the number of \emph{alive} offline nodes of type $i$ at time $t$, where an offline node is said to be alive if has not yet departed due to the expiry of its exponential clock with parameter $\mu_i$.\footnote{Note that this definition entirely ignores the matching algorithm. Although a node must be alive for it to be matched, a node could be matched and still count as ``alive'' until it departs.} It is immediate that for each offline type $i \in I$, the queue $A_i(t)$ is described by an independent birth-death process, with constant birth rate $\lambda_i$ and linear death-rate of $k \cdot \mu_i$ in state $k$. Thus via \Cref{claim:stationarydistbirthdeath}, for any subset of offline types $H \subseteq I$, in the stationary distribution $\pi$ of $(A_i(t))_{i \in H}$ we have 
    \begin{align}
    \Pr_{A \sim \pi} \left[ \sum_{i \in H} A_i = 0 \right] = \prod_{i \in H} \exp \left( - \frac{\lambda_i}{\mu_i} \right) . \label{appendixclaimempty}
    \end{align} Any algorithm (online or offline) can only match a type $i$ when it is alive; thus an arrival of type $j$ at time $t$ is matched to an offline type $i \in H$ only if $\sum_{i \in H} A_i(t) > 0$. By the PASTA property (\Cref{lem:pasta}) we hence have $$\sum_{i \in H} x_{i,j}^* = \frac{\mathbb{E}[\sum_{i \in H} \text{OPT}_{i,j}[0,t]]}{t} \le \gamma_j \cdot \mathop{\mathbb{E}}_{A \sim \pi} \left[ 1 -  \mathds{1} \left[ \sum_{i \in H} A_i = 0 \right] \right] \overset{\eqref{appendixclaimempty}}{=} \gamma_j \cdot \left( 1 - \exp \left( - \sum_{i \in H} \frac{\lambda_i}{\mu_i} \right) \right). $$
Thus Constraint~\eqref{eqn:tightOnlineFlow} holds for the match rates induced by any online or offline policy. This shows that $\{x_{i,j}^*, x_{i,a}^*\}$ is a feasible solution to \eqref{TLPon} whose objective is the stationary reward of the optimum online matching policy; hence $\textup{OPT}\eqref{TLPon} \geq \opton$. 
\end{proof}

\subsection{Proof of \Cref{TLPpolysolvable}} \label{app:proofTLPpolysolvable}
\TLPpolysolvable*

\begin{proof} It is sufficient to find an efficient separation oracle for constraints \eqref{eqn:tightOnlineFlow}. These constraints are equivalent to the condition that for any non-negative weights $0 \leq \theta_i \leq \frac{\lambda_i}{\mu_i}$, we have \begin{align}\label{ineq:equivalent_hall_constraints}
 \sum_{i \in I} {\theta_i} \cdot \frac{x_{i, j}\mu_i}{\gamma_j \lambda_i} \leq 1 - \exp\left(-\sum_{i \in I} \theta_i \right), \end{align} which follows from convexity of the function $\sum_{i \in I} {\theta_i} \cdot \frac{x_{i, j}\mu_i}{\gamma_j \lambda_i} + \exp\left(-\sum_{i \in I} \theta_i \right)$ in $\theta_i$ for each $i \in I$. Now, for a given $\boldsymbol{x}$ and $i, j$, finding $A \subseteq I$ that maximizes \[\sum_{i \in A} \frac{x_{i,j}}{\gamma_j} + \exp\left(-\sum_{i \in A} \frac{\lambda_i}{\mu_i}\right)\] is equivalent to finding $0 \leq \theta_i \leq \frac{\lambda_i}{\mu_i}, \forall i \in I$ that maximizes the equivalent formulation \eqref{ineq:equivalent_hall_constraints}. For any fixed $\sum_{i \in I} \theta_i$, the maximum is attained by a greedy algorithm that assigns $\theta_i$ in decreasing order of $\frac{x_{i,j}\mu_i}{\gamma_j \lambda_i}$. Therefore, we only need to examine the $n$ nested constraints in $\eqref{eqn:tightOnlineFlow}$ for subsets comprised of the first $k$ elements in the order of $\frac{x_{i,j}\mu_i}{\gamma_j \lambda_i}$. 
\end{proof}

\subsection{Proof of \Cref{claim:stochasticdominance}} \label{appendix:proofofclaimstochasticdominance}
\claimstochasticdominance*

\begin{proof} The Markov chains $Q(t)$ and $Q^i(t)$ both have states in $\mathbb{Z}_{\ge 0}^I$. In this proof we consider the standard partial order on $\mathbb{Z}_{\ge 0}^I$ where $x \le y$ for $x, y \in \mathbb{Z}_{\ge 0}^I$ if $x_i \le y_i$ for every $i \in I$. We recall the formal definition of stochastic dominance for Markov chains on a partially-ordered set. 

%\Alicomment{The standard definition of stochastic dominance (in BL 94) is in fact stronger. It should hold for any initialization $A(0) \leq B(0)$. This is not implied by just $A(0) = B(0)$. We can either keep Definition C.1 and change the ``if and only if'' of Lemma C.3 or change this definition (but ensure we always use stochastic dominance in the stronger sense)}
\begin{definition}\label{def:coupling}
    Let $A(t)$ and $ B(t)$ be continuous-time Markov chains taking values in a countable partially-ordered set $\mathcal{Y}$, with deterministic starting states $A(0) \le B(0)$. We say $A$ is stochastically dominated by $B$ if there exists a coupled Markov chain $(A'(t), B'(t))_{t \ge 0}$ such that the marginal distribution of $A'(t)$ (resp. $B'(t)$) is equivalent to that of $A(t)$ (resp. $B(t)$) and $$\Pr \left[ A'(t) \le B'(t) \text{ for all } t \right] = 1.$$
\end{definition}

From a natural monotonicity in \Cref{alg:corrprop}, stochastic dominance is sufficient to argue a relationship between the stationary gain. 
\begin{fact}\label{app:factstochasticdom}
    If $A(t)$ is stochastically dominated by $B(t)$, then $$ \mathbb{E}_{x \sim A(t)}[\textup{\textsf{ALG}}(j,x)] \le \mathbb{E}_{x \sim B(t)}[\textup{\textsf{ALG}}(j,x)].$$
\end{fact}

\begin{proof}
    For any vector $x \in \mathbb{Z}_{\ge 0}^I$, define  $f_w(x) := \Pr[\textsf{ALG}(j,x) \ge w].$ We can observe $$ f_w (x) =  \min \left( 1, \sum_{i : r_{i,j} \ge w} x_{i} \cdot p_{i,j} \right).$$ This is because in \Cref{alg:corrprop}, when queue lengths are given by $x$, we run pivotal sampling on a vector that includes $x_i$ copies of the marginal $p_{i,j}$. These marginals are sorted in decreasing order of reward, so by Property \ref{level-set:prefix} of pivotal sampling the chance at least one of them above threshold $w$ is realized equals $f_w(x)$. We can immediately see $f_w(\cdot)$ is monotone, i.e., for $x, y \in \mathbb{Z}_{\ge 0}^I$ with $x \le y$ we have $f_w(x) \le f_w(y)$. The coupling from \Cref{def:coupling} implies that for any $w > 0$ $$\mathbb{E}_{x \sim A(t)}[ f_w(x) ] = \mathbb{E}_{x_1, x_2 \sim (A'(t), B'(t))} [f_w(x_1) ] \le \mathbb{E}_{x_1, x_2 \sim (A'(t), B'(t))} [f_w(x_2) ] = \mathbb{E}_{x \sim B(t)} [f_w(x) ] $$
  
    Thus $$ \mathbb{E}_{x \sim A(t)}[\textup{\textsf{ALG}}(j,x)] =\mathbb{E}_{x \sim A(t)} \left[ \int_0^{\infty}  f_w(x) \right] \, dw \le \mathbb{E}_{x \sim B(t)} \left[ \int_0^{\infty}  f_w(x) \right] \, dw = \mathbb{E}_{x \sim B(t)}[\textup{\textsf{ALG}}(j,x)]$$ as claimed.
\end{proof}

To show stochastic dominance between the queues $(Q_i(t))_i$ induced by \Cref{alg:corrprop} and the independent Markov chains $(Q^\textsf{ind}_i(t))_i$, we utilize the following well-known criterion (also used by \cite{aouad2022dynamic}). 

\begin{lemma}[Stochastic dominance, c.f. \cite{kamae1977stochastic, brandt1994pathwise, lopez2000stochastic}] \label{stochasticdominanceviamonotone}
Let $A(t)$ and $ B(t)$ be continuous-time Markov chains taking values in a countable partially-ordered set $\mathcal{Y}$, with time-homogeneous stochastic kernels $P^{A}$ and $P^{B}$.\footnote{In our case, the kernel $P^A$ can be identified with a matrix, where for $x, y \in \mathbb{Z}_{\ge 0}^I$ entry $P^A_{x,y}$ equals the rate at which the Markov chain transitions from state $x$ to $y$ (this analogously holds for $P^B$).} Then, $A$ is stochastically dominated by $B$ if for every $x, y \in \mathcal{Y}$, and every upwards-closed\footnote{Note: we say $Z$ is upwards-closed if $x \in Z$ and $x \le y$ implies $y \in Z$.} set $Z \subseteq \mathcal{Y}$ we have

$$
x \leq y \text { with } x \in Z \text { or } y \notin Z \Longrightarrow \sum_{z \in Z} P_{x, z}^{A} \leq \sum_{z \in Z} P_{y, z}^{B}.
$$
\end{lemma}

For each $i \in I$ and $j \in J$, let $\textsf{Prob-Matched}_{i,j} ( Q(t))$ denote the probability if $j$ arrives, and each offline type $i'$ has $Q(t)_{i'}$ copies waiting, when running Lines~\ref{lin:inside-loop-start} to \ref{lin:match} of \Cref{alg:corrprop} we match $j$ to $i$. At ever time $t$, the queue for type $i$ increases by 1 at rate $\lambda_i$, and decreases by 1 at rate $$Q_i(t) \cdot \mu_i + \sum_j \gamma_j \cdot \textsf{Prob-Matched}_{i,j} ( Q(t)).$$ Because pivotal sampling matches marginals (Property~\ref{level-set:marginals}), the union bound implies \begin{align}
    \textsf{Prob-Matched}_{i,j}(Q(t)) \le Q(t)_i \cdot p_{i,j}. \label{eqn:boundprob-matched}
\end{align} We will show that for a upwards-closed $A \subseteq \mathbb{Z}_{\ge 0}^n$ if $x \in A$ or $y \notin A$ with $x \le y$, we have $\sum_{z \in A} Q^i_{x,z} \le \sum_{z \in A} Q_{y, z}.$\footnote{For $x, y \in \mathbb{Z}_{\ge 0}^n$, we let $Q_{x,y}$ (resp. $Q^i_{x,y}$) denote the transition rate for $Q$ (resp. $Q^i$) from state $x$ to state $y$.}

If $y \notin A$, then $x \notin A$. For $i \in [n]$, let $e_i$ denote the vector with 1 in its $i$\textsuperscript{th} entry and zeros elsewhere. No entry $x' \le x$ can be in $A$. Hence $$\sum_{z \in A} Q^i_{x,z} = \sum_{i : x + e_i \in A} Q^i_{x,x+e_i} = \sum_{i : x + e_i \in A} \lambda_i \le \sum_{i : y + e_i \in A} \lambda_i = \sum_{z \in A} Q_{y,z}$$ where the inequality follows from $x \le y$ and that $A$ is upwards-closed. 

If $x \in A$, then $y \in A$. It suffices to show $\sum_{z \notin A} Q^i_{x,z} \ge \sum_{z \notin A} Q_{y,z}$. Using that $A$ is upwards-closed, we can compute
\begin{align*}
    \sum_{z \notin A}Q^i_{x,z} &= \sum_{i : x - e_i \notin A} x_i \cdot \left(\mu_i + \sum_j \gamma_j \cdot p_{i,j} \right) \\
    &\ge \sum_{i : y - e_i \notin A}  x_i \cdot \left(\mu_i + \sum_j \gamma_j \cdot p_{i,j} \right) \\
    &= \sum_{i : y - e_i \notin A}  y_i \cdot \left(\mu_i + \sum_j \gamma_j \cdot p_{i,j} \right) \\
    &\ge \sum_{i : y - e_i \notin A}  \left( y_i \cdot \mu_i + \sum_j \gamma_j \cdot  \textsf{Prob-Matched}_{i,j} (y) \right) && \text{via \eqref{eqn:boundprob-matched}} \\
    &=   \sum_{z \notin A} Q_{y,z}
\end{align*}
which completes the proof.
\end{proof}

\subsection{Proof of \Cref{lem:weightedPoisCvx}} \label{app:proofofweightedPoisCvx}
\weightedPoisCvx*

\begin{proof}
    If $a = 0$ the statement is immediate. We now prove the lemma for $a, b \in \mathbb{Q}_{> 0}$; the full claim follows from the fact that the rational numbers are dense in the reals. 

    Write $a = a'/n$ and $b = b'/n$ for positive integers $a'$, $b'$, and $n$. Note that it suffices to show $$\frac{\text{Pois}(b'/n)}{b'} \cvxle \frac{\text{Pois}(a'/n)}{a'}.$$ As the sum of Poisson random variables is Poisson, the LHS is distributed precisely as the empirical mean of $b'$ copies of $\text{Pois}(1/n)$, while the RHS is the empirical mean of $a'$ copies. The statement then follows from \cite[3.A.29]{shaked2007stochastic}.
\end{proof}

\section{Deferred Proofs of Section~\ref{sec:approximation}}\label{app:approximation}

\subsection{Proofs for the Three Easy Cases} \label{app:threecases}

\obscaseone*
\begin{proof} Recall from \Cref{sec:oneminusoneovere} our lower bound
    \begin{align*}
        \textsf{ALG}_j &\ge \gamma_j \cdot \int_0^{\infty} \E [ \min(1, R_j(w))] \, dw 
        \ge \gamma_j \cdot \int_0^{\infty} \left( 1 - \exp \left(-\sum_{i : r_{i,j} \ge w} x_{i,j}/\gamma_j \right) \right) \, dw.
    \end{align*}
    For convenience we define the function $g(\cdot)$ such that $1 - \exp(-z) = (1 - 1/e + g(z)) \cdot z$ and note that $g(\cdot)$ is positive and decreasing for $z \in (0,1)$. We then can bound 
 \begin{align*}
        \textsf{ALG}_j \ge \gamma_j \cdot \int_0^{\infty} \left( 1 - 1/e + g(1 - \epsilon ) \right) \cdot \left( \sum_{i : r_{i,j} \ge w}  x_{i,j} / \gamma_j \right)  \, dw = (1 - 1/e + g(1 - \epsilon )) \cdot \textsf{LP-Gain}_j
    \end{align*} as claimed.
\end{proof}

\obscasetwo*

\begin{proof}
    Let $w^*$ be the smallest weight such that $j$ has at least an $\epsilon / 2$ fraction of its LP-weighted gain sent to rewards in $[0, w^*]$; formally, define $$w^* := \inf \left\{ w \ge 0:  \frac{\textsf{LP-Gain}_j^{[0, w]}}{\textsf{LP-Gain}_j} \ge \epsilon / 2 \right\}$$ and 
    note that 
    \[
        \textsf{LP-Gain}_j^{[0, w^*)} := \sum_{i : r_{i,j} \in [0,w^*)} x_{i,j} \cdot r_{i,j} < {\eps}/{2} 
    \] by exclusion of $w^*$. Moreover, we know by assumption that $j$ sends at most a $1-\epsilon$ fraction of its LP-weighted gain in the interval $[w^*, (1+\epsilon) w^*].$ Thus the fraction of LP-weighted gain $j$ has from rewards in the interval $((1+\epsilon) w^*, \infty)$ is at least $\epsilon/2.$ Also, because $\textsf{LP-Gain}_j^{[0,w^*]} \ge \epsilon/2 \cdot \textsf{LP-Gain}_j$, we observe that
    \begin{align}
    \sum_{i : r_{i,j} \in [0, w^*]} x_{i,j}   \ge \epsilon/2 \cdot \sum_{i} x_{i,j} \ , \label{lb:x_above_wstar}
    \end{align}
    which in particular implies that for every $w > w^*$, we have
    \begin{align}
        \sum_{i:r_{i,j} \geq w} x_{i,j}/\gamma_j \leq 1-{\eps}/{2} \ . \label{ub:x_above_w}
    \end{align}
    Intuitively, \eqref{ub:x_above_w} allows us to get a boost in reward from values higher than $(1+\eps) w^*$. Defining the function $g(\cdot)$ as in the proof of \Cref{obscase1}, we have
    \begin{align*}
        \textsf{ALG}_j 
        &\ge \gamma_j \cdot \int_0^{\infty} \left( 1 - 1/e + g \left( \sum_{i: r_{i,j} \ge w} x_{i,j} / \gamma_j \right) \right) \cdot \left( \sum_{i : r_{i,j} \ge w}  x_{i,j} / \gamma_j \right)  \, dw 
        \\ &\ge \gamma_j \cdot \int_0^{\infty} \sum_i (1 - 1/e + g(1 - \epsilon/2) \cdot \mathbbm{1}[w > w^*]) \cdot x_{i,j}/\gamma_j \cdot \mathbbm{1}[w \le r_{i,j}] \, dw \\
        &= (1-1/e) \cdot \textsf{LP}_j +  g(1 - \epsilon/2) \cdot \sum_i x_{i,j} \cdot  \max(0, r_{i,j} - w^*) \, dw \\
        &\ge (1-1/e) \cdot \textsf{LP}_j +  g(1 - \epsilon/2) \cdot \sum_{i : r_{i,j} \ge (1+\epsilon)w^*} x_{i,j} \cdot r_{i,j} \cdot \epsilon/2   && \epsilon \le 1 \\
        &= (1-1/e) \cdot \textsf{LP}_j +  g(1 - \epsilon/2) \cdot \textsf{LP}_j^{[(1+\epsilon)w^*,\infty)} \cdot \epsilon/2  .
    \end{align*}
    Recalling that $\textsf{LP}_j^{[(1+\epsilon)w^*,\infty)} \geq \epsilon / 2 \cdot \textsf{LP}_j$ concludes the argument.     
\end{proof}

\obscasethree*
\begin{proof}
We recall the bound \begin{align*}
 \textsf{ALG}_j &\ge \gamma_j \cdot \int_0^{\infty} \E [ \min(1, R_j(w)] \, dw.
 \end{align*} For any $w > 0$, let $I_w$ denote the set of $i \in I$ with reward $r_{i,j} \ge w$. Let $\sigma_w := \sum_{i \in I_w} x_{i,j} / \gamma_j$ and let $\sigma_{w, S} := \sum_{i \in I_w \cap S} x_{i,j} / \gamma_j$. In the convex order, we have by \Cref{lem:weightedPoisCvx} that \begin{align*}
       R_{j}(w) = \sum_{i \in I_w} p_{i,j} \cdot \textup{Pois}(x_{i,a} / \mu_i ) & \cvxle  \sum_{i \in I_w \cap S} (1-\epsilon) \cdot \textup{Pois}\left(\frac{x_{i,a} \cdot p_{i,j}}{1-\epsilon}\right) +  \sum_{i \in I_w \setminus S} 1 \cdot \textup{Pois}\left(x_{i,a} \cdot p_{i,j} \right).
    \end{align*}
As $p_{i,j} := \frac{x_{i,j} / \gamma_j}{x_{i,a} / \mu_i}$ the RHS is distributed exactly as $(1-\epsilon) \cdot \text{Pois} \left( \frac{\sigma_{w,S}}{1-\epsilon} \right) + \text{Pois} \left( \sigma_w - \sigma_{w,S} \right).$ For convenience denote this distribution by $D$; we then have the stationary probability $j$ is matched is 
\begin{align*}
\mathbb{E}[\min(1, R_j(w))] 
    &\ge \mathbb{E}[\min(1, D)] \\
    % &= (1 - \exp(-\sigma_w + \sigma_{w,S})) \cdot 1 + \exp(-\sigma_w + \sigma_{w,S}) \cdot \mathbb{E} \left[  \min \left( 1, (1-\epsilon) \cdot \text{Pois} \left( \frac{\sigma_{w,S}}{1-\epsilon} \right)  \right) \right] \\
    &= 1 - \exp(-\sigma_w + \sigma_{w,S}) \cdot \mathbb{E} \left[ \max \left(0, 1 - (1-\epsilon) \cdot \text{Pois} \left( \frac{\sigma_{w,S}}{1-\epsilon} \right)  \right) 
    \right]  
\end{align*}
Note that when $\sigma_w = 1$ and $\sigma_{w,S} = 0$, the RHS is precisely $1-1/e$; when either $\sigma_w < 1$ or $\sigma_{w,S} > 0$, this bound is loose. By the case hypothesis, we know $\sigma_w - \sigma_{w,S} \le 1 - \epsilon$; from this it is straightforward to see $ \mathbb{E}[\min(1,R_j(w))] \ge (1-1/e + b_3(\epsilon)) \cdot \sigma_w$. (For a complete proof, refer to the auxiliary \Cref{app:expectationbounding}). We conclude via direct computation, as
\begin{align*}
    \textsf{ALG}_j &\ge \gamma_j \cdot \int_0^{\infty} \E [ \min(1, R_j(w)] \, dw \\
    &\ge \gamma_j \cdot (1 - 1/e + b_3(\epsilon)) \cdot \int_0^{\infty} \sigma_w \, dw =  (1 - 1/e + b_3(\epsilon)) \cdot \textsf{LP-Gain}_j. \qedhere
\end{align*}
\end{proof}

 \begin{claim} \label{app:expectationbounding}
      For $\sigma \in [1-\eps/2,1]$, $\sigma_S \in [0, \sigma]$, and $\epsilon < 0.1$ such that $\sigma - \sigma_S \le 1 - \epsilon$, we have $$ 1 - \exp(-\sigma + \sigma_S) \cdot \mathbb{E} \left[ \max \left(0, 1 - (1-\epsilon) \cdot \textup{Pois} \left( \frac{\sigma_S}{1-\epsilon} \right)  \right) 
     \right]   \ge (1-1/e + b(\epsilon)) \cdot \sigma $$ for some $b(\epsilon ) > 0$. 
 \end{claim}

 \begin{proof}
     For convenience let $(\star)$ denote the LHS of the above inequality. As $\epsilon < 0.5$, the expectation in the LHS is positive if and only if  $ \textup{Pois} \left( \frac{\sigma_S}{1-\epsilon} \right)$ realizes in $\{0, 1\}$. Using this we can compute
     \begin{align*}
      (\star) &= 1 - \exp(-\sigma + \sigma_S) \cdot \left( \exp \left( - \frac{\sigma_S}{1-\epsilon} \right) + \exp \left( - \frac{\sigma_S}{1-\epsilon} \right) \cdot \frac{\sigma_S}{1-\epsilon} \cdot \epsilon \right) && \epsilon < 0.5 \\
     &= 1 - \exp(-\sigma + \sigma_S) \cdot  \exp \left( - \frac{\sigma_S}{1-\epsilon} \right) \cdot \left( 1 +  \frac{\sigma_S}{1-\epsilon} \cdot \epsilon \right).
     \end{align*}
 As $1 + x + 0.5x^2 \le \exp(x)$ for non-negative $x$, we can further bound
     \begin{align*}
     (\star) &\ge 1 - \exp(-\sigma + \sigma_S) \cdot  \exp \left( - \frac{\sigma_S}{1-\epsilon} \right) \cdot \left( \exp \left(  \frac{\sigma_S}{1-\epsilon} \cdot \epsilon \right) - \frac{\sigma_S^2 \epsilon^2}{2(1-\epsilon)^2}\right) \\
     &= 1 - \exp \left( - \sigma \right) \cdot \left( 1 +  \exp \left(  \sigma_S - \frac{\sigma_S}{1 - \epsilon} \right) \cdot\frac{\sigma_S^2 \epsilon^2}{2 (1 - \epsilon)^2} \right).
     \end{align*}
     Defining $h(z) := \exp \left(  z - \frac{z}{1 - \epsilon} \right) \cdot\frac{z^2 \epsilon^2}{2 (1 - \epsilon)^2}$ our bound can be written as $(\star) \ge 1 - \exp( - \sigma ) \cdot ( 1 + h(\sigma_S)).$
    
     Since $\sigma > 1 - \epsilon / 2$, we have by the claim's assumption that $\sigma_S \ge \epsilon / 2$. For $\epsilon \le 0.5$ we have that $h(\cdot)$ is increasing on $(0,1)$, so $$(\star) \ge 1 - \exp(-\sigma) \cdot (1 + h(\epsilon / 2)) \ge \sigma \cdot (1 - 1/e + b(\epsilon)).$$ For the final inequality, one can observe that $\frac{1-\exp(-\sigma) \cdot (1+h(\epsilon / 2))}{\sigma} $ for $\sigma \in [1-\epsilon / 2, 1]$ is minimized at $\sigma = 1$ for $h(\epsilon / 2) \le 0.1$, which concludes the proof with $b(\eps) = 1-1/e \cdot (1+h(\eps/2))$. 
 \end{proof}

\subsection{Proof of \Cref{lem:casesBeat1-1/e}} \label{app:casesbeat}

\casesbeatoneminusoneovere*
\begin{proof}
    For any type $j$ that is not hard, we know that $\textsf{ALG}_j \cdot \textsf{LP-Gain}_j^{-1}$ is at least $(1-1/e + b(\epsilon))$ if $j$ falls into \hyperref[case2]{Case 2}. Otherwise, it must be the case that either \hyperref[case1]{Case 1} or \hyperref[case3]{Case 3} holds for $j$ with neighborhood $S_j$, with parameter $\epsilon'$. In this case 
    \begin{align*}
    \textsf{ALG}_j \cdot \textsf{LP-Gain}_j^{-1} \ge (1-\eps) \cdot (1 - 1/e + b(b^{-1}(2\eps))) \ge 1-1/e+\eps
    \end{align*}
    where the inequality holds for sufficiently small $\eps$ (e.g. $\eps \le 0.1$). The long-term average reward of \Cref{alg:corrprop} is thus given by 
    \begin{align*}
         \sum_{j \in \mathcal{H}} \textsf{ALG}_j + \sum_{j \in J \setminus \mathcal{H}} \textsf{ALG}_j &\ge (1-1/e) \cdot \textsf{LP-Gain}_{\mathcal{H}} + (1-1/e + \min \{ \epsilon, b (\epsilon) \} ) \cdot \textsf{LP-Gain}_{J \setminus \mathcal{H} } \\
         &\ge (1-1/e +  \min \{ \epsilon, b (\epsilon ) \} \cdot \epsilon) \cdot \textsf{LP-Gain}_J. \qedhere
    \end{align*}
\end{proof}

% \begin{claim}\label{xijxiabounds}
% For LP solutions satisfying \ref{structuralassumption} we have
%     \begin{align*}
% (1-\eps) \cdot \frac{\gamma_j \lambda_i}{\mu_i + \sum_{k \in N_i} \gamma_k} \leq x_{i,j} \leq \frac{1}{1-\eps} \cdot \frac{\gamma_j \lambda_i}{\mu_i + \sum_{k \in N_i} \gamma_k} \quad \text{ for } i \in I, j \in N_i  
% \end{align*}
% and \begin{align*}
%     \frac{\mu_i\lambda_i}{\mu_i + \sum_{k \in N_i} \gamma_k} \leq x_{i,a} \leq \frac{1}{1-\eps} \cdot \frac{\mu_i\lambda_i}{\mu_i + \sum_{k \in N_i} \gamma_k} \quad \text{for } i \in I. 
% \end{align*}
% \end{claim} 
\subsection{Proofs of Properties of the Modified Instance}\label{prf:instance_property}

\subsubsection{Property (v): Binary queues}\label{app:binary_queues}
If we split a type into sufficiently many types, we have that for every offline type, the departure rate is much higher than the arrival rate; for this reason, we can show that truncating the offline queues to be binary results in only a small loss in the approximation ratio. 

Concretely, we split offline type $i$ into $K$---with $K \gg 1$ to be specified---identical types, each with arrival rate $\lambda_i / K$. We split the LP solution similarly: each new type has $x_{i,j} / K$ and $x_{i,a} / K$ for every $j \in J$. As before, this new LP solution is feasible with the same objective value, in the (further) modified instance. Importantly, we preserve the membership of vertices in TOP and BOT, i.e. if $i \in \text{TOP}$, we split it into $K$ vertices in TOP. For this reason, this second-stage modification does not alter \Cref{alg:second}'s decisions. This is particularly noteworthy since the second-stage splits may lead to an instance with a super-polynomial size, compared to the initial values. 

To determine $K$, we define $u = \max_{i \in I} {\lambda_i}/{\mu_i}$ (for values before performing the ``binary queues'' step). If $u > 1$, we let $K = \lceil \frac{u^2 n^2}{\eps^2} \rceil$; otherwise, $K = \frac{n^2}{\eps^2}$. Therefore, for queue $i$ in the newly modified instance, we have $\lambda_i / \mu_i \leq u/K$. Now, note that regardless of the matching policy and its decisions, the queue cannot be bigger than the case where we match no one. In this case, queue $i$'s length is a $\text{Pois}( \lambda_i / \mu_i)$ random variable, independent from other queues. Thus, we have
\begin{align}
    \pr{Q_i \geq 2} \leq 1-\exp\left(-\frac{\lambda_i}{\mu_i}\right) - \frac{\lambda_i}{\mu_i} \cdot \exp\left(-\frac{\lambda_i}{\mu_i}\right) \leq \left(\frac{\lambda_i}{\mu_i}\right)^2 \leq \frac{u^2}{K^2} \ . 
\end{align}
Therefore, by union bound, the probability that some queue has length more than 1 is at most $u^2/K^2 \cdot nK = nu^2/K \leq \eps^2 /n$. Due to this small loss of accuracy and to keep this notation simple, we assume from now on that the queues are binary, i.e., if there is some offline vertex in the queue that is not matched yet, the queue cannot grow bigger. All guarantees hold up to a $(1-\eps^2/n)$-factor loss.

% \subsection{The new instance $\mathcal{I}'$} \label{app:newinstance}

% \begin{definition}
% Let $\mathcal{I}$ denote our original instance $(I, J, \lambda, \gamma, r)$ with optimal LP solution $(x_{i,j}, x_{i,a})$. We will construct a new instance $\mathcal{I}' = (I', J', \lambda', \gamma', r')$ with LP solution $(x'_{i,j}, x'_{i,a})$ by starting with $\mathcal{I}' = \mathcal{I}$ and $x' = x$ and applying the following operations. 
% \begin{enumerate}
%     \item Delete any online types $j \in J' \setminus H$. All $x'_{i,j}$ with $j \in J' \setminus H$ are set to zero. 
%     \item For every remaining $j \in J'$, we will then remove any edge to some $i$ with $i \notin T_j$ (i.e., set $r_{i,j}' = 0$ and $x_{i,j}' = 0$). For every $i \in T_j$, set reward $r_{i,j}' = r_j$.
%     \item  \label{splitting} Let $M = \max_{i \in I} \frac{\lambda_i}{\mu_i}$, and recall $n = |I|$ denotes the number of offline types in the original instance $\mathcal{I}$. We split every offline type $i \in I$ into $K := 2 \lceil M^2 n^2 \rceil$ new types $\{i_1, i_2, \ldots, i_K\}$ which we add to $\mathcal{I}'$. Each $i_k$ has arrival rate $\lambda'_{i_k} := \lambda_i / K$. We naturally also split the original LP solution, so $x'_{i_k, j} := x_{i,j} / K$ and $x'_{i_k, a} = x_{i,a} / K$; it is easy to see this still satisfies all constraints in \eqref{TLPon}. \todo{Fix when $M$ is small}
%     \item For tiebreaking purposes, we will give each node in $I'$ a label. For each original $i \in I$ and for $k \in [K]$ we will label $i_k$ with the pair $l(i_k) := (k, i)$.
% \end{enumerate}
% \end{definition}

\subsubsection{Approximation of the LP solution}\label{app:approximate_lp}
    For $j \in N_i$, from $1 - \epsilon \le p_{i,j} \le 1$ we have 
    \begin{align} \label{xijupperandlowerbound}
    (1 - \epsilon) \cdot \gamma_j \cdot \frac{x_{i,a}}{\mu_i} \le x_{i,j} \le  \gamma_j \cdot \frac{x_{i,a}}{\mu_i}.
    \end{align} 
    Thus $$ (1-\epsilon) \cdot \frac{x_{i,a}}{\mu_i} \cdot \left( \sum_{k \in N_i} \gamma_k \right) \le \sum_{k \in N_i} x_{i,k} \le \frac{x_{i,a}}{\mu_i} \cdot \left( \sum_{k \in N_i} \gamma_k \right).$$ From Constraint~\eqref{eqn:tightOfflineFlow} we recall $\sum_{k \in N_i} x_{i,k} = \lambda_i - x_{i,a}$; substituting and simplifying gives 
    $$ \frac{\mu_i\lambda_i}{\mu_i + \sum_{k \in N_i} \gamma_k} \le x_{i,a} \le \frac{\mu_i \lambda_i} {\mu_i + (1-\epsilon) \sum_{k \in N_i} \gamma_k}$$ which implies the second half of the claim. Substituting these bounds into \eqref{xijupperandlowerbound} implies the first half.

\subsection{Proof of \Cref{claim:stochasticdominancesecond}} \label{proofofstochasticdominancesecond}
\claimstochasticdominancesecond*

\begin{proof}
% We will show $Q \succeq_{st} Q^{\textup{weak}}$; the proof that $Q^{\textup{weak}} \succeq_{st} Q^{\textup{\textsf{ind}}}$ follows similar steps and is not required for our result. 

We again use \Cref{stochasticdominanceviamonotone} to prove stochastic dominance. For each $i \in I$ and $j \in J$, let $\textsf{Prob-Matched}_{i,j} ( Q(t))$ denote the probability if $j$ arrives, and each offline type $i'$ has $Q_{i'}(t)$ copies waiting, that \Cref{alg:second} will match $j$ to $i$. At every time $t$, the queue for type $i$ increases by 1 at rate $\lambda_i$, and decreases by 1 at rate $$Q_i(t) \cdot \mu_i + \sum_j \gamma_j \cdot \textsf{Prob-Matched}_{i,j} ( Q(t)).$$  We will show that for an upwards-closed $A \subseteq \mathbb{Z}_{\ge 0}^n$ if $x \in A$ or $y \notin A$ with $x \le y$, we have $\sum_{z \in A} Q^{\textup{weak}}_{x,z} \le \sum_{z \in A} Q_{y, z}.$\footnote{For $x, y \in \mathbb{Z}_{\ge 0}^n$, we let $Q_{x,y}$ (resp. $Q^{\textup{weak}}_{x,y}$) denote the entry in the intensity matrix corresponding to $Q$'s (resp. $Q^{\textup{weak}}$'s) transition from state $x$ to state $y$.}

If $y \notin A$, then $x \notin A$. For $i \in [n]$, let $e_i$ denote the vector with 1 in its $i$\textsuperscript{th} entry and zeros elsewhere. No entry $x' \le x$ can be in $A$. Hence $$\sum_{z \in A} Q^{\textup{weak}}_{x,z} = \sum_{i : x + e_i \in A} Q^{\textup{weak}}_{x,x+e_i} = \sum_{i : x + e_i \in A} \lambda_i \le \sum_{i : y + e_i \in A} \lambda_i = \sum_{z \in A} Q_{y,z}$$ where the inequality follows from $x \le y$ and that $A$ is upwards-closed. 

If $x \in A$, then $y \in A$. It suffices to show $\sum_{z \notin A} Q^{\textup{weak}}_{x,z} \ge \sum_{z \notin A} Q_{y,z}$. Using that $A$ is upwards-closed, we can compute
\begin{align}
    \sum_{z \notin A}Q^{\textup{weak}}_{x,z} &= \sum_{\substack{i : i \in \text{TOP}, \\ x - e_i \notin A}} x_i \cdot \left(\mu_i + \sum_j \gamma_j  \right) + \sum_{\substack{i : i \in \text{BOT}, \\ x-e_i \notin A}} x_i \cdot \left( \mu_i + \sum_{j \in N_i} \left(1 - \mathbbm{1} \left[\sum_{i' \in N_j^\uparrow} x_{i'} > 0\right] \right) \gamma_j \right) \nonumber \\
    &\ge \sum_{\substack{i : i \in \text{TOP}, \\ y - e_i \notin A}} x_i \cdot \left(\mu_i + \sum_j \gamma_j  \right) + \sum_{\substack{i : i \in \text{BOT}, \\ y-e_i \notin A}} x_i \cdot \left( \mu_i +  \sum_{j \in N_i} \left(1 - \mathbbm{1} \left[\sum_{i' \in N_j^\uparrow} x_{i'} > 0\right] \right) \gamma_j  \right) \nonumber \\
    &= \sum_{\substack{i : i \in \text{TOP}, \\ y - e_i \notin A}} y_i \cdot \left(\mu_i + \sum_j \gamma_j  \right) + \sum_{\substack{i : i \in \text{BOT}, \\ y-e_i \notin A}} y_i \cdot \left( \mu_i + \sum_{j \in N_i} \left(1 - \mathbbm{1} \left[\sum_{i' \in N_j^\uparrow} x_{i'} > 0\right] \right) \gamma_j \right) \label{equalitysecondstochasticdominance} \\
    &\ge \sum_{\substack{i : i \in \text{TOP}, \\ y - e_i \notin A}} y_i \cdot \left(\mu_i + \sum_j \gamma_j  \right) + \sum_{\substack{i : i \in \text{BOT}, \\ y-e_i \notin A}} y_i \cdot \left( \mu_i + \sum_{j \in N_i} \left(1 - \mathbbm{1} \left[\sum_{i' \in N_j^\uparrow} y_{i'} > 0\right] \right) \gamma_j \right) \nonumber  \\
    &\ge \sum_{i : y - e_i \notin A}  \left( y_i \cdot \mu_i + \sum_j \gamma_j \cdot  \textsf{Prob-Matched}_{i,j} (y) \right) && \hspace{-5em} \text{via \eqref{eqn:boundprob-matched}} \nonumber \\
    &=   \sum_{z \notin A} Q_{y,z} \ . 
\end{align}
Note that for the equality in \Cref{equalitysecondstochasticdominance}, we used that if $y-e_i \notin A$, then $x_i = y_i$ becuase $x \le y$.
\end{proof}

\jointproballqueuesempty*

\begin{proof}
    
Let $\textsf{NA}_i$ denote the event that there was no departure in $Q_i^w(t)$ in the interval $[t^{\infty} - t_i, t^{\infty}]$ due to the depletion of sMc.  The queue for $i$ ends up empty at $t^{\infty}$ if $A_i = 0$ ($i$ departed due to wait time expiring), or $A_i = 1$ and $\textsf{NA}_i=1$ ($i$ departed due to depletion in sMc). Thus 
\begin{align}
        \Pr \left[ \sum_{i \in N_j^{\downarrow}} Q_i^w = 0 \mid \psi_j \right] = \mathbb{E} \Bigg[ \prod_{i \in N_j^{\downarrow}} \left(1 - A_i \cdot \textsf{NA}_i  \right)  \Bigg | \psi_j \Bigg].
\end{align}

\begin{comment}
\begin{align*}
    \mathbb{E} \Bigg[ \prod_{i \in N_j^{\downarrow}} \left(1 - A_i \cdot \textsf{NA}_i  \right)  \Bigg | \psi_j \Bigg]  & = \ex{ \ex{ \prod_{i \in N_j^{\downarrow}} \left(1 - A_i \cdot \textsf{NA}_i  \right)  \Bigg | (A_i)_i, (\delta_{i,k})_k, \psi_j} \Bigg | \psi_j} \\ 
    & = \ex{  \prod_{i \in N_j^{\downarrow}} \left(1 - A_i \cdot \exp\left(-\sum_{k \in N_i} \gamma_k \delta_{i,k}\right)  \right)   \Bigg | \psi_j} \\
    &= \ex{\prod_{i \in N_j^\downarrow} \left(1- \frac{\lambda_i}{\lambda_i + \mu_i} \cdot \exp\left(-\sum_{k \in N_i} \gamma_k \delta_{i,k}\right) \right) \Bigg | \psi_j}
\end{align*}
\end{comment} 

We can expand and simplify as follows:
\begin{align}
        \Pr & \left[ \sum_{i \in N_j^{\downarrow}} Q_i^w = 0 \mid \psi_j \right] \nonumber \\
        &= \sum_{I' \subseteq N_j^{\downarrow}} (-1)^{|I'|} \cdot \mathbb{E} \left[ \prod_{i \in I'} A_i \prod_{i \in I'} \textsf{NA}_{i} \Big | \psi_j \right]  \nonumber \\
        &= \sum_{I' \subseteq N_j^{\downarrow}} (-1)^{|I'|} \cdot \prod_{i \in I'} \frac{\lambda_i}{\mu_i + \lambda_i} \cdot \mathbb{E} \left[ \prod_{i \in I'} \textsf{NA}_{i} \Big | \bigwedge_{i \in I'} A_i, \psi_j \right] \label{eqn:Aisindep} \\
        &= \sum_{I' \subseteq N_j^{\downarrow}} (-1)^{|I'|} \cdot \prod_{i \in I'} \frac{\lambda_i}{\mu_i + \lambda_i} \cdot \mathbb{E} \left[ \prod_{i \in I'} \exp \left( -\sum_{k \in N_i} \gamma_k \cdot \delta_{i,k} \right) \Big | \bigwedge_{i \in I'} A_i, \psi_j \right] \label{eqn:depletionrate} \\
        &=  \mathbb{E} \left[ \sum_{I' \subseteq N_j^{\downarrow}} (-1)^{|I'|} \cdot \prod_{i \in I'} \frac{\lambda_i}{\mu_i + \lambda_i} \cdot  \prod_{i \in I'} \exp \left( -\sum_{k \in N_i} \gamma_k \cdot \delta_{i,k} \right) \Big | \bigwedge_{i \in N_j^{\downarrow}} A_i, \psi_j \right] \nonumber \\
        &= \mathbb{E} \left[ \prod_{i \in N_j^{\downarrow}} \left( 1 - \frac{\lambda_i}{\mu_i + \lambda_i} \cdot \exp \left( - \sum_{k \in N_i} \gamma_k \cdot \delta_{i,k} \right)\right) \Big | \bigwedge_{i \in N_j^{\downarrow}} A_i, \psi_j \right] \nonumber
    \end{align}
Note that \eqref{eqn:Aisindep} is directly implied by \eqref{probAi2}, and for \eqref{eqn:depletionrate} we simply use that the depletion rate of $i$ in sMc is $  \sum_{k \in N_i} \psi_k(t) \cdot \gamma_k$. 
\end{proof}

\subsection{Proof of \Cref{clm:important_i}} \label{app:proofofclaimimportanti}
\claimimportanti*

\begin{proof}
Recall that \[
\sum_{i\in N_j^\downarrow\setminus I_j} \frac{\lambda_i}{\mu_i + \Gamma_i} \overset{\eqref{ineq:abandonment_bound}}{\leq}  \sum_{i\in N_j^\downarrow\setminus I_j} \frac{x_{i,a}}{\mu_i}   \overset{\eqref{eqn:jBalancedness}}{\leq}  \frac{1}{2(1-\eps')}.
\]
By the definition of $I_j$, for any $i \in N_j^{\downarrow} \setminus I_j$ we have $\Gamma_i < \eta \cdot \mu_i$. 
Thus \[
\sum_{i\in N_j^\downarrow\setminus I_j} \frac{\lambda_i}{\mu_i} \leq (1+\eta ) \sum_{i\in N_j^\downarrow\setminus I} \frac{x_{i,a}}{\mu_i} \le (1+\eta ) \cdot \frac{1}{2(1-\eps')}.
\]
Recall that $p_{i,j} := \frac{x_{i,j}/\gamma_j}{x_{i,a}/\mu_i} \ge 1-\eps'$ for all $i \in N_j$ after performing \hyperref[transf]{\textsf{Instance Transformation}}. Using this along with Constraint \eqref{eqn:tightOnlineFlow} of \eqref{TLPon}, we can bound
\[
\sum_{i\in N_j^\downarrow\setminus I_j} \frac{x_{i,a}}{\mu_i} \leq \frac{1}{1-\eps'} \cdot \sum_{i\in N_j^\downarrow\setminus I_j} \frac{x_{i,j}}{\gamma_j} \leq \frac{1}{1-\eps'} \cdot \left(1-e^{-\frac{1+\eta }{2(1-\eps')}}\right) .
\]
As $\sum_{i\in  N_j^\downarrow} \frac{x_{i,a}}{\mu_i} \ge (1-\eps') / 2$ by \Cref{eqn:jBalancedness}, we have that
\[
\sum_{i \in I_j} \frac{x_{i,a}}{\mu_i} = \sum_{i\in  N_j^\downarrow} \frac{x_{i,a}}{\mu_i} - \sum_{i\in  N_j^\downarrow\setminus I_j} \frac{x_{i,a}}{\mu_i} \geq \frac{1-\eps'}{2} - \frac{1}{1-\eps'} \cdot \left(1-e^{-\frac{1+\eta }{2(1-\eps')}}\right).
\] 
As $\eps' \rightarrow 0$ and $\eta \rightarrow 0$ the right-hand side approaches $\frac{1}{2} - (1 - e^{-1/2}) \approx 0.107$. It is straightforward to see that if $\eps', \eta $ are sufficiently small, in particular at most $10^{-3}$, we have $\sum_{i \in I} \frac{x_{i, a}}{\mu_i} \geq 0.1$, as desired.
\end{proof} 

\subsection{Proof of \Cref{clm:Ij_bound}} \label{app:proofclaimIjbound}
\clmIjbound*

\begin{proof}
As $\tau_i$ has density function $(\mu_i + \lambda_i) \exp ( - (\mu_i + \lambda_i) z )$ for $z > 0$, we can compute
\begin{align}
    & \ex{\exp\left(-\left(1- c \cdot \mathbbm{1}\left[\tau_i \geq \tau_i^*\right]\right) \cdot \tau_i \Gamma_i \right)} \nonumber \\ & \quad = \int_{0}^\infty (\mu_i+\lambda_i)\exp(-(\mu_i+\lambda_i)z) \cdot \exp\left(-\left(1- c \cdot \mathbbm{1}\left[z \geq \tau_i^*\right]\right) \cdot z \Gamma_i \right) d z  \nonumber \\ & \quad = \int_{0}^{\tau_i^*} (\mu_i+\lambda_i)\exp(-(\mu_i+\lambda_i)z) \cdot \exp\left(-z \Gamma_i \right) d z  \nonumber \\ 
    & \quad \quad \quad \quad \quad \quad \quad \quad + \int_{\tau_i^*}^\infty (\mu_i+\lambda_i)\exp(-(\mu_i+\lambda_i)z) \cdot \exp\left(-\left(1-c\right) z \Gamma_i \right) d z  \nonumber \\ & \quad = \frac{(\mu_i+\lambda_i) \cdot \left(1 - e^{-\left(\mu_i+\lambda_i + \Gamma_i \right)\tau_i^*}\right) }{\mu_i+\lambda_i + \Gamma_i}  + \frac{(\mu_i+\lambda_i)e^{-\left(\mu_i + \lambda_i + (1-c)\Gamma_i \right)\tau_i^*}}{\mu_i + \lambda_i + (1-c)\Gamma_i} \notag \\ & \quad \geq \frac{(\mu_i+\lambda_i) \cdot \left(1 - e^{-\left(\mu_i+\lambda_i + \Gamma_i \right)\tau_i^*}\right) }{\mu_i+\lambda_i + \Gamma_i}  + \frac{(\mu_i+\lambda_i)e^{-\left(\mu_i + \lambda_i + \Gamma_i \right)\tau_i^*}}{\mu_i + \lambda_i + (1-c)\Gamma_i} \notag  \ .
\end{align}
Applying Jensen's inequality with function $x \mapsto \frac{\mu_i+\lambda_i}{\mu_i + \lambda_i + (1-x) \Gamma_i}$, we have that
\begin{align}
\ex{\exp\left(-\left(1- c \cdot \mathbbm{1}\left[\tau_i \geq \tau_i^*\right]\right) \cdot \tau_i \Gamma_i \right)}  \ge
     \frac{\mu_i + \lambda_i}{\mu_i+\lambda_i+ (1-c') \Gamma_i} \label{jensenconsequence}
\end{align}
 for $c' := c \cdot e^{-(\mu_i+\lambda_i+ \Gamma_i )\tau_i^*}. $ Recall that $\Gamma_i \geq \eta  \cdot \mu_i$ for $i \in I_j$, and $\lambda_i \le \eps^2 \cdot \mu_i$ by \hyperref[propertyv]{Property (v)}; thus, we have
\[ 
    (\mu_i+\lambda_i+ \Gamma_i)\tau_i^* \leq (\mu_i(1+\eps^2) + \Gamma_i)\tau_i^*  \leq \left(\frac{1+\eps^2}{\eta } + 1\right) \tau_i^* \Gamma_i = \frac{1+\eta  + \eps^2}{\delta \eta }
\] 
Therefore $c' \ge c''$ for  $
   c'' := c \cdot \exp \left( -\frac{1+\eta + \eps^2}{\delta \eta} \right) \in (0,1)$. 
From inequality~\eqref{jensenconsequence} we then have
\begin{align}
    \ex{\exp\left(-\left(1- c \cdot \mathbbm{1}\left[\tau_i \geq \tau_i^*\right]\right) \cdot \tau_i \Gamma_i \right)} &\ge \frac{\mu_i + \lambda_i}{\mu_i+\lambda_i+ (1-c'') \Gamma_i} \\
    &= \frac{\mu_i + \lambda_i}{\mu_i + \lambda_i + \Gamma_i} \cdot \left(1 + \frac{c''\Gamma_i}{\mu_i+\lambda_i+ (1-c'') \Gamma_i}\right) \nonumber \\
    &\geq \frac{\mu_i + \lambda_i}{\mu_i + \lambda_i + \Gamma_i} \cdot \left(1 + \frac{c''\eta\mu_i}{(1+\eps^2)\mu_i+ (1-c'') \eta\mu_i}\right) \nonumber \\
    &= \frac{\mu_i + \lambda_i}{\mu_i + \lambda_i + \Gamma_i} \cdot \left(1 + \frac{c''\eta}{1+\eps^2+ (1-c'') \eta}\right) \nonumber \\
    &\geq \frac{\mu_i + \lambda_i}{\mu_i + \lambda_i + \Gamma_i} \cdot (1 + \tilde{c}) \nonumber
\end{align}
for the constant $\tilde{c} := \frac{c'' \eta}{2 +(1-c'')\eta} > 0$ independent of $\eps$, where we used $\Gamma_i \geq \eta\mu_i$ and $\lambda_i \leq \eps^2\mu_i$ in the second-to-last inequality.
\end{proof}

\subsection{Details from Proof of \Cref{lem:correl}}\label{prf:correl}

\subsubsection{Existence of $I^{\rm core}$}\label{app:core_set}
Via the definition of $J^{\rm correl}$ and property \eqref{eqn:jBalancedness}, we have for every $k \in J^{\rm correl}$ that 
\[
    \sum_{i \in N_k^\uparrow \cap N_j^\uparrow} x_{i,a} / \mu_i \ge \frac{1-\eps'}{2} - \kappa \ . 
\]
Consequently, we have
\begin{align}
    \sum_{i' \in N_j^\uparrow} \sum_{k \in J^{\rm correl} \cap N_i}  \gamma_k \cdot \frac{x_{i',a}}{\mu_{i'}} \cdot \mathbbm{1}[i' \in N_k] \geq \sum_{k \in J^{\rm correl} \cap N_i} \gamma_k \cdot \left(\frac{1-\eps'}{2} - \kappa \right) = \left(\frac{1-\eps'}{2} - \kappa \right) \tilde{\Gamma} . \label{ineq:tot_val}
\end{align} 

Now, define the set of offline types in $N_j^{\uparrow}$ that have at most $\tilde{\Gamma}/2$ arrival rate from their neighbors in $J^{\rm correl} \cap N_i$ by
\[ 
    U := \left \{ i' \in N_j^\uparrow : \sum_{k \in J^{\rm correl} \cap N_i} \gamma_k \mathbbm{1}[k \in N_{i'}] \leq \frac{\tilde{\Gamma}}{2} \right \}
\]
and let $O := N_j^\uparrow \setminus U$. We then define $x_U := \sum_{i' \in U} \frac{x_{i', a}}{\mu_{i'}}$ and $x_O := \sum_{i' \in O} \frac{x_{i', a}}{\mu_{i'}}$. The definition of $U$ gives
\begin{align*}
    \sum_{i' \in U} \sum_{k \in J^{\rm correl} \cap N_i}  \gamma_k \cdot \frac{x_{i',a}}{\mu_{i'}} \cdot \mathbbm{1}[i' \in N_k] \leq \frac{x_U \tilde{\Gamma}}{2} 
\end{align*} and trivially we have
\begin{align*}
    \sum_{i' \in O} \sum_{k \in J^{\rm correl} \cap N_i}  \gamma_k \cdot \frac{x_{i',a}}{\mu_{i'}} \cdot \mathbbm{1}[i' \in N_k] \leq {x_O \tilde{\Gamma}} .
\end{align*}
Since $x_O + x_U \in [\frac{1-\eps'}{2}, \frac{0.5}{1-\eps'}]$ by property \eqref{eqn:jBalancedness}, inequality \eqref{ineq:tot_val} implies
\begin{align*}
     \left(\frac{0.5}{1-\eps'}-x_O\right) \cdot \frac{\tilde{\Gamma}}{2} + x_{O} \tilde{\Gamma} \ge x_U \cdot \frac{\tilde{\Gamma}}{2} + x_O \cdot \tilde{\Gamma} \geq \left(\frac{1-\eps'}{2} - \kappa \right) \tilde{\Gamma}  
\end{align*}
or equivalently
\begin{align}
    x_O \geq {1-\eps'} - 2\kappa - \frac{0.5}{1-\eps'} \ . \label{ineq:x_O_lb}
\end{align}
 Letting $I^{\rm core} = O$ hence proves the existence of the desired core set. We also infer that 
\begin{align*}
    \sum_{{i'} \in I^{\rm core}} \lambda_{i'} & \geq \sum_{{i'} \in I^{\rm core}} \sum_{k\in J^{\rm correl} \cap N_i} x_{{i'},k} \cdot \mathbbm{1}[{i'} \in N^\uparrow_k]
    \\ & \geq (1-\eps') \sum_{{i'} \in I^{\rm core}} \sum_{k\in J^{\rm correl} \cap N_i}\gamma_k \cdot \frac{x_{{i'},a}}{\mu_{i'}} \cdot \mathbbm{1}[{i'} \in N^\uparrow_k] && p_{i,j} \geq 1-\eps'
    \\ & \geq \frac{1-\eps'}{2} \left(\sum_{{i'}\in I^{\rm core}} \frac{x_{{i'},a}}{\mu_i}\right)  \left(\sum_{k\in J^{\rm correl} \cap N_i} \gamma_k\right) 
    \\ & \geq \frac{(1-\eps') u(\eps', \kappa)}{2}  \left(\sum_{k\in J^{\rm correl} \cap N_i} \gamma_k\right)
    \\ &  \geq 0.045\left(\sum_{k\in J^{\rm correl} \cap N_i} \gamma_k\right) \ , && \eps',\kappa \in (0,0.1)
\end{align*}
where the second inequality follows from the definition of occupied types and the third inequality is by inequality \eqref{ineq:x_O_lb}. Inequality \eqref{ineq:u_lb} is now proved.
\hfill \qedsymbol

\subsubsection{Concentration bounds for $\sum_{t = 1}^T \eta^{\rm e}_t$ and $\sum_{t = 1}^T \eta^{\rm n}_t$}\label{app:concentration_bounds}
We will make use of the following facts:
\begin{fact}[c.f. \cite{janson2018tail} Theorem 5.1]\label{fact:exp_concentration}
    Given a random variable $Y$ that is a sum of $n$ i.i.d. exponential random variables, we have for all $\theta\geq 1$
\begin{eqnarray} 
\pr{Y\geq \theta \ex{Y}} \leq \frac{1}{\theta}e^{-n(\theta-1-\ln(\theta))} \ . \notag
\end{eqnarray}
\end{fact}
\begin{fact}[Paley–Zygmund inequality]\label{fact:pz_ineq}
     Given a random variable $Y\geq 0$ and parameter $0 < \theta < 1$, we have
\begin{eqnarray} 
\pr{Y\geq \theta \ex{Y}} \geq (1-\theta)^2\frac{\expar{Y}^2}{\ex{Y^2}} \ . \notag
\end{eqnarray}
\end{fact}

By \Cref{clm:nonbusy}, ${\eta_t^{\rm n}} \geq_{st} Z_{I^{\rm core}}$, where $Z_{I^{\rm core}}$ is a hyperexponential random variable specified in the claim. Note that we can directly bound the first and second moments of the hyperexponential distribution of \Cref{clm:nonbusy} as follows:
\begin{claim}\label{clm:ZS_moment_bounds}
For $Z_S$ denoting the hyperexponential distribution in \Cref{clm:nonbusy} we have
    \begin{eqnarray*}
    \ex{Z_S} = \frac{1}{\sum_{i'\in S} \lambda_{i'}}\sum_{i' \in S}  \frac{\lambda_{i'}}{\mu_{i'} + \Gamma_{i'}} \overset{\eqref{ineq:abandonment_bound}}{\geq} (1-\eps') \sum_{{i'} \in S} \frac{x_{{i'},a}}{\mu_i} \cdot \frac{1}{\sum_{{i''} \in S} \lambda_{i''}}
\end{eqnarray*}
and 
\begin{eqnarray*}
    \Var[Z_S] \leq \ex{Z_S^2} = \frac{2}{\sum_{i' \in S} \lambda_{i'}}\sum_{i'\in S}  \frac{\lambda_{i'}}{(\mu_{i'} + \Gamma_{i'})^2} \leq \frac{2}{\min_{i' \in S} (\mu_{i'} + \Gamma_{i'})} \cdot \ex{Z_S}.
\end{eqnarray*}
\end{claim}
Therefore, by the properties of $I^{\rm core}$ and \Cref{clm:ZS_moment_bounds}, we obtain 
\begin{align}\label{ineq:z_bounds}
    \ex{\eta_t^{\rm n}} \geq (1-\eps') u(\eps', \kappa) \cdot \frac{1}{\tilde{\Lambda}} \ .
\end{align}
Using \Cref{fact:pz_ineq} for $\theta=\frac{1}{2}$, we obtain
\begin{align}
    \pr{\sum_{t=1}^T \eta^{\rm n}_t \geq \frac{1}{4\delta' \tilde{\Gamma}}}  &\geq \frac{\ex{\sum_{t=1}^T \eta^{\rm n}_t}^2}{\Var\left(\sum_{t=1}^T \eta^{\rm n}_t\right) + \ex{\sum_{t=1}^T \eta^{\rm n}_t}^2} \label{concentration1}
    \\ = \frac{1}{4} \cdot \frac{T \cdot \ex{\eta_t^{\rm n}}^2}{\Var(\eta_t^{\rm n}) + T \cdot \ex{\eta_t^{\rm n}}^2} \nonumber
    \\ &\ge \frac{1}{4} \cdot \frac{T \cdot \ex{\eta_t^{\rm n}}^2}{\frac{2}{{\tilde{\Gamma}}/{2}} \cdot \ex{\eta_t^{\rm n}} + T \cdot \ex{\eta_t^{\rm n}}^2}   \label{concentration2}
    \\ &= \frac{1}{4} \cdot \frac{1}{\frac{2}{T\ex{\eta_t^{\rm n}} \cdot ({\tilde{\Gamma}}/{2})} + 1} \nonumber
    \\ &\ge \frac{1}{4} \cdot \frac{1}{\frac{8}{(1-\eps')u(\eps', \kappa)} \delta' + 1} \label{concentration3}
    \\ & \ge \frac{1}{4} \cdot \frac{1}{90\delta' + 1} \ , \label{concentration4}
\end{align}
where equality~\eqref{concentration1} uses independence of $(\eta_t^{\rm n})_t$, inequality~\ref{concentration2} uses \Cref{clm:nonbusy} and \Cref{clm:ZS_moment_bounds} along with inequality \eqref{ineq:core_gamma_lb}, inequality~\eqref{concentration3} uses inequality \eqref{ineq:z_bounds}, and inequality ~\eqref{concentration4} follows from \eqref{ineq:u_lb}. Concurrently, \Cref{fact:exp_concentration} and \eqref{ineq:u_lb} give
\begin{eqnarray} 
    \pr{\sum_{t=1}^T \eta^+_t\geq \frac{2}{\delta' \tilde{\Gamma}}} \leq \frac{1}{2}  e^{-T(1 -\ln 2)} \leq \frac{1}{2} e^{-\frac{(1-\eps')u(\eps', \kappa)(1 -\ln 2)}{2\delta'}} \leq \frac{1}{2} e^{-\frac{0.01}{\delta'}} \ ,\notag
\end{eqnarray}
proving the desired concentration bounds. \hfill \qedsymbol

\subsection{Proof of \Cref{lem:rand}}\label{prf:rand}
    Consider some $k \in J^{\rm indep} \cap N_i$. 
    Let $I^k := N_k^\uparrow \setminus N_j^\uparrow$, and recall $\sum_{i' \in I^k} {x_{i', a}}/{\mu_{i'}} \geq \kappa$ follows from having $k \in J^{\rm indep}$. We can upper bound $\theta_{i,k}$ by considering only the times when queues in $I_k$ are empty, i.e., when $\sum_{i' \in I^k} Q_{i'}^w(t) = 0$. Recall that the queues in TOP evolve identically in $Q^{\textsf{weak}}$ and $Q^{\textsf{ind}}$. Therefore, from \Cref{claim:stationarydistiMcPoisson}, we know that the unconditional probability of having $\sum_{i' \in I^k} Q_{i'}^w = 0$ is given by 
    \begin{align*}
        \pr{\sum_{i' \in I^k} Q_{i'}^w(t) = 0} = \exp\left(-\sum_{i' \in I^k} x_{i', a} / \mu_{i'} \right) \leq \exp(-\kappa) \ , 
    \end{align*} for every time $t \geq 0$.  Note that under the weakly correlated Markov chains, the event that the queues in $I^k$ are empty is independent of $\textsf{TE}_j$ and $(A_{i'})_{i'}$. The time $t_i$ is additionally drawn independently of $\textsf{TE}_j$, $(A_{i'})_{i' \in N_j^\downarrow}$, and the queues in $I^k$. Thus, we have 
     \begin{align}
        \ex{\theta_{i,k} \left | t_i, (A_{i'})_{i' \in N_j^\downarrow}, \textsf{TE}_j \right.} = \ex{\theta_{i,k} \mid t_i} \le \int_{t^\infty - t_{i}}^{t^\infty} \pr{\sum_{i' \in I^k} Q_{i'}^w(t) = 0} dt \leq \exp(-\kappa) \cdot t_i \ . \label{ineq:expectation_theta}
    \end{align}
     Using inequality \eqref{ineq:expectation_theta} for every $k \in J^{\rm indep} \cap N_i$, we can apply Markov's inequality to obtain
    \begin{align*}
        \pr{\left.\sum_{k \in J^{\rm indep} \cap N_i} \gamma_k \theta_{i,k} \geq \tilde{c} \cdot \exp(-\kappa) \cdot t_i \sum_{k \in J^{\rm indep} \cap N_i} \gamma_k \right | t_i, (A_{i'})_{i' \in N_j^\downarrow}, \textsf{TE}_j } \leq \frac{1}{\tilde{c}} \ .
    \end{align*}
    Consequently, letting $\tilde{c} = \frac{1}{\exp(-\kappa)} - \tilde{\eps}$ for some fixed $\tilde{\eps} \ll \frac{1}{\exp(-\kappa)}$ proves the lemma.  \hfill \qedsymbol

\section{Improved Competitive Ratio}\label{app:competitive}
In this section, we adapt our machinery for competing with the optimal \underline{online} algorithm to improve upon the best-known competitive ratio (i.e., competing against the optimal \underline{offline} algorithm. \cite{patel2024combinatorial} give an algorithm that achieves a $(1-1/\sqrt{e})$-fraction of $\optoff$'s stationary reward; we both give a simpler proof of this existing bound and improve on it, as in the result below. 
\compratiothm*

The algorithm for \Cref{thm:compratio} and the LP relaxation it uses are nearly identical to the ones in \Cref{sec:algo}, albeit with slight modifications to reflect the optimal offline benchmark.

\subsection{Algorithm with Pivotal Sampling}

We use an LP relaxation almost identical to \eqref{TLPon}, except that it does not include constraint \eqref{eqn:tightOnlineConstraint}. This LP is a tightening of the one used by \cite{patel2024combinatorial} via imposing constraint \eqref{eqn:offtightOnlineFlow} for every subset $H \subseteq I$. 

\begin{align}
	\nonumber  \max \quad &  \sum_{i \in I} \sum_{j \in J} r_{i,j} \cdot x_{i,j} && \tag{TLP$_{\text{off}}$} \label{TLPoff} \\
	\textrm{s.t.} \quad & x_{i,a} + \sum_j x_{i,j} = \lambda_i && \text{for all } i \in I   \label{eqn:offtightOfflineFlow}\\
	&  \sum_{i \in H} x_{i,j} \le \gamma_j \cdot \left(1 - \exp\left(-\sum_{i \in H}  \lambda_i / \mu_i \right)\right) && \text{for all } j \in J, H  \subseteq I\label{eqn:offtightOnlineFlow} \\
 & x_{i,j}, x_{i,a} \ge 0 & \text{for all } i \in I, j \in J \label{eqn:offnonnegativity}
\end{align}
Identically to \Cref{claim:lprelaxation} and \Cref{TLPpolysolvable}, we can show that \eqref{TLPoff} is a valid relaxation of $\optoff$ and we can solve \eqref{TLPoff} in polynomial time. Moreover, the algorithm is identical to \Cref{alg:corrprop}, other than the relevant LP relaxation that is solved (in \Cref{line:solveLP}) and the formula for $p_{i,j}$ (in \Cref{line:setPij}). 

\begin{algorithm}[H]
	\caption{Correlated Proposals for \eqref{TLPoff}}
	\label{alg:corrpropoff}
	\begin{algorithmic}[1]
 \State   Solve \eqref{TLPoff} for $\{x_{i,j}\}$   \label{line:offsolveLP}
 \State  $p_{i,j} \gets \frac{x_{i,j} / \gamma_j}{1 - \exp(-\lambda_i / \mu_i)}$  \label{line:offsetPij}
 \medskip
    
		\ForAll{timesteps where an online $j$ arrives} \label{line:offloop-start}
        \State Label present and unmatched offline nodes by $\{1, 2, \ldots, N\}$, in decreasing order of $r_{i,j}$
        \State $\textsf{Proposers} \gets \textsf{PS}( (p_{i,j}))$ \label{lin:PS} \Comment{$\textsf{PS}(\cdot)$ denotes \emph{pivotal sampling}, as in \Cref{def:pivotalsampling}}
		\State If \textsf{Proposers} is non-empty, match $j$ to a node of type $i^*  = \text{Type} \left( \min \left(  \textsf{Proposers}  \right)   \right)  $   \label{lin:offmatch}
        %\EndIf        
        \EndFor
	\end{algorithmic}
\end{algorithm}	

Note that $p_{i,j} \in [0,1]$ by Constraint~\eqref{eqn:offtightOnlineFlow}.

\subsection{$\mathbf{(1-1/\sqrt{e})}$-Competitive Ratio via Independent Queues}
Let $Q(t)$ be the Markov chain for the number of available nodes of each type under \Cref{alg:corrpropoff}. Similarly to \Cref{def:imc} for \Cref{alg:corrprop}, we define the independent Markov chains $Q^{\rm ind}$ for \Cref{alg:corrpropoff}, where queue $i$ is increased by 1 at rate $\lambda_i$ and at queue length $k$, it is depleted at rate $k \cdot (\mu_i + \sum_{j \in J} \gamma_j \cdot p_{i,j})$, independent from other queues. In the same fashion as \Cref{claim:stochasticdominance}, we can show that $Q(t)$ stochastically dominates $Q^{\rm ind}(t)$. Therefore, we now focus on the stationary distribution of $Q^{\rm ind}(t)$.

\paragraph{Steady-state analysis of $\boldsymbol{Q^{\rm ind}}$(t).} Since $Q_i^{\rm ind}(t)$ at state $k$ is depleted at rate $k \cdot (\mu_i + \sum_{j \in J} \gamma_j \cdot p_{i,j})$, we can simplify the death rate below, where we use the notation $u_i := 1-\exp(-\lambda_i/ \mu_i)$:
\begin{align*}
\mu_i + \sum_j \gamma_j \cdot p_{i,j} &= \mu_i + \sum_j \frac{x_{i,j}}{u_i} 
\\ & \leq \mu_i +  \frac{ \lambda_i}{u_i}
\\ & = \mu_i \left(1 + \frac{\lambda_i/\mu_i}{u_i}\right) \ . 
\end{align*}
Thus, by \Cref{claim:stationarydistbirthdeath} and rearrangement of the above expression, we infer that the stationary distribution of $Q_i^{\rm ind}(t)$ stochastically dominates a $\textsf{Pois}(u_i \cdot \lambda_i/\mu_i \cdot (u_i + \lambda_i/\mu_i))^{-1})$ distribution. 

% Therefore, hereafter, we analyze this queue that goes up by 1 at rate $\lambda_i$ and goes down, at length $k$, with rate $k \cdot \mu_i (1 + \lambda_i / (\mu_i u_i))$, independently from other queues. We abuse the notation to let  $Q_i^{\textsf{ind}}$ refer to this queue. 

\paragraph{The $\mathbf{(1-1/\sqrt{e})}$ bound.} Now, a lower bound of $(1-1/\sqrt{e})$ on the competitive ratio is straightforward. By \Cref{lem:weightedPoisCvx}, we have 
\begin{align}
R_{j}(w) &:= \sum_{i : r_{i,j} \ge w} p_{i,j} \cdot \text{Pois}\left(\frac{\lambda_i/\mu_i}{u_i+\lambda_i/\mu_i} \cdot u_i\right) \notag \\ 
& \cvxle \textup{Pois} \left( \sum_{i:r_{i,j} \geq w}\frac{\lambda_i/\mu_i}{u_i+\lambda_i/\mu_i} \cdot \frac{x_{i,j}}{\gamma_j} \right) \ . \label{ineq:r_pois_formula}
\end{align}
As $\min(1,x)$ is concave, we have that $\mathbb{E}[\min(1, R_{j}(w))]$ is lower bounded by  
\begin{align}
 \mathbb{E} \left[ \min \left( 1 , \text{Pois} \left(\sum_{i:r_{i,j} \geq w}\frac{\lambda_i/\mu_i}{u_i+\lambda_i/\mu_i} \cdot \frac{x_{i,j}}{\gamma_j}\right) \right) \right] 
&= 1 - \exp \left( - \sum_{i:r_{i,j} \geq w}\frac{\lambda_i/\mu_i}{u_i+\lambda_i/\mu_i} \cdot \frac{x_{i,j}}{\gamma_j} \right)  \notag
\\ & \geq 1 - \exp \left( - \left( \min_{i \in I} \frac{\lambda_i/\mu_i}{u_i+\lambda_i/\mu_i}\right) \cdot \sum_{i:r_{i,j} \geq w} \frac{x_{i,j}}{\gamma_j} \right) \label{ineq:loose_ratio_bound} \\ &
\ge \left(1-\exp\left(-\min_{i \in I} \frac{\lambda_i/\mu_i}{u_i+\lambda_i/\mu_i}\right)\right) \cdot  \sum_{i : r_{i,j} \ge w} \frac{x_{i,j}}{ \gamma_j } \notag
\\ & \geq (1-1/\sqrt{e}) \cdot  \sum_{i : r_{i,j} \ge w} \frac{x_{i,j}}{ \gamma_j } \ , \label{ineq:sqrte_bound}
\end{align}
where the second inequality holds because $1-\exp(-bx) \geq (1-\exp(-b)) \cdot x$ for every $x \in (0,1)$, and the third inequality follows from the next claim with $x = \lambda_i / \mu_i$.
\begin{claim}\label{clm:presence_function}
    The function $f(x) := \frac{x}{x + 1-\exp(-x)}$ is increasing and we have $\lim_{x \to 0} f(x) = \frac{1}{2}$.
\end{claim}
\begin{proof}
    The monotonicity is immediate from deriving $f'(x) = \frac{\mathrm{e}^{x} \left(\mathrm{e}^{x} - x - 1\right)}{\left(\left(x + 1\right) \mathrm{e}^{x} - 1\right)^{2}} \geq 0$. The limit at 0 can be found by the L'Hôpital's rule as 
    \begin{align*}
        \lim_{x \to 0} \frac{x}{x + 1-\exp(-x)} = \lim_{x \to 0} \frac{1}{1 + \exp(-x)} = \frac{1}{2} \ . \hfill  \qedhere
    \end{align*} 
\end{proof}
We can now calculate the stationary reward from matching to $j$ by integrating the complementary CDF, as follows:
\begin{align*}
    \gamma_j \cdot  \mathbb{E}_{Q \sim \pi^{\textsf{ind}}} [\textup{\textsf{ALG}}(j, Q)] &\ge \gamma_j \cdot \int_0^\infty \E[\min(1,R(w))] \, dw && \hspace{-7em} \text{\Cref{claim:stochasticdominance}, \Cref{claim:gainequalexpectedmin}}\\
    &\ge \gamma_j \cdot \int_0^\infty (1 - 1/\sqrt{e}) \cdot \sum_{i: r_{i,j} \ge w} x_{i,j} / \gamma_j \, dw && \hspace{-2.7cm} \text{By inequality } \eqref{ineq:sqrte_bound}\\
    &\ge (1-1/\sqrt{e}) \cdot \int_0^\infty \sum_i x_{i,j} \cdot \mathbbm{1}[r_{i,j} \ge w]  \, dw   
    = (1-1/\sqrt{e}) \cdot \sum_i x_{i,j} \cdot r_{i,j}.
\end{align*}
This demonstrates that \Cref{alg:corrpropoff} achieves a $(1-1/\sqrt{e})$-competitive ratio to $\textup{OPT}\eqref{TLPoff}$.

\subsection{A Sharpened $\mathbf{(1-1/\sqrt{e} + \eta)}$-competitive Analysis}
Similarly to \Cref{sec:sharpened}, we fix an impatient type $j$ and show that $\textsf{ALG}_j \geq (1-1/\sqrt{e} + \Omega(1)) \cdot \textsf{LP-Gain}_j$. Since this holds for every $j \in J$, \Cref{thm:compratio} follows immediately. Establishing this result is significantly less challenging than \Cref{thm:main}. 

We use \hyperref[case3]{Case 3} (\Cref{obscase3}) in \Cref{sec:approximation} almost identically to leverage the correlation in pivotal sampling. Consequently, we consider a second case that relies on a classification of offline types to \emph{abundant} and \emph{scarce} queues---which depends on the magnitude of $\lambda_i / \mu_i$. If a constant fraction of $\textsf{LP-Gain}_J$ is from abundant types, our analysis can be tightened to show a $(1-1/\sqrt{e} + \Omega(1))$-competitive ratio on $j$. Lastly, the third case considers the scenario where the LP matches $j$ mostly to scarce offline types. In this case, we show that if the majority of proposal probabilities are close to 1---i.e., the first case does not apply to the instance---scarce queues cannot saturate $j$. Hence, repeating \Cref{obscase1} implies the $(1-1/\sqrt{e} + \Omega(1))$-competitive ratio of our algorithm on $j$. Unlike when competing with optimum online, VWHC instances are easily handled here because if $j$ is saturated, a constant fraction of its matches must come from abundant offline types, which yields a competitive ratio better than $(1-1/\sqrt{e})$, as discussed in the second case.

Next, we formally discuss the different cases. In the following observations, we can assume $b(\cdot)$ is a continuous, increasing function with $b(0) \geq 0$. 

 \paragraph{Case 1: Some of $j$'s proposal probabilities $\{p_{i,j}\}$ are bounded away from 1.\label{cr:case1}} We begin with this case, which is identical to \Cref{obscase3} for our approximation analysis. We recall the intuition which is that when a non-negligible number of $j$'s neighbor have a marginal probability bounded away from 1, we can take advantage of the correlation introduced by pivotal sampling. In particular, we can tighten our bounding of $R_j(w)$ in the convex order, resulting in the following observation. Since the proof is nearly identical to \Cref{obscase3}, it is omitted. 
\begin{obs}\label{obs:cr_case1}
    Consider some small $\eps \in (0, 0.1)$. Suppose that there exists $S \subseteq I$ such that $\sum_{i \in S} x_{i,j} \ge \epsilon \cdot \gamma_j$ and $\max_{i \in S} p_{i,j} \leq 1-\epsilon$. Then, we have $\textup{\textsf{ALG}}_j \geq (1-1/\sqrt{e}+b (\epsilon)) \cdot \textup{\textsf{LP-Gain}}_{j}$ for some $b (\epsilon) > 0$.
\end{obs}

\paragraph{Case 2: Some of $j$'s gain is from ``abundant'' queues. \label{cr:case2}}
The second case relies on a tighter analysis than inequality \eqref{ineq:loose_ratio_bound}. To this end, we require a classification of queues. 

\begin{definition}[Queue classification]
    We call queue $i$ ``{abundant}'' if it satisfies $\lambda_i / \mu_i \geq 1$. Otherwise, it is ``{scarce}''. 
\end{definition}

We now show that the lower bound \eqref{ineq:loose_ratio_bound} is loose if a non-negligible fraction of $\textsf{LP-Gain}_j$ comes from abundant queues. The rationale is that, as shown by \Cref{clm:presence_function}, the bound $\lambda_i / \mu_i \cdot (u_i + \lambda_i / \mu_i)^{-1} \geq 1/2$ in inequality \eqref{ineq:loose_ratio_bound} is loose for types $i$ that have a large $\lambda_i  / \mu_i$. This motivates us to consider the following case.

\begin{obs}\label{obs:cr_case2}
    Consider some small $\eps \in (0, 0.1)$. Suppose that there exists $S \subseteq I$ of \underline{abundant} queues such that $\sum_{i \in S} x_{i,j}/\gamma_j \ge \epsilon$. Then, we have $\textsf{ALG}_j \geq (1-1/\sqrt{e} + b(\eps)) \cdot \textsf{LP-Gain}_j$ for some $b(\eps) > 0$. 
\end{obs}

\begin{proof}
    In this case, we can give a sharpened analysis of inequality \eqref{ineq:loose_ratio_bound}. Consider some $w \in {\bb R}_{\geq 0}$. We show that under our observation hypothesis, we have 
    \begin{align}
        \ex{\min(1,R_j(w))} \geq (1-1/\sqrt{e} + b(\eps)) \cdot  \sum_{i \in I:r_{i,j} \geq w} \frac{x_{i,j}}{\gamma_j} \ , \label{ineq:r_abundant_bound}
    \end{align}
    which immediately proves the observation. Consider first the case where $\sum_{i: r_{i,j} \geq w} x_{i,j} / \gamma_j \leq 1-\eps/2$. In this case, repeating \Cref{obscase1} shows that inequality \eqref{ineq:r_abundant_bound} holds with
    $b(\eps) \leq g(1-\eps/2)$. Thus, we now assume $\sum_{i: r_{i,j} \geq w} x_{i,j} / \gamma_j > 1-\eps/2$. Recalling the observation hypothesis, we infer that $\sum_{i \in S: r_{i,j} \geq w} x_{i,j} / \gamma_j \geq \eps/2$. Using this fact, we argue that
 \begin{align}
    &\mathbb{E} \left[ \min \left( 1 , \text{Pois} \left(\sum_{i:r_{i,j} \geq w}\frac{\lambda_i/\mu_i}{u_i+\lambda_i/\mu_i} \cdot \frac{x_{i,j}}{\gamma_j}\right) \right) \right] \notag 
    \\ & \; = 1 - \exp \left( - \sum_{i:r_{i,j} \geq w}\frac{\lambda_i/\mu_i}{u_i+\lambda_i/\mu_i} \cdot \frac{x_{i,j}}{\gamma_j} \right)  \notag
    \\ & \;\geq 1 - \exp \left( - \left( \min_{i \in S} \frac{\lambda_i/\mu_i}{u_i+\lambda_i/\mu_i}\right) \cdot \sum_{i \in S:r_{i,j} \geq w} \frac{x_{i,j}}{\gamma_j} - \left( \min_{i \in I \setminus S} \frac{\lambda_i/\mu_i}{u_i+\lambda_i/\mu_i}\right) \cdot \sum_{i \in I \setminus S:r_{i,j} \geq w} \frac{x_{i,j}}{\gamma_j}  \right) \notag
    \\ & \;\geq 1 - \exp \left( - \frac{1}{2-1/e} \cdot \sum_{i \in S:r_{i,j} \geq w} \frac{x_{i,j}}{\gamma_j} - \frac{1}{2} \cdot \sum_{i \in I \setminus S:r_{i,j} \geq w} \frac{x_{i,j}}{\gamma_j}  \right) \notag
    \\ & \;\geq 1 - \exp \left( - \frac{1}{2} \cdot \sum_{i \in I:r_{i,j} \geq w} \frac{x_{i,j}}{\gamma_j} - \left(\frac{1}{2-1/e} - \frac{1}{2} \right) \cdot \frac{\eps}{2} \cdot \sum_{i \in I:r_{i,j} \geq w} \frac{x_{i,j}}{\gamma_j} \right) \notag
    \\ & \; \geq \left(1- e^{-\left(\frac{1}{2} + \left(\frac{1}{2-1/e} - \frac{1}{2}\right) \cdot \frac{\eps}{2}\right)} \right) \cdot  \sum_{i \in I:r_{i,j} \geq w} \frac{x_{i,j}}{\gamma_j} \ ,
    \end{align}
    where the second inequality uses \Cref{clm:presence_function} with $\frac{{\lambda_i}/{\mu_i}}{u_i + {\lambda_i}/{\mu_i}} \geq \frac{1}{2-1/e}$ for every abundant queue $i$. The third inequality uses the fact that 
    \[
        \sum_{i \in S: r_{i,j} \geq w} x_{i,j} / \gamma_j  \geq \eps/2 \cdot \sum_{i \in I: r_{i,j} \geq w} x_{i,j} / \gamma_j \ . \]
    Finally, the last inequality uses the inequality $1-\exp(-bx) \geq (1-\exp(-b)) \cdot x$ for every $x \in (0, 1)$. 

    Now, accounting for the first case of $b(\eps) \leq g(1-\eps/2)$, the observation follows immediately with \[
        b(\eps) =  \min \left \{1/\sqrt{e} - \exp\left({-\left(\frac{1}{2} + \left(\frac{1}{2-1/e} - \frac{1}{2}\right) \cdot \frac{\eps}{2}\right)}\right) \ , \  g(1-\eps/2)\right \} \ . \qedhere \]
\end{proof}

\paragraph{Case 3: Scarce queues cannot saturate $\boldsymbol{j}$.\label{cr:case3}}
We are now in a position where abundant queues constitute a small fraction of $j$'s matches according to \eqref{TLPoff}. We show that in this case, constraints \eqref{eqn:offtightOnlineFlow} imply that $j$'s match rate in the LP is bounded away from 1. Hence, we can use the same reasoning as \Cref{obscase1} to show that our algorithm is  $(1-1/\sqrt{e} + \Omega(1))$-competitive on $j$. 

\begin{obs}\label{obs:cr_case3}
        Consider some small $\eps \in (0, 0.1)$. Let $S$ be the collection of $j$'s neighboring \underline{scarce} queues that satisfy $p_{i,j} \geq 1-\eps$. If  $\sum_{i \in I \setminus S} x_{i,j}/\gamma_j \leq \eps$, then we have $\textsf{ALG}_j \geq (1-1/\sqrt{e} + b(\eps)) \cdot \textsf{LP-Gain}_j$ for some $b(\eps) > 0$. 
\end{obs}

\begin{proof}
We recall that $u_i = 1-\exp(-\lambda_i/\mu_i)$. By the definition of $S$, we have
\begin{align}
    \sum_{i \in S} \frac{x_{i,j}}{\gamma_j} = \sum_{i \in S} \frac{x_{i,j}}{\gamma_j u_i} \cdot u_i \geq \sum_{i \in S} (1-\eps) \cdot u_i \ , \label{ineq:x_gamma_lb}
\end{align}
where we used the fact that $p_{i,j} \geq 1-\eps$ for every $i \in S$. Combining inequality \eqref{ineq:x_gamma_lb} with constraint \eqref{eqn:offtightOnlineFlow} with $H = S$, we obtain
\begin{align}
    \sum_{i \in S} \left(1-\exp\left(-\frac{\lambda_i}{\mu_i}\right)\right) \le \frac{1}{1-\eps} \cdot \left(1 - \exp\left(-\sum_{i \in S}  \frac{\lambda_i}{\mu_i} \right)\right)  \ . \label{ineq:exponential_comparison}
\end{align}
We now use inequality \eqref{ineq:exponential_comparison} to provide an upper bound for $\sum_{i \in S} \frac{\lambda_i}{\mu_i}$, which in turn leads to the desired result using the same reasoning as Observation \ref{obscase1}. Recall that ${\lambda_i}/{\mu_i} \leq 1$ for every scarce queue $i \in S$. Using the inequality $\frac{x}{2} \leq 1-\exp(-x)$ for $x \in [0,1]$, inequality \eqref{ineq:exponential_comparison} gives
\begin{align*}
    \frac{1}{2} \cdot \sum_{i \in S} \frac{\lambda_i}{\mu_i} \leq \frac{1}{1-\eps} \cdot \left(1 - \exp\left(-\sum_{i \in S}  \frac{\lambda_i}{\mu_i} \right)\right)  \ ,
\end{align*}
which implies \[\sum_{i \in S} \frac{\lambda_i}{\mu_i} \leq \frac{2}{1-\eps} \ . \] 
Therefore, by constraint \eqref{eqn:offtightOnlineFlow}, we have
\begin{align*}
    \sum_{i \in S} \frac{x_{i,j}}{\gamma_j} \leq 1-\exp\left(-\frac{2}{1-\eps}\right) \ ,
\end{align*} which in turn leads to
\[
    \sum_{i \in I} x_{i,j} / \gamma_j = \sum_{i \in S} x_{i,j} / \gamma_j + \sum_{i \in I \setminus S} x_{i,j} / \gamma_j \leq 1-\exp\left(-\frac{2}{1-\eps}\right) + \eps \ .
\]
We can now use a reasoning similar to \Cref{obscase1} to prove the observation with 
\[
    b(\eps) = g \left(1-\exp\left(-\frac{2}{1-\eps}\right) + \eps\right)  \ , \qedhere
\]
Nevertheless, since $g (1-\exp(-\frac{2}{1-\eps}) + \eps) $ is decreasing in $\eps$ whereas we would like $b(\eps)$ to be increasing, we define $b(\eps) = g (1-\exp(-\frac{2}{1-\tilde{\eps}}) + \tilde{\eps}) $ for $\tilde{\eps} = 0.1$; we get a constant $b(\eps) > 0$ for every $\eps \in (0, 0.1)$. 
\end{proof}
Combining \Cref{obs:cr_case1}-\Cref{obs:cr_case3} proves \Cref{thm:compratio}, as discussed next.

\paragraph{Proof of \Cref{thm:compratio}.} Let $\xi > 0$ be a small constant and consider some fixed online type $j \in J$. If $j$ satisfies the conditions of \Cref{obs:cr_case1} or \Cref{obs:cr_case2} with $\eps = \xi$, we have $\textsf{ALG}_j \geq (1-1/\sqrt{e} + b(\xi)) \cdot \textsf{LP-Gain}_j$. Otherwise, letting 
\[
    I_1 = \left\{i \in N_j: p_{i,j} < 1-\xi \right \} \quad \text{and} \quad I_2 = \left\{i \in N_j: \lambda_i / \mu_i \geq 1 \right \}  \ ,
\]
we know that $\sum_{i \in I_1} x_{i,j} / \gamma_j \leq \xi$ and $\sum_{i \in I_2} x_{i,j} / \gamma_j \leq \xi$. Thus, the conditions of \Cref{obs:cr_case3} hold with $\eps = 2\xi$. We hence have $\textsf{ALG}_j \geq (1-1/\sqrt{e} + b(2\xi)) \cdot \textsf{LP-Gain}_j$. Therefore, we choose $\xi = 0.01$ and we always have $\textsf{ALG}_j \geq (1-1/\sqrt{e} + \delta) \cdot \textsf{LP-Gain}_j$, where $\delta = b(0.01)$. Since this argument works for every $j \in J$, \Cref{thm:compratio} immediately follows. \hfill \qedsymbol

\section{Extensions to other dynamic matching settings}\label{app:extensions}

Here, we argue that Theorem~\ref{thm:main} implies improvement on the approximation ratio in more general dynamic stochastic matching settings with (i) symmetric departure rates, or (ii) general (non-bipartite) graphs, which were previously studied by~\cite{aouad2022dynamic}.

\paragraph{Symmetric departure rates.} We define the {\em symmetric setting} as a generalization of the online stationary matching problem where each online node of type $j \in J$ also leaves the system after a duration ${\rm Exp}(\mu_j)$, but remains available for matching anytime before  departure. The online stationary matching problem can be viewed as the limiting case where $\mu_j = +\infty$ for all $j
\in J$, creating an asymmetry of patience between customers and suppliers. \cite{aouad2022dynamic} obtained a $
\frac{1}{2}(1-1/e)$-factor approximation of the online optimum in the symmetric setting, using an implicit reduction to stationary matching. By extending this reduction, Theorem~\ref{thm:main} improves on this result by providing a  $\frac{1}{2}(1-1/e+\delta)$-approximation of the online optimum.

An outline of the reduction is provided here. Let $\pi$ be a feasible policy in the symmetric setting. We show that our LP  in Section~\ref{subsec:tlp} is a valid relaxation for a variant of $\pi$ that operates in an online stationary matching setting. \cite{aouad2022dynamic} used the notion of active/passive agents. An agent is said to be {\em passive}
if it is matched by $\pi$ with another node who arrived to the market earlier; all other non-passive agents are said to be {\em active}. Consequently, we define $x^\pi_{i,j}$ to be the expected average rate of matches between active type $i\in I$ offline nodes with  passive type $j\in J$ online nodes, and similarly define $x^\pi_{j,i}$ as the match rate for the reverse active/passive labeling. Without loss of generality, let us assume that the ``online'' side and ``offline'' side are labeled such that $\sum_{i\in I} \sum_{j\in J} r_{i,j} x^\pi_{i,j} \geq \sum_{i\in I} \sum_{j\in J} r_{i,j} x^\pi_{j,i}$. 

Consider a new policy $\vec{\pi}$ in the symmetric setting that only executes the matches of $\pi$ that involve an active offline node with a passive online node. By construction, we have $r(\vec{\pi}) \geq \frac{1}{2} r(\pi)$. Next, we can show that the variables $x_{i,j}^{\vec{\pi}}$ satisfy all constraints in \eqref{TLPon} (i.e., the expected average match rates induced by $\vec{\pi}$ form a feasible solution of that LP). To do so, we extend the proof ideas of Lemmas~1 and 9 in~\cite{aouad2022dynamic}.\footnote{The authors constructed a policy $\tilde{\pi}$ that mimics $\vec{\pi}$, but it anticipates the future and executes certain matches earlier. Specifically, $\tilde{\pi}$ matches each (passive) online node immediately upon its arrival to the corresponding (active) offline node to whom it is eventually matched by $\vec{\pi}$. Note that $\tilde{\pi}$ is not a valid non-anticipating policy but it generates the same expected long-term average match rates and rewards as $\vec{\pi}$. \cite{aouad2022dynamic} showed  that constraints~\eqref{eqn:tightOnlineConstraint} are satisfied by $\tilde{\pi}$ via induction. Now, observe that  $\tilde{\pi}$ can  match an online node with an offline node only if the latter is {\em present} in the system upon the online node's arrival. Taking expectations and summing over the corresponding matching and presence indicator random variables yields that $\tilde{\pi}$ (and thus $\vec{\pi}$) also satisfies constraints~\eqref{eqn:tightOnlineFlow}. Lastly, it is not difficult to prove that the flow balance constraints~\eqref{eqn:tightOfflineFlow} are met.} Thus a $(1-1/e +\delta)$-approximate policy with respect to \eqref{TLPon} in the stationary matching setting, considering either active offline nodes and passive online nodes, or the reverse online/offline labeling, yields a $\frac{1}{2}(1-1/e +\delta)$-approximation of the online optimum in the original instance of the symmetric setting.

\paragraph{General graphs.} As a further extension, we consider settings where the edge-weighted graph comprises node types $V$ which are no longer  partitioned into offline nodes $I$ and online nodes $J$. Without of generality, this non-bipartite graph is complete as the edge rewards are possibly equal to zero. This is the most general setting considered by \cite{aouad2022dynamic}, where they obtained a $
\frac{1}{4}(1-1/e)$-factor approximation of the online optimum. By reduction to the bipartite setting, Theorem~\ref{thm:main} improves on this result by providing a  $\frac{1}{4}(1-1/e+\delta)$-approximation of the online optimum.

Let $\pi$ be a policy for a general (non-bipartite) instance. Construct an {\em extended instanced}, formed by a bipartite graph $(I,J)$ where both $I$ and $J$ are copies of $V$. Upon the arrival of any node of type $v\in V$ in the original instance, we toss a balanced coin and reassign its type to one of $v$'s two copies in   $I$ or $J$. Next, we consider policy $\pi'$ that simulates $\pi$ but only executes matches on edges between $I$ and $J$ (i.e., a match of $\pi$ is executed by $\pi'$ if the nodes' reassigned types are in opposite partite sets). The important observation is that the type reassignment is independent of the state evolution under $\pi$, and therefore, the conditional probability that a match of $\pi$ is executed by $\pi'$ given the system state is exactly $1/2$. This implies that $r(\pi') \geq \frac{1}{2} r(\pi)$. Now, observe that $\pi'$ is a valid policy for the extended bipartite instance. Combining this argument with our preceding reduction shows that the general setting is reducible to stationary matching with a $\frac{1}{4}$-factor loss.

\newpage
\bibliographystyle{alpha}
\bibliography{a}

\newcommand{\etalchar}[1]{$^{#1}$}
\begin{thebibliography}{FMMM09}

\bibitem[AAS24]{amani2024adaptive}
Alireza AmaniHamedani, Ali Aouad, and Amin Saberi.
\newblock Adaptive policies and approximation schemes for dynamic matching.
\newblock {\em Working paper}, 2024.

\bibitem[ABD{\etalchar{+}}23]{ashlagi2019edge}
Itai Ashlagi, Maximilien Burq, Chinmoy Dutta, Patrick Jaillet, Amin Saberi, and
  Chris Sholley.
\newblock Edge-weighted online windowed matching.
\newblock {\em Math. Oper. Res.}, 48(2):999--1016, 2023.

\bibitem[AGKK23]{antoniadis2023secretary}
Antonios Antoniadis, Themis Gouleakis, Pieter Kleer, and Pavel Kolev.
\newblock Secretary and online matching problems with machine learned advice.
\newblock {\em Discret. Optim.}, 48(Part 2):100778, 2023.

\bibitem[AGKM11]{aggarwal2011online}
Gagan Aggarwal, Gagan Goel, Chinmay Karande, and Aranyak Mehta.
\newblock Online vertex-weighted bipartite matching and single-bid budgeted
  allocations.
\newblock In {\em Proceedings of the 22nd Annual ACM-SIAM Symposium on Discrete
  Algorithms (SODA)}, pages 1253--1264, 2011.

\bibitem[AM23]{aouad2023nonparametric}
Ali Aouad and Will Ma.
\newblock A nonparametric framework for online stochastic matching with
  correlated arrivals.
\newblock In Kevin Leyton{-}Brown, Jason~D. Hartline, and Larry Samuelson,
  editors, {\em Proceedings of the 24th {ACM} Conference on Economics and
  Computation, {EC} 2023, London, United Kingdom, July 9-12, 2023}, page 114.
  {ACM}, 2023.

\bibitem[AR21]{ashlagi2021kidney}
Itai Ashlagi and Alvin~E Roth.
\newblock Kidney exchange: An operations perspective.
\newblock {\em Management Science}, 67(9):5455--5478, 2021.

\bibitem[AS22]{aouad2022dynamic}
Ali Aouad and {\"O}mer Sar{\i}ta{\c{c}}.
\newblock Dynamic stochastic matching under limited time.
\newblock {\em Operations Research}, 70(4):2349--2383, 2022.

\bibitem[BC22]{blanc2022multiway}
Guy Blanc and Moses Charikar.
\newblock Multiway online correlated selection.
\newblock In {\em 2021 IEEE 62nd Annual Symposium on Foundations of Computer
  Science (FOCS)}, pages 1277--1284. IEEE, 2022.

\bibitem[BDML22]{braverman2022max}
Mark Braverman, Mahsa Derakhshan, and Antonio Molina~Lovett.
\newblock Max-weight online stochastic matching: Improved approximations
  against the online benchmark.
\newblock In {\em Proceedings of the 23rd ACM Conference on Economics and
  Computation}, pages 967--985, 2022.

\bibitem[BDP{\etalchar{+}}24]{braverman2024new}
Mark Braverman, Mahsa Derakhshan, Tristan Pollner, Amin Saberi, and David Wajc.
\newblock New philosopher inequalities for online bayesian matching, via
  pivotal sampling.
\newblock {\em arXiv preprint arXiv:2407.15285}, 2024.

\bibitem[Ber12]{bertsekas2012dynamic}
Dimitri Bertsekas.
\newblock {\em Dynamic programming and optimal control: Volume I}, volume~4.
\newblock Athena scientific, 2012.

\bibitem[BGM13]{buvsic2013stability}
Ana Bu{\v{s}}i{\'c}, Varun Gupta, and Jean Mairesse.
\newblock Stability of the bipartite matching model.
\newblock {\em Advances in Applied Probability}, 45(2):351--378, 2013.

\bibitem[BH06]{burchard2006rearrangement}
Almut Burchard and Hichem Hajaiej.
\newblock Rearrangement inequalities for functionals with monotone integrands.
\newblock {\em Journal of Functional Analysis}, 233(2):561--582, 2006.

\bibitem[BL94]{brandt1994pathwise}
Andreas Brandt and G{\"u}nter Last.
\newblock On the pathwise comparison of jump processes driven by stochastic
  intensities.
\newblock {\em Mathematische Nachrichten}, 167(1):21--42, 1994.

\bibitem[CDB22]{cadas2022analysis}
Arnaud Cadas, Josu Doncel, and Ana Bu{\v{s}}i{\'c}.
\newblock Analysis of an optimal policy in dynamic bipartite matching models.
\newblock {\em Performance Evaluation}, 154:102286, 2022.

\bibitem[CHMS10]{chawla2010multi}
Shuchi Chawla, Jason~D Hartline, David~L Malec, and Balasubramanian Sivan.
\newblock Multi-parameter mechanism design and sequential posted pricing.
\newblock In {\em Proceedings of the 42nd Annual ACM Symposium on Theory of
  Computing (STOC)}, pages 311--320, 2010.

\bibitem[CHS24]{chen2024stochastic}
Ziyun Chen, Zhiyi Huang, and Enze Sun.
\newblock Stochastic online correlated selection.
\newblock {\em arXiv preprint arXiv:2408.12524}, 2024.

\bibitem[CILB{\etalchar{+}}20]{collina2020dynamic}
Natalie Collina, Nicole Immorlica, Kevin Leyton-Brown, Brendan Lucier, and Neil
  Newman.
\newblock Dynamic weighted matching with heterogeneous arrival and departure
  rates.
\newblock In {\em Web and Internet Economics: 16th International Conference,
  WINE 2020, Beijing, China, December 7--11, 2020, Proceedings 16}, pages
  17--30. Springer, 2020.

\bibitem[DPS12]{dickerson2012dynamic}
John Dickerson, Ariel Procaccia, and Tuomas Sandholm.
\newblock Dynamic matching via weighted myopia with application to kidney
  exchange.
\newblock In {\em Proceedings of the AAAI Conference on Artificial
  Intelligence}, volume~26, pages 1340--1346, 2012.

\bibitem[DSSX21]{dickerson2021reusable}
John~P. Dickerson, Karthik~A. Sankararaman, Aravind Srinivasan, and Pan Xu.
\newblock Allocation problems in ride-sharing platforms: Online matching with
  offline reusable resources.
\newblock {\em ACM Trans. Econ. Comput.}, 9(3), June 2021.

\bibitem[EFGT22]{ezra2022prophets}
Tomer Ezra, Michal Feldman, Nick Gravin, and Zhihao~Gavin Tang.
\newblock Prophet matching with general arrivals.
\newblock {\em Mathematics of Operations Research}, 47(2):878--898, 2022.

\bibitem[EIVR23]{echenique2023online}
F.~Echenique, N.~Immorlica, V.V. Vazirani, and A.E. Roth.
\newblock {\em Online and Matching-Based Market Design}.
\newblock Cambridge University Press, 2023.

\bibitem[EKLS21]{eckl2021stronger}
Alexander Eckl, Anja Kirschbaum, Marilena Leichter, and Kevin Schewior.
\newblock A stronger impossibility for fully online matching.
\newblock {\em Operations Research Letters}, 49(5):802--808, 2021.

\bibitem[FGL15]{feldman2015combinatorial}
Michal Feldman, Nick Gravin, and Brendan Lucier.
\newblock Combinatorial auctions via posted prices.
\newblock In {\em Proceedings of the 26th Annual ACM-SIAM Symposium on Discrete
  Algorithms (SODA)}, pages 123--135, 2015.

\bibitem[FHTZ20]{fahrbach2020edge}
Matthew Fahrbach, Zhiyi Huang, Runzhou Tao, and Morteza Zadimoghaddam.
\newblock Edge-weighted online bipartite matching.
\newblock In {\em Proceedings of the 61st Symposium on Foundations of Computer
  Science (FOCS)}, pages 412--423, 2020.

\bibitem[FMMM09]{feldman2009online}
Jon Feldman, Aranyak Mehta, Vahab Mirrokni, and S~Muthukrishnan.
\newblock Online stochastic matching: Beating 1-1/e.
\newblock In {\em Proceedings of the 50th Symposium on Foundations of Computer
  Science (FOCS)}, pages 117--126, 2009.

\bibitem[GHH{\etalchar{+}}21]{gao2021improved}
Ruiquan Gao, Zhongtian He, Zhiyi Huang, Zipei Nie, Bijun Yuan, and Yan Zhong.
\newblock Improved online correlated selection.
\newblock In {\em Proceedings of the 62nd Symposium on Foundations of Computer
  Science (FOCS)}, 2021.
\newblock To Appear.

\bibitem[GKPS06]{gandhi2006dependent}
Rajiv Gandhi, Samir Khuller, Srinivasan Parthasarathy, and Aravind Srinivasan.
\newblock Dependent rounding and its applications to approximation algorithms.
\newblock {\em Journal of the ACM (JACM)}, 53(3):324--360, 2006.

\bibitem[GKS19]{gamlath2019beating}
Buddhima Gamlath, Sagar Kale, and Ola Svensson.
\newblock Beating greedy for stochastic bipartite matching.
\newblock In {\em Proceedings of the Thirtieth Annual ACM-SIAM Symposium on
  Discrete Algorithms}, pages 2841--2854. SIAM, 2019.

\bibitem[GU23]{goyal2023online}
Vineet Goyal and Rajan Udwani.
\newblock Online matching with stochastic rewards: Optimal competitive ratio
  via path-based formulation.
\newblock {\em Operations Research}, 71(2):563--580, 2023.

\bibitem[HJS{\etalchar{+}}23]{huang2023onlinestochastic}
Zhiyi Huang, Hanrui Jiang, Aocheng Shen, Junkai Song, Zhiang Wu, and Qiankun
  Zhang.
\newblock Online matching with stochastic rewards: Advanced analyses using
  configuration linear programs.
\newblock In Jugal Garg, Max Klimm, and Yuqing Kong, editors, {\em Web and
  Internet Economics - 19th International Conference, {WINE} 2023, Shanghai,
  China, December 4-8, 2023, Proceedings}, volume 14413 of {\em Lecture Notes
  in Computer Science}, pages 384--401. Springer, 2023.

\bibitem[HKT{\etalchar{+}}18]{huang2018match}
Zhiyi Huang, Ning Kang, Zhihao~Gavin Tang, Xiaowei Wu, Yuhao Zhang, and Xue
  Zhu.
\newblock How to match when all vertices arrive online.
\newblock In {\em Proceedings of the 50th Annual ACM Symposium on Theory of
  Computing (STOC)}, pages 17--29, 2018.

\bibitem[HPT{\etalchar{+}}19]{huang2019tight}
Zhiyi Huang, Binghui Peng, Zhihao~Gavin Tang, Runzhou Tao, Xiaowei Wu, and
  Yuhao Zhang.
\newblock Tight competitive ratios of classic matching algorithms in the fully
  online model.
\newblock In {\em Proceedings of the 30th Annual ACM-SIAM Symposium on Discrete
  Algorithms (SODA)}, pages 2875--2886, 2019.

\bibitem[HS21]{huang2021online}
Zhiyi Huang and Xinkai Shu.
\newblock Online stochastic matching, poisson arrivals, and the natural linear
  program.
\newblock In Samir Khuller and Virginia~Vassilevska Williams, editors, {\em
  Proceedings of the 53rd Annual ACM SIGACT Symposium on Theory of Computing},
  pages 682--693. {ACM}, 2021.

\bibitem[HSY22]{huang2022power}
Zhiyi Huang, Xinkai Shu, and Shuyi Yan.
\newblock The power of multiple choices in online stochastic matching.
\newblock In Stefano Leonardi and Anupam Gupta, editors, {\em {STOC} '22: 54th
  Annual {ACM} {SIGACT} Symposium on Theory of Computing, Rome, Italy, June 20
  - 24, 2022}, pages 91--103. {ACM}, 2022.

\bibitem[HTW24]{huang2024online}
Zhiyi Huang, Zhihao~Gavin Tang, and David Wajc.
\newblock Online matching: A brief survey.
\newblock {\em ACM SIGecom Exchanges}, 22(1):135--158, 2024.

\bibitem[HTWZ20]{huang2020fully}
Zhiyi Huang, Zhihao~Gavin Tang, Xiaowei Wu, and Yuhao Zhang.
\newblock Fully online matching ii: Beating ranking and water-filling.
\newblock In {\em 2020 IEEE 61st Annual Symposium on Foundations of Computer
  Science (FOCS)}, pages 1380--1391. IEEE, 2020.

\bibitem[Jan18]{janson2018tail}
Svante Janson.
\newblock Tail bounds for sums of geometric and exponential variables.
\newblock {\em Statistics \& Probability Letters}, 135:1--6, 2018.

\bibitem[JL13]{jaillet2013online}
Patrick Jaillet and Xin Lu.
\newblock Online stochastic matching: New algorithms with better bounds.
\newblock {\em Mathematics of Operations Research}, 2013.

\bibitem[JM22]{jin2022online}
Billy Jin and Will Ma.
\newblock Online bipartite matching with advice: Tight robustness-consistency
  tradeoffs for the two-stage model.
\newblock In Sanmi Koyejo, S.~Mohamed, A.~Agarwal, Danielle Belgrave, K.~Cho,
  and A.~Oh, editors, {\em Advances in Neural Information Processing Systems
  35: Annual Conference on Neural Information Processing Systems 2022, NeurIPS
  2022, New Orleans, LA, USA, November 28 - December 9, 2022}, 2022.

\bibitem[KKO77]{kamae1977stochastic}
Teturo Kamae, Ulrich Krengel, and George~L O'Brien.
\newblock Stochastic inequalities on partially ordered spaces.
\newblock {\em The Annals of Probability}, 5(6):899--912, 1977.

\bibitem[KSSW22]{kessel2022stationary}
Kristen Kessel, Ali Shameli, Amin Saberi, and David Wajc.
\newblock The stationary prophet inequality problem.
\newblock In {\em Proceedings of the 23rd ACM Conference on Economics and
  Computation}, pages 243--244, 2022.

\bibitem[KVV90]{karp1990optimal}
Richard~M Karp, Umesh~V Vazirani, and Vijay~V Vazirani.
\newblock An optimal algorithm for on-line bipartite matching.
\newblock In {\em Proceedings of the 22nd Annual ACM Symposium on Theory of
  Computing (STOC)}, pages 352--358, 1990.

\bibitem[KW19]{kleinberg2019matroid}
Robert Kleinberg and S~Matthew Weinberg.
\newblock Matroid prophet inequalities and applications to multi-dimensional
  mechanism design.
\newblock {\em Games and Economic Behavior}, 113:97--115, 2019.

\bibitem[LMS00]{lopez2000stochastic}
F~Javier L{\'o}pez, Servet Mart{\'\i}nez, and Gerardo Sanz.
\newblock Stochastic domination and markovian couplings.
\newblock {\em Advances in Applied Probability}, 32(4):1064--1076, 2000.

\bibitem[LWY23]{li2023fully}
Zihao Li, Hao Wang, and Zhenzhen Yan.
\newblock Fully online matching with stochastic arrivals and departures.
\newblock In {\em Proceedings of the AAAI Conference on Artificial
  Intelligence}, volume~37, pages 12014--12021, 2023.

\bibitem[MBM21]{moyal2021product}
Pascal Moyal, Ana Bu{\v{s}}i{\'c}, and Jean Mairesse.
\newblock A product form for the general stochastic matching model.
\newblock {\em Journal of Applied Probability}, 58(2):449--468, 2021.

\bibitem[MGS12]{manshadi2012online}
Vahideh~H Manshadi, Shayan~Oveis Gharan, and Amin Saberi.
\newblock Online stochastic matching: Online actions based on offline
  statistics.
\newblock {\em Mathematics of Operations Research}, 37(4):559--573, 2012.

\bibitem[MM16]{mairesse2016stability}
Jean Mairesse and Pascal Moyal.
\newblock Stability of the stochastic matching model.
\newblock {\em Journal of Applied Probability}, 53(4):1064--1077, 2016.

\bibitem[MP12]{mehta2012online}
Aranyak Mehta and Debmalya Panigrahi.
\newblock Online matching with stochastic rewards.
\newblock In {\em 53rd Annual {IEEE} Symposium on Foundations of Computer
  Science, {FOCS} 2012, New Brunswick, NJ, USA, October 20-23, 2012}, pages
  728--737. {IEEE} Computer Society, 2012.

\bibitem[MSVV07]{mehta2007adwords}
Aranyak Mehta, Amin Saberi, Umesh Vazirani, and Vijay Vazirani.
\newblock Adwords and generalized online matching.
\newblock {\em Journal of the ACM (JACM)}, 54(5):22, 2007.

\bibitem[NSW23]{naor2023dependentroundingarxiv}
Joseph Naor, Aravind Srinivasan, and David Wajc.
\newblock Online dependent rounding schemes.
\newblock {\em CoRR}, abs/2301.08680, 2023.

\bibitem[{\"O}W20]{ozkan2020dynamic}
Erhun {\"O}zkan and Amy~R Ward.
\newblock Dynamic matching for real-time ride sharing.
\newblock {\em Stochastic Systems}, 10(1):29--70, 2020.

\bibitem[PPSW23]{papadimitriou2023online}
Christos Papadimitriou, Tristan Pollner, Amin Saberi, and David Wajc.
\newblock Online stochastic max-weight bipartite matching: Beyond prophet
  inequalities.
\newblock {\em Mathematics of Operations Research}, 2023.

\bibitem[Put14]{puterman2014markov}
Martin~L Puterman.
\newblock {\em Markov decision processes: discrete stochastic dynamic
  programming}.
\newblock John Wiley \& Sons, 2014.

\bibitem[PW24]{patel2024combinatorial}
Neel Patel and David Wajc.
\newblock Combinatorial stationary prophet inequalities.
\newblock In {\em Proceedings of the 2024 Annual ACM-SIAM Symposium on Discrete
  Algorithms (SODA)}, pages 4605--4630. SIAM, 2024.

\bibitem[Sri01]{srinivasan2001distributions}
Aravind Srinivasan.
\newblock Distributions on level-sets with applications to approximation
  algorithms.
\newblock In {\em Proceedings of the 42nd Symposium on Foundations of Computer
  Science (FOCS)}, pages 588--597, 2001.

\bibitem[SS07]{shaked2007stochastic}
Moshe Shaked and J~George Shanthikumar.
\newblock {\em Stochastic orders}.
\newblock Springer, 2007.

\bibitem[SW21]{saberi2021greedy}
Amin Saberi and David Wajc.
\newblock The greedy algorithm is not optimal for online edge coloring.
\newblock In {\em 48th International Colloquium on Automata, Languages, and
  Programming (ICALP 2021)}, pages 109:1--109:18, 2021.

\bibitem[TWW22]{tang2022fractional}
Zhihao~Gavin Tang, Jinzhao Wu, and Hongxun Wu.
\newblock (fractional) online stochastic matching via fine-grained offline
  statistics.
\newblock In {\em Proceedings of the 54th Annual ACM Symposium on Theory of
  Computing (STOC)}, pages 77--90, 2022.

\bibitem[TZ22]{tang2022improved}
Zhihao~Gavin Tang and Yuhao Zhang.
\newblock Improved bounds for fractional online matching problems.
\newblock {\em arXiv preprint arXiv:2202.02948}, 2022.

\bibitem[VM21]{varma2021transportation}
Sushil~Mahavir Varma and Siva~Theja Maguluri.
\newblock Transportation polytope and its applications in parallel server
  systems.
\newblock {\em arXiv preprint arXiv:2108.13167}, 2021.

\bibitem[Wol82]{wolff1982poisson}
Ronald~W Wolff.
\newblock Poisson arrivals see time averages.
\newblock {\em Operations research}, 30(2):223--231, 1982.

\bibitem[WW15]{wang2015two}
Yajun Wang and Sam Chiu-wai Wong.
\newblock Two-sided online bipartite matching and vertex cover: Beating the
  greedy algorithm.
\newblock In {\em Proceedings of the 42nd International Colloquium on Automata,
  Languages and Programming (ICALP)}, pages 1070--1081, 2015.

\end{thebibliography}

\end{document}